\documentclass[11pt]{article}
\usepackage{amsmath,amsthm}
\usepackage{amssymb}
\usepackage{graphicx}

\newcommand{\pa}{\partial}

\newcommand{\be}{\begin{equation}}
\newcommand{\ee}{\end{equation}}

\usepackage{color}

\newcommand{\Ik}{\mathcal{I}_k}

\newtheorem{theorem}{Theorem}[section]
\newtheorem{corollary}{Corollary}[section]

\begin{document}

\title{Turing conditions for pattern forming systems on evolving manifolds}

\author{Robert A. Van Gorder\thanks{Department of Mathematics and Statistics, University of Otago, P.O. Box 56, Dunedin 9054, New Zealand ({rvangorder@maths.otago.ac.nz})} \and V{\'a}clav Klika\thanks{Department of Mathematics FNSPE, Czech Technical University in Prague, Trojanova 13, 12000 Prague, Czech Republic} \and Andrew L. Krause\thanks{Mathematical Institute, University of Oxford, Andrew Wiles Building, Radcliffe Observatory Quarter, Woodstock Road, Oxford, OX2 6GG, UK}}
\maketitle

\begin{abstract}
The study of pattern-forming instabilities in reaction-diffusion systems on growing or otherwise time-dependent domains arises in a variety of settings, including applications in developmental biology, spatial ecology, and experimental chemistry. Analyzing such instabilities is complicated, as there is a strong dependence of any spatially homogeneous base states on time, and the resulting structure of the linearized perturbations used to determine the onset of instability is inherently non-autonomous. We obtain general conditions for the onset and structure of diffusion driven instabilities in reaction-diffusion systems on domains which evolve in time, in terms of the time-evolution of the Laplace-Beltrami spectrum for the domain and functions which specify the domain evolution. Our results give sufficient conditions for diffusive instabilities phrased in terms of differential inequalities which are both versatile and straightforward to implement, despite the generality of the studied problem. These conditions generalize a large number of results known in the literature, such as the algebraic inequalities commonly used as a sufficient criterion for the Turing instability on static domains, and approximate asymptotic results valid for specific types of growth, or specific domains. We demonstrate our general Turing conditions on a variety of domains with different evolution laws, and in particular show how insight can be gained even when the domain changes rapidly in time, or when the homogeneous state is oscillatory, such as in the case of Turing-Hopf instabilities. Extensions to higher-order spatial systems are also included as a way of demonstrating the generality of the approach.\\
\\
keywords: pattern formation, Turing instability, evolving spatial domains, reaction-diffusion\\
\\
AMS Subject Classifications: 35B36, 92C15, 70K50, 58C40, 58J32
\end{abstract}

\section{Introduction}
Since Turing first proposed reaction--diffusion systems as a model for pattern formation \cite{Turing1952} much work has been done to understand the theoretical, as well as chemical and biological aspects of this mechanism \cite{baker2008partial, marcon2012turing}. Many authors have elucidated the importance of domain size and shape on the formation of patterns, and the impact of geometry on the kinds of admissible patterns that can arise due to a Turing instability \cite{Murray,seul1995domain}. While Turing's original reaction--diffusion theory postulated the existence of a pre-pattern before growth occurs, in the past few decades biological and theoretical evidence has suggested that growth of the patterning field itself influences the pattern forming potential of a system, and modulates the emergent patterns \cite{binder2008modeling, coen2004genetics, liu2006two, maini2002implications, miura2006mixed, raspopovic2014digit}. Since reaction--diffusion systems are more difficult to analyze on growing domains, pattern formation has been typically considered on different--sized static domains to simulate very slow growth \cite{Varea1997}. This requires the reaction and diffusion of the chemical species to occur on a much faster timescale than the growth \cite{klika2017history, madzvamuse2010stability}, and to be independent of the growth, although this assumption is certainly not valid for all systems in biology and chemistry. 

There have been a variety of studies connecting growth and pattern formation in reaction-diffusion systems. Uniform and isotropic domain growth in one--dimensional reaction--diffusion systems in slow and fast growth regimes was considered by \cite{Crampin}, and frequency-doubling of the emergent Turing patterns was demonstrated. This frequency-doubling was discussed with the aim to help resolve the (lack of) robustness of pattern formation in reaction-diffusion systems \cite{baker2008partial, barrass2006mode}. More recently, it was shown that such frequency-doubling can depend somewhat sensitively on the kind of growth rates involved, even in a 1-D domain growing isotropically \cite{ueda2012mathematical}. Turing and Turing-Hopf instabilities of the FitzHugh-Nagumo system in an exponentially and isotropically growing square were studied in \cite{castillo2016turing}, and this work suggested that anisotropy and curvature are important considerations for extending their analysis. Such instabilities were also studied on isotropically growing spherical and toroidal domains \cite{sanchez2018turing}. A general formulation of reaction--diffusion theory on isotropically evolving one and two--dimensional manifolds was given in \cite{Maini}, with motivation from biological settings where growth and curvature both play a role in organism development. Corrections to the classical conditions for Turing instabilities in the case of slow isotropic growth were derived by \cite{madzvamuse2010stability}, while \cite{hetzer2012characterization} considered a type of quasi-asymptotic stability, although large deviations from an approximately homogeneous state at finite time were not considered. In contrast, \cite{klika2017history} considered domain growth based on Lyapunov stability, which captures large deviations from the reference state at finite time before growth saturates, thereby capturing some of the history dependence inherent in the growth dynamics. The study of \cite{klika2017history} was able to relax a number of restrictive assumptions made in \cite{hetzer2012characterization} and \cite{madzvamuse2010stability}, although the results were obtained for fairly specific special cases. Beyond computing linear instability criteria, \cite{comanici2008patterns} analytically explored how patterns change and evolve under growth by exploiting the framework of amplitude equations. While analytical results on mode competition and selection can be valuable, these are often extremely limited as they only apply near the bifurcation boundary in the parameter space, and they become computationally intractable in many cases of interest \cite{krechetnikov2017stability}. While systems with time-dependent diffusion coefficients have been studied via asymptotic and Floquet-theoretic methods \cite{mendez2010reaction}, such results break down for domain evolution due to the dynamic nature of the base state.

Most of the above models of reaction--diffusion systems on growing domains only analyzed the case of isotropic (or apical) growth, which is unable to recapitulate the full range of complex biological structures found in developing organisms \cite{corson2009turning, peaucelle2015control, striegel2009chemically, ubeda2008root}. Investigating arbitrary anisotropic growth in the context of biological patterning is a natural extension to reaction--diffusion theories of pattern formation, and has been considered in biomechanical models of growth across a range of tissues and organisms \cite{amar2013anisotropic, bittig2008dynamics, menzel2005modelling, saez2007rigidity}. Additionally, contraction and other complex tissue movements have been observed in embryogenesis, suggesting the need to further generalize models of growth and domain restructuring over time beyond monotonically increasing domains \cite{amar2013anisotropic, nechaeva2002rhythmic, toyama2008apoptotic, wang2015three}. The influence of non--uniform domain growth on one--dimensional reaction--diffusion systems, including apical or boundary growth, was considered in \cite{crampin2002pattern}, while \cite{rossi2016control} studied concentration-dependent growth of a scalar reaction-diffusion equation on a time--dependent manifold. Anisotropic growth, consisting of independent dilations of an underlying manifold in each orthogonal Euclidean coordinate, was recently studied in \cite{krause2018influence}. Some experimental models of reaction--diffusion processes on growing domains (using, for instance, photosensitive reactions) have been explored, although these typically involve apical or otherwise spatially-dependent forms of growth \cite{miguez2006effect,konow2019turing}. Recent experiments have also explored hydrodynamic instabilities in time-dependent domains, finding important impacts of the growth on the nature of such instabilities and subsequent pattern evolution \cite{ghadiri_krechetnikov_2019}.

Difficulties arise when the local rate of volume expansion or contraction depends on the spatial coordinates (or more generally, the morphogen concentrations themselves), which induces an advection term from mass conservation of the respective chemical species. Such systems are exceptionally difficult to analyze due to spatial heterogeneity in diffusion and advection, in addition to their non-autonomous nature. Still, as we shall later discuss, some forms of anisotropic growth permit volume expansion or contraction which is global, depending only on time and not on the spatial coordinates. It is this class of growing domains for which we provide a general method to compute the instability of a spatially homogeneous equilibrium. This generalizes much of the above literature, and provides a way to compute a time-dependent instability criterion for a large class of growth functions for reaction--diffusion systems posed on smooth, compact, and simply connected, but otherwise arbitrary, manifolds.

The remainder of this paper is organized as follows. In Sec.~\ref{sec2} we outline the derivation of a general reaction-diffusion model on time-evolving domains. We give the precise mathematical formulation of component-wise dilational growth considered throughout the paper, and outline the general spectral problem. We also discuss several difficulties in the analysis of such systems, necessitating the need for new approaches from those often employed in the literature. In Sec.~\ref{sec3}, we present the main theoretical results for systems of two reaction-diffusion equations on evolving domains. After first obtaining a general instability result for second-order ODE through a comparison principle, we derive an instability criterion which generalizes the Turing conditions for diffusive instabilities to the non-autonomous case. We also discuss various reductions of these conditions, highlighting that they collapse to the standard Turing conditions on static domains in the appropriate limits. In addition to systems of reaction-diffusion equations, we also discuss the extension of our results to systems with higher-order space derivatives which are also heavily studied in the pattern formation literature, and capture non-local effects in a variety of models. In Sec. \ref{sec4}-\ref{shsec} we provide a number of applications of the theory, and compare our results with direct numerical simulations of various reaction diffusion systems on growing domains in one, two, and three spatial dimensions.  We generalize several examples from the literature without restriction to asymptotic regimes, and consider novel classes of domain evolution which have not yet been considered. In particular, applications of our approach for evolving one-dimensional intervals are described in Sec. \ref{sec4} (including the situation where the evolution is non-monotone), whereas the more complicated applications to evolving manifolds in two or more dimensions are given in Sec. \ref{mfds}. In addition to growing domains, our approach allows us to consider domains which evolve yet preserve area or volume (as discussed in Sec. \ref{sec44}), which has seemingly not been considered previously. We also give some examples related to higher-order systems on evolving manifolds in Sec. \ref{shsec}. We finally discuss and summarize our findings in Sec.~\ref{sec5}.

\section{General model and diffusive instability framework}\label{sec2}

\subsection{Reaction-diffusion systems on evolving domains}\label{formulation}
We consider a manifold $\Omega(t)\subseteq \mathbb{R}^N$ which grows in a dilational manner along each Euclidean coordinate axis. We also assume that the domain $\Omega (t)$ is compact, simply connected, smooth, and Riemannian, in order to ensure that the spectrum of the Laplace-Beltrami operator is discrete and non-negative for all time. Concentrations on manifolds with boundary will be subject to Neumann data at the boundary. We shall restrict our attention to growth for which the time evolution of $\Omega(t)$ results in spatially homogeneous volume expansion or contraction, and shall make this statement more precise later. The case of locally varying volume expansion or contraction results in strongly non-uniform growth which we do not consider.
 
Let $\hat{\Omega}(t)$ be a volume element of the manifold, such that $\hat{\Omega}(t) \subset \Omega(t)$.
Let $\mathbf{u}(\mathbf{X},t)$, $\mathbf{u}: \Omega(t)\times [0,\infty) \rightarrow \mathbb{R}^n$ be a concentration function defined on the manifold $\Omega(t)$. Here $\mathbf{u}$ may describe the concentration of $n\geq 2$ chemical species, or morphogens, on the manifold $\Omega (t)$, though other interpretations such as cells or effective genetic circuits use the same mathematical formulations \cite{kondo2010reaction}. We assume that $\mathbf{u}$ is $C^1([0,\infty))$ in time and $C^2(\Omega(t))$ in the spatial coordinates. We note that formalizing this space of functions is easier to do after mapping back to a static domain, which we will also do for analytic and numerical convenience shortly.

The conservation of mass equation reads
\begin{equation}
\label{CoM_eqn}
\frac{d}{dt} \int_{\hat{\Omega}(t)} \mathbf{u} \, d\Omega = \int_{\hat{\Omega}(t)} \left( - \nabla \cdot  \mathbf{j} + \mathbf{f}(\mathbf{u}) \right) d\Omega\,,
\end{equation}
where $\mathbf{j}$ denotes the fluxes of concentrations $\mathbf{u}$, $\mathbf{f}\in C^1(\mathbb{R}^n)$ is the function denoting the reaction kinetics, and $d\Omega$ is the local volume element on the manifold.
Using Reynold's transport theorem on the left hand side of Eq. \eqref{CoM_eqn}, we have
\begin{equation}
\label{RTT_eqn}
\frac{d}{dt}\int_{\hat{\Omega}(t)} \mathbf{u} \, d\Omega = \int_{\hat{\Omega}(t)} \left( \frac{\partial \mathbf{u}}{\partial t} + \nabla_{\Omega(t)}\cdot (\mathbf{Q}\mathbf{u}) \right) d\Omega\,,
\end{equation}
where $\mathbf{Q}$ is the local velocity vector field generated by changes in the manifold $\Omega(t)$. We denote $\nabla_{\Omega(t)}\cdot$ as the  divergence operator on $\Omega(t)$ and $\nabla_{\Omega(t)}^2$ to be the Laplace--Beltrami operator on $\Omega(t)$.

By applying Eq.~\eqref{RTT_eqn} to Eq.~\eqref{CoM_eqn} and using Fick's law of diffusion, we have the reaction--diffusion--advection system 
\begin{equation} \label{RD_system:eqn}
\frac{\partial \mathbf{u}}{\partial t} + \nabla_{\Omega(t)} \cdot (\mathbf{Q}\mathbf{u}) = D \nabla_{\Omega(t)}^2 \mathbf{u} + \mathbf{f}(\mathbf{u}) .
\end{equation}
Here $D = \text{diag}(d_1,\dots , d_n)$ is the matrix of diffusion parameters. The term $\nabla_{\Omega(t)} \cdot (\mathbf{Q}\mathbf{u})$ can be written as $ \mathbf{Q} \cdot \nabla_{\Omega(t)} \mathbf{u} + \mathbf{u} \cdot \nabla_{\Omega(t)}\mathbf{Q}$. We note that the term $\mathbf{Q} \cdot \nabla_{\Omega(t)} \mathbf{u}$ corresponds to advection due to local growth of the manifold, whereas the $\mathbf{u} \cdot \nabla_{\Omega(t)}\mathbf{Q}$ term corresponds to dilution of the concentrations $\mathbf{u}$ due to local volume changes. If $\Omega(t)$ is a manifold with boundary $\pa \Omega(t)$, we assume no flux conditions $\frac{\pa \mathbf{u}}{\pa \mathbf{n}} = \mathbf{0}$ for $\mathbf{X}\in \pa \Omega(t)$.

\subsection{Domain evolution as dilations in each orthogonal direction}\label{seccurve}
We consider the case where the evolution of the manifold is such that volume expansion or contraction does not vary locally, in other words such that $\nabla_{\Omega(t)}\cdot \mathbf{Q}$ depends strictly on time and never on spatial coordinates. Other kinds of growth, such as apical growth or anisotropic growth of surfaces, may result in spatially dependent volume expansion \cite{krause2018influence}, and while interesting, we do not consider this manner of growth. 

Consider moving coordinates $\mathbf{X}$ written as 
\begin{equation}\label{coords}
\mathbf{X} = \left(r_1(t)\chi_1(\mathbf{x}), \dots , r_N(t)\chi_N(\mathbf{x}) \right),
\end{equation}
for stationary coordinates on the initial manifold $\mathbf{x}=(x_1,\dots ,x_N)\in \Omega^*=\Omega(t=0)$. In addition to covering all cases of general dilational growth (where the dilation may be different along each coordinate), this assumption will ensure that the metric tensor for these coordinates, $G$, will have the property that $\det G$ is multiplicatively separable in time and space. Here each coordinate again has an independent dilation function $r_j(t)$, though each $X_j$ depends possibly on multiple stationary coordinates. When $r_j(t)=r(t)$ for all $j=1,\dots ,N$, we have isotropic evolution of the manifold $\Omega(t)$. When at least two $r_j(t)$'s differ, then we have anisotropic evolution which is still dilational in the individual orthogonal Cartesian directions.

In order to remove the advection term induced by the growing manifold, $\nabla_{\Omega(t)} \cdot (\mathbf{Q}\mathbf{u})$,  we apply a change of variables to a moving coordinate system. As the space and time variables in $\mathbf{X}$ are separable, and noting $\nabla_{\Omega(t)} \cdot (\mathbf{Q}\mathbf{u}) = (\nabla_{\Omega(t)}\cdot \mathbf{Q})\mathbf{u} + \mathbf{Q}\cdot (\nabla_{\Omega(t)}\mathbf{u})$, 
the change of variable will contribute a term $-\mathbf{Q}\cdot (\nabla_{\Omega(t)}\mathbf{u})$, canceling the latter term. After the coordinate change, we will have a contributing advection term of the form $(\nabla_{\Omega(t)}\cdot \mathbf{Q})\mathbf{u}$.
We find \cite{krause2018influence}
\begin{equation} \label{adgrowth:eqn}
\nabla_{\Omega(t)}\cdot \mathbf{Q} = \frac{\partial}{\partial t}\left( \log\left( |\det {G}|^{\frac{1}{2}}\right) \right) \,.
\end{equation}

From \eqref{adgrowth:eqn}, we see that the manner of growth for which volume expansion or contraction is spatially homogeneous is equivalent to considering a coordinate chart such that the time derivative of $\log (|\det G|)$ is independent of space, i.e.~a coordinate chart for which $\det G$ is multiplicatively separable in space and time variables. Considering only such moving coordinates \eqref{coords}, we then have that \eqref{RD_system:eqn} becomes \cite{krause2018influence}
\begin{equation} \label{growthgen}
\frac{\partial \mathbf{u}}{\partial t} =  \frac{D}{|\det G|^{\frac{1}{2}}}\sum_{i,j=1}^N \frac{\pa}{\pa x_i}\left( |\det G|^{\frac{1}{2}}G_{ij}^{-1}\frac{\pa \mathbf{u}}{\pa x_j}\right)  - \frac{\partial}{\partial t}\left( \log\left( |\det {G}|^{\frac{1}{2}}\right) \right)\mathbf{u}+ \mathbf{f}(\mathbf{u}) .
\end{equation}
The Laplace-Beltrami operator on the fixed reference manifold $\Omega^*$ is time-dependent, as the coordinates \eqref{coords} depend explicitly on time.

\subsection{General linear instability analysis} \label{sec.GenLinStabAnal}
We now consider diffusion-driven instabilities arising from systems of the form \eqref{growthgen}. Consider first the eigenvalue problem 
\begin{equation}\label{evp1}
\nabla_{\Omega(t)}^2 \Psi = - \rho \Psi ,
\end{equation}
which is held subject to $\frac{\pa \Psi}{\pa \mathbf{n}} = 0$ for $\mathbf{X} \in \pa \Omega(t)$ when $\Omega(t)$ has boundary. From the assumptions made earlier on $\Omega(t)$, for any fixed time $t\geq 0$, we have that a non-negative spectrum $\rho_k(t)$ exists, where $0=\rho_0(t) < \rho_1(t) \leq \rho_2(t) \leq \cdots \rightarrow \infty$. In general, for manifolds of dimension greater than one, $k$ denotes a multi-index. As the growth functions are assumed smooth, and $\Omega(t)$ is assumed a simply connected Riemannian manifold with smooth boundary for all $t\geq 0$, then we shall assume that $\Omega(t)$ is such that the spectrum can be continued as a smooth function of time, with $\rho_k(t)> 0$ for all $k\geq 1$. For our purposes, we assume any given $\Omega(t)$ permits such a construction, as we are concerned with dynamics on a prescribed $\Omega(t)$. We denote the corresponding eigenfunctions by $\Psi_k(\mathbf{X})$. Constructing such eigenvalues and eigenfunctions can be very difficult, and although our results only require existence rather than explicit construction, we will give examples later for domains where such constructions can be carried out. 

If we carry out the change of coordinates \eqref{coords}, and note that the stationary form of each eigenfunction is
\begin{equation}\label{efxn}
\psi_k(\mathbf{x}) = \Psi_k\left( \frac{X_1}{r_1(t)}, \dots, \frac{X_N}{r_N(t)}\right) = \Psi_k\left( \chi_1(\mathbf{x}),\dots, \chi_N(\mathbf{x})\right), 
\end{equation}
the eigenvalue problem \eqref{evp1} is put into the form
\begin{equation}\label{evp2}
\frac{1}{|\det G|^{\frac{1}{2}}}\sum_{i,j=1}^N \frac{\pa}{\pa x_i}\left( |\det G|^{\frac{1}{2}}G_{ij}^{-1}\frac{\pa \psi_k(\mathbf{x})}{\pa x_j}\right)  = - \rho_k(t) \psi_k(\mathbf{x})\,.
\end{equation} 
In the special case where domain evolution is isotropic, that is $r_1(t)=\cdots =r_N(t) = r(t)$, the Laplace-Beltrami operator simplifies so that we have 
$\frac{1}{r(t)^2}\nabla_{\Omega(0)}^2 \psi_k(\mathbf{x})  = - \frac{\rho_k(0)}{r(t)^2} \psi_k(\mathbf{x})$, and hence $\rho_k(t) = \frac{\rho_k(0)}{r(t)^2}$. 
This is of course not true in general, and for more complicated growth finding $\rho_k(t)$ can be involved. However, under growth of the form \eqref{coords}, each eigenfunction is stationary, and each eigenvalue is a function of time as smooth as the dilations $r_i(t)$, granting existence of the $\rho_k(t)$. 

From the form of growth assumed, we have that volume expansion is not dependent on space, so we can write 
\begin{equation}\label{vol}
\frac{\dot{\mu}(t)}{\mu(t)} = \frac{\partial}{\partial t}\left( \log\left( |\det {G}|^{1/2}\right) \right)
\end{equation}
for some function $\mu(t)$. 
As $|\det {G}| = |\mathcal{G}_1(t)\mathcal{G}_2(\mathbf{x})|$, then $\mu(t) = |\mathcal{G}_1(t)|^{1/2}$. 

A spatially uniform solution to \eqref{growthgen}, $\mathbf{u}(\mathbf{x},t) = \mathbf{U}(t)$, is governed by the equation
\begin{equation} \label{uniform}
\frac{d \mathbf{U}}{d t} =  - \frac{\dot{\mu}(t)}{\mu(t)}\mathbf{U}+ \mathbf{f}(\mathbf{U}) , \quad \mathbf{U}(0) = \mathbf{U}^*,
\end{equation}
where we choose $\mathbf{U}^*$ to satisfy $\mathbf{f}(\mathbf{U}^*)=\mathbf{0}$. We choose this initial condition so that the dynamics will initially agree with those of a time-independent steady state in the absence of growth. In this way, if there is no growth (or, more generally, when there is no net volume expansion or contraction so that $\dot{\mu} \equiv 0$ for all $t\geq 0$), then the exact solution to \eqref{uniform} is $\mathbf{U}(t) \equiv \mathbf{U}^*$, which is what is assumed when deriving the standard Turing conditions on a static domain. Therefore, \eqref{uniform} generalizes the static uniform base state to account for dilution due to growth.\footnote{We note that for complex reaction-diffusion systems spatio-temporal base states consisting of plane waves are also possible \cite{knobloch2014stability}. However, as our concern is with generalising the Turing conditions to account for evolving domains, we only consider spatially uniform base states. Additionally, plane waves do not generalize well to manifolds which are not flat.}

We consider a perturbation of this spatially homogeneous solution in order to determine its stability. Although the solution to \eqref{uniform} may tend to a steady state, this is not required, and for many examples will not occur. We choose a general spatial perturbation of the form 
\begin{equation}\label{pert2}
\mathbf{u}(\mathbf{x},t) = \mathbf{U}(t) + \epsilon \frac{\psi_k(\mathbf{x})}{\mu(t)} \mathbf{V}_k(t),
\end{equation}
where $\psi_k(\mathbf{x})$ is the $k$th scaled eigenfunction in \eqref{efxn} with corresponding Laplace-Beltrami eigenvalue $\rho_k(t)$. For each $k = 1,2,\dots,$ we then have the following linearized problem for the growth or decay of the perturbation
\begin{equation}\label{ode4}
\frac{d\mathbf{V}_k}{dt} = - \rho_k(t) D\mathbf{V}_k + J(\mathbf{U})\mathbf{V}_k,
\end{equation} 
which is an ODE for the unknown function $\mathbf{V}_k(t)$, the long-time asymptotic behavior of which determines the stability or instability of the perturbation \eqref{pert2}. The matrix $J(\mathbf{U})$ denotes the (in general, time-dependent) Jacobian matrix corresponding to the linearization of $\mathbf{f}$ at $\mathbf{u}=\mathbf{U}(t)$.

Equation \eqref{ode4} results in a solution $\mathbf{V}_k(t)$, and we say that a perturbation \eqref{pert2} is asymptotically stable for a given $k\in\mathbb{N}$ provided $\mu(t)^{-1}|\mathbf{V}_k(t)|\rightarrow 0$ as $t\rightarrow \infty$. We say that the perturbation \eqref{pert2} is asymptotically unstable for a given $k\in\mathbb{N}$ provided $\mu(t)^{-1}|\mathbf{V}_k(t)|\rightarrow \infty$ as $t\rightarrow \infty$. If $V_k(t)$ satisfies neither of these, then we might say that the perturbation is neutrally stable or unstable, depending upon the context. This is akin to the classical Turing perturbation for which $\mathbf{V}_k(t) = \mathbf{C} \exp(\lambda(k) t)$, where $\mathbf{C}\in \mathbb{R}^N$ is a constant vector, and $\lambda(k) \in \mathbb{C}$, in which case the perturbation is stable if $\text{Re}(\lambda(k)) <0$ and unstable if $\text{Re}(\lambda(k)) >0$. We will also discuss conditions for \emph{transient} stability or instability, wherein a perturbation may decay or grow for some set of time, as this can generate a pattern in the fully nonlinear setting (though such linear analysis cannot guarantee that an instability leads to a patterned state in finite time periods). As we shall see, transient instabilities play a much larger role in evolving domains than any asymptotic stability criterion.

\subsection{Systems with higher-order spatial derivatives}\label{hoformulation}
There are various applications for which higher-order spatial derivatives are used, with the Swift-Hohenberg equation \cite{swift1977hydrodynamic, cross1993pattern}, Cahn-Hilliard equation \cite{cahn1958free, kielhofer1997pattern, novick1984nonlinear}, and Kuramoto-Sivashinsky \cite{kuramoto1978diffusion, sivashinsky1977nonlinear, facsko2004dissipative, hyman1986kuramoto, rost1995anisotropic} equations some common examples of pattern forming systems with fourth-order spatial derivatives, with related models of even higher-order arising in applications \cite{korzec2008stationary, pawlow2011sixth}. Such higher-order derivatives often represent nonlocal interactions, and this has been extensively applied in biological applications such as cellular signalling and ecological interactions; see \cite{cohen1981generalized, ochoa1984generalized} and Chapter 11 of \cite{MurrayI}. We show here how to carry out similar analysis outlined above for such higher-order equations on evolving spatial domains. Similar to what we have done in Section \ref{formulation}, we may write such higher-order systems in the form 
\begin{equation}\label{higher1}
\frac{\partial \mathbf{u}}{\partial t} + \nabla_{\Omega(t)} \cdot (\mathbf{Q}\mathbf{u}) = D \mathcal{P}\left(\nabla_{\Omega(t)}^2\right) \mathbf{u} + \mathbf{f}(\mathbf{u}),
\end{equation}
where $\mathcal{P}$ is a polynomial of degree $p\geq 1$ satisfying $\mathcal{P}(0)=0$, and the second term involving $\mathbf{Q}$ again arises from a conservation principle on the evolving manifold. On manifolds with boundary, we consider a generalisation of no-flux boundary conditions $\frac{\partial^{2\ell -1} \mathbf{u}}{\partial^{2\ell -1} \mathbf{n}}=\mathbf{0}$ for $\mathbf{X}\in \partial \Omega(t)$, where $\ell = 1,2,\dots, p$. All other quantities are as defined earlier, with the only difference between equation \eqref{higher1} and the earlier discussed equation \eqref{RD_system:eqn} being the more general operator involving the spatial derivatives. Dynamics from the complex Swift-Hohenberg equation, with $u^3$ replaced by $|u|^2u$, were considered on an evolving domain by \cite{knobloch2014stability} and \cite{krechetnikov2017stability}, with the domain being a time-dependent interval. 

Mapping the problem \eqref{higher1} to the stationary frame (assuming that the manifold $\Omega(t)$ obeys all properties required in Section \ref{seccurve}), we find that \eqref{higher1} is put into the form
\begin{equation} \label{higher2}
\frac{\partial \mathbf{u}}{\partial t} =  D \mathcal{P}\left(\frac{1}{|\det G|^{\frac{1}{2}}}\sum_{i,j=1}^N \frac{\pa}{\pa x_i}\left( |\det G|^{\frac{1}{2}}G_{ij}^{-1}\frac{\pa \mathbf{u}}{\pa x_j}\right) \right)  - \frac{\partial}{\partial t}\left( \log\left( |\det {G}|^{\frac{1}{2}}\right) \right)\mathbf{u}+ \mathbf{f}(\mathbf{u}) .
\end{equation}
Note that the equation governing a spatially uniform state is the same as in \eqref{uniform}. 

Regarding the spectral problem, as the operator $\mathcal{P}$ is a polynomial of the Laplace-Beltrami operator, we have 
\begin{equation}\label{higher3}
\mathcal{P}\left(\nabla_{\Omega(t)}^2\right)\psi_k = \sum_{\ell=1}^p \left(\nabla_{\Omega(t)}^2\right)^\ell \psi_k  = \sum_{\ell=1}^p \left(-\rho_k(t)\right)^\ell \psi_k = \mathcal{P}\left(-\rho_k(t)\right)\psi_k
\end{equation}
for any spatial eigenfunction $\psi_k$ satisfying \eqref{evp1} and \eqref{efxn}. Therefore, assuming a spatial perturbation $\psi_k$ involving the $k$th spatial eigenfunction as in \eqref{pert2}, we obtain the problem 
\begin{equation}\label{higher4}
\frac{d\mathbf{V}_k}{dt} = D\mathcal{P} \left( -\rho_k(t) \right) \mathbf{V}_k + J(\mathbf{U})\mathbf{V}_k.
\end{equation} 
The qualitative analysis for \eqref{higher4} is the same as that discussed for \eqref{ode4}. 

Of course, depending upon the form of the boundary conditions, there can be other spatial eigenfunctions for higher-order problems; that is to say, the Laplace-Beltrami eigenfunctions $\psi_k$ described in \eqref{evp1} and \eqref{efxn} are in general a subset of the eigenfunctions possible for a given higher-order eigenvalue problem such as \eqref{higher3}. The study of such cases would involve the classification of eigenfunctions and eigenvalues of the spatial problem $\mathcal{P}\left(\nabla_{\Omega(t)}^2\right)\hat{\psi}_k = -\zeta_k \hat{\psi}_k$, where $\hat{\psi}_k$ need not be an eigenfunction of the standard Laplace-Beltrami problem $\nabla_{\Omega(t)}^2 \psi_k = \rho_k \psi_k$. For the purpose of this paper, we shall primarily consider only the relatively simple case of Laplace-Beltrami eigenfunctions, which will be sufficient for studying the instability problem. Of course, given a desired operator, $\mathcal{P}\left(\nabla_{\Omega(t)}^2\right)$, one may perform similar calculations and scaling to what we do for the standard Laplace-Beltrami case, in order to find the time-dependent spectrum $\zeta_k=\zeta_k(t)$. For most higher-order elliptic operators of practical importance, $\zeta_k$ will be bounded away from $-\infty$ and non-decreasing, $-\infty < \zeta_0 \leq \zeta_1 < \cdots$ with $\zeta_k \rightarrow \infty$ as $k\rightarrow \infty$ \cite{colbois2019neumann}, although for many commonly studied operators on bounded domains the spectrum is non-negative \cite{pleijel1950eigenvalues, levine1985unrestricted, laptev1997dirichlet}. We avoid a discussion on the classification of such problems, noting that for many higher-order problems the behavior of the spectrum appears to be an open problem. Assuming one can find the spectrum $\zeta_k$, we briefly comment on how to make use of this in Section \ref{instabhigher}.

\subsection{Difficulties arising in the study of such systems}\label{difficult}
The study of the asymptotic stability of systems of the form \eqref{ode4} is made difficult for a number of reasons, which we now outline.

The base states governed by \eqref{uniform} depend on the global rate of volume expansion or contraction, and hence are time-varying. Therefore, we expand and linearize the reaction-diffusion system about a spatially uniform yet temporally varying base state, resulting in a non-autonomous Jacobian matrix. This non-autonomy is in addition to the non-autonomy due to the spectrum of the evolving domain, hence the systems for the linearization $\mathbf{V}_k(t)$ given in \eqref{ode4} are non-autonomous in all components rather than just in diagonal components. As the Jacobian $J$ depends on the specific form of the nonlinear reaction kinetics $\mathbf{f}$, non-autonomous entries in $J$ may be non-monotone, even for monotone growth functions. Compared to their autonomous counterparts, there is very little general theory for the dynamics of such systems. 

We note that many works in this area (see, for instance, \cite{madzvamuse2008stability}) will attempt to overcome this complication by assuming a time-independent steady state solution of the ODE system governing the reaction kinetics in the presence of growth. This would be equivalent to obtaining a fixed point of the right hand side of \eqref{uniform}, i.e., to finding an algebraic solution $\mathbf{U}^*$ of the algebraic equation $\frac{\dot{\mu}(t)}{\mu(t)}\mathbf{U} = \mathbf{f}(\mathbf{U})$. There are two problems with this, one regarding feasibility and one more philosophical. Regarding feasibility, a time-independent steady state $\mathbf{U}^*$ exists only if $\frac{\dot{\mu}(t)}{\mu(t)}$ is identically equal to a constant for all time, which is restrictive of the kinds of growth considered. One exception is to consider a state which is identically zero, provided that the reaction kinetics $\mathbf{f}$ permit this. For a zero state, the volume expansion or contraction will still permit a zero state. Of course, this is then fairly restrictive on the form of the reaction kinetics, particularly in light of the fact that for many physical or biological systems, loss of stability of a positive steady state is useful for applications. Unlike what is done in the static-domain case, shifting an equilibrium to the zero state would influence the dynamics due to the dilution term, as \eqref{ode4} would then no longer be a homogeneous system.

In order to remedy this, one may be tempted to instead consider the limit $t \rightarrow \infty$, for which taking either of these quantities to be constant (at least in the case of growth no more rapid than exponential) is seemingly more sensible. This leads to the second, more philosophical, problem. If one is interested in understanding how both growth and diffusion interact to induce the Turing instability and resulting pattern formation, then as pointed out in \cite{klika2017history}, history of the domain growth must factor into the Turing conditions in some manner. If one neglects growth in the base state, then one is arriving at the final spatially uniform state after growth has occurred and effectively obtaining Turing conditions for the final configuration of the problem domain. Depending on the properties of the growth function, mass conservation may result in drastic changes in the spatially uniform state over time, and the changes will become more drastic with an increased number of spatial dimensions. As such, we maintain this dependence on growth in the base states despite the added mathematical difficulties and complications. We shall later show that there is indeed one natural scenario for which the base state can be assumed time-independent, corresponding to domains which evolve in such a way that preserves volume and hence mass. 

Regarding a second difficulty, we remark that eigenvalues are not the appropriate criterion to employ for determining the long-time dynamics of such non-autonomous systems \cite{josic2008unstable, madzvamuse2010stability, mierczynski2017instability}. Signing the real part of eigenvalues of an appropriate Jacobian matrix is the standard approach for determining the stability of an autonomous ODE system, and is the approach commonly used to deduce conditions for the Turing instability. For non-autonomous systems of the form $\dot{\mathbf{Y}}=A(t)\mathbf{Y}$, this approach is neither informative nor appropriate. For sake of demonstration, \cite{vinograd1952criterion} provide an example of a time-dependent matrix $A(t)$ with strictly negative eigenvalues admitting a solution which grows without bound as $t\rightarrow \infty$, while \cite{wu1974note} give an example of $A(t)$ with one positive eigenvalue that results in bounded solutions. These two counterexamples demonstrate that eigenvalues are predictive of neither stability nor instability of non-autonomous ODE systems. Furthermore, employing time-dependent eigenvalues is perhaps more dubious, and we avoid making use of eigenvalues in this manner.

A final difficulty, which is more prominent in non-autonomous systems, is the transient growth and hence the limitation in the correspondence between the asymptotic stability of the linearized system and the actual long-time evolution. In the simpler autonomous case (a static domain) it can be shown that the significance of this transient effect is limited only to a fine parameter tuning (at the fringe of the classical Turing space) \cite{klika2017significance}. However, in non-autonomous systems these transient effects can become more frequent, dependent on the wavenumber and initial conditions. In particular, it was shown that for (exponential) growth with a characteristic time-scale comparable to the characteristic time-scale of reaction kinetics, all wavenumbers above certain threshold grow in initial times yielding a breakdown of the continuum description in finite time \cite{klika2017history}. Therefore, in order to understand the emergence of patterns when the reaction kinetics are on a compatible timescale with evolution of the domain, it is necessary to consider transient pattern formation, and hence necessary to consider a criterion for the emergence of transient instabilities. This will entail the consideration of spatial modes which may be unstable over some finite interval, before again becoming stable for large time, since by such a large time the pattern has already been selected (the actual patterning process occurs within finite time). 

In light of the above, we will develop a criterion for the \textit{transient} instability of specific spatial perturbations, akin to the Turing criteria for static domains. This will allow us to understand both which spatial modes result in instability and lead to possible patterning, but also the duration over which such instabilities persist. We shall later show that there is good agreement between full numerical simulations of the reaction-diffusion dynamics and the selection of spatial perturbations which are transiently unstable under our criterion.  

\section{Turing instability criteria on evolving domains}\label{sec3}
As pointed out in Section \ref{difficult}, the perturbation problem is non-autonomous and hence we can no longer rely on eigenvalues. Furthermore, it is transient dynamics, rather than dynamics as $t\rightarrow \infty$, which lead to patterning in reaction-diffusion systems on evolving domains. As such, we consider a comparison principle for the growth of a spatial perturbation, over some interval of time. This will allow us to determine when a perturbation grows with a certain rate and leads to a transient instability. As we shall later show with numerical simulations, these transiently unstable spatial modes do indeed correspond to patterns formed under the full nonlinear dynamics, subject to the restrictions of a linear analysis. Once a nonlinear pattern has developed, our results are no longer formally valid, though they can give some insight into pattern evolution as we will show later.

The long-time behavior of generic non-autonomous systems such as \eqref{ode4} are too complicated to consider in full generality (as even in the autonomous case, one would appeal to the Routh-Hurwitz stability criterion which becomes cumbersome for large systems \cite{satnoianu2000turing}). In what follows, we restrict our attention to cases commonly arising in the literature. We consider the $n=2$ case for reaction-diffusion systems in detail, as this is the standard case considered in the literature for activator-inhibitor Turing systems. Still, if one is concerned with particular reaction kinetics with $n\geq 3$, then \eqref{ode4} can be solved numerically. We also consider systems which are higher-order in space. Even scalar systems of such a form can permit spatial patterning, and we consider both the scalar case as well as the case of two coupled equations. In all of these results, a dot over a quantity denotes a time derivative, two dots over a quantity denote the second time derivative.

\subsection{Comparison principle}\label{cp}
Growth rates for a scalar first-order non-autonomous problem are trivial to obtain, and we do not discuss this case. The corresponding second-order problem is not as simple, and we establish some growth bounds on general second-order non-autonomous ODE of the form
\begin{equation}\label{gen1}
\ddot{Y} + F(t) \dot{Y} + G(t) Y = 0\,.
\end{equation}
There have been a variety of results for second-order non-autonomous ODE systems \cite{gerber2003riccati, grigoryan2015stability, ince, josic2008unstable}. Due to the breakdown of oscillating solutions, determining conditions for stability can be quite involved and can depend strongly on the properties of non-autonomous terms. On the other hand, obtaining conditions which are sufficient for instability can be viewed as somewhat easier. We begin with a result which gives sufficient conditions for a solution $Y(t)$ to \eqref{gen1} to grow on an arbitrary time interval, which we shall denote $\mathcal{I}\subseteq [0,\infty)$. In particular we can choose $\mathcal{I}$ unbounded to satisfy $|Y(t)| \rightarrow \infty$ as $t\rightarrow \infty$ at a prescribed rate of growth, or $\mathcal{I}$ bounded to only denote regions of transient instability. While such results will be sufficient rather than necessary for instability, as we shall later see, these results will provide the most natural generalization of the standard algebraic Turing instability conditions.

\begin{theorem}\label{Thm1}
Let $\Phi\in C^2(\mathbb{R})$ such that $\Phi(t) > 0$ for all $t\in\mathcal{I}$.  Consider the ODE \eqref{gen1} and suppose that 
\begin{equation}\label{genbound}
G(t) \leq - \frac{\ddot{\Phi}}{\Phi} - \frac{\dot{\Phi}}{\Phi}F(t), \quad t\in\mathcal{I}.
\end{equation}
Then, \eqref{gen1} has a fundamental solution $Y(t)$ with $|Y(t)|\geq\Phi(t)$ for all $t\in\mathcal{I}$.
\end{theorem}
\begin{proof}
We begin with the case where equality holds in the bound \eqref{genbound}. For this case, one may verify that $\Phi(t)$ is in the fundamental solution set of \eqref{ode4} and hence for general initial data \eqref{ode4} has one solution satisfying $|Y(t)| = \Phi(t)$ for any $t\in\mathcal{I}$, although there may be a second solution which grows faster. Therefore $|Y(t)|\geq\Phi(t)$ for all $t\in\mathcal{I}$, with at least equality holding.

Next, consider $G(t) = -\frac{\ddot{\Phi}}{\Phi} - \frac{\dot{\Phi}}{\Phi}F(t) - H(t)$ for some $H(t) \geq 0$ for all $t\in\mathcal{I}$. Then, \eqref{gen1} takes the form
\begin{equation} \label{thm1c}
\ddot{Y} + F(t) \dot{Y} - \left(\frac{\ddot{\Phi}}{\Phi} + \frac{\dot{\Phi}}{\Phi}F(t)+ H(t)\right) Y = 0.
\end{equation}
There are two fundamental solutions to this equation, and any solution will be a linear combination of these. 

We choose initial data $Y(t_0) = \Phi(t_0)$ and $\dot{Y}(t_0) = \dot{\Phi}(t_0)$, and make the change of variable $Y(t) = Y(t_0)\exp\left( \int_{t_0}^t Z(s)ds\right)$, which puts \eqref{thm1c} into the form of the Riccati equation
\begin{equation}
\dot{Z} = - Z^2 - F(t)Z + \frac{\ddot{\Phi}}{\Phi} + \frac{\dot{\Phi}}{\Phi}F(t) + H(t).
\end{equation}
Note that $Z(t_0) = \dot{\Phi}(t_0)/\Phi(t_0)$, so we have
\begin{equation}\label{thm1de}
\dot{Z}  = - Z^2 - F(t)Z + \frac{\ddot{\Phi}}{\Phi} + \frac{\dot{\Phi}}{\Phi}F(t) + H(t)\geq  -Z^2 - F(t)Z + \frac{\ddot{\Phi}}{\Phi} + \frac{\dot{\Phi}}{\Phi}F(t)
= \dot{Z}_1,
\end{equation}
where we define $Z_1(t)$ as a function satisfying
\begin{equation}
\dot{Z}_1 = - Z_1^2 - F(t)Z_1 + \frac{\ddot{\Phi}}{\Phi} + \frac{\dot{\Phi}}{\Phi}F(t),
\end{equation}
with initial data $Z_1(t_0) = \dot{\Phi}(t_0)/\Phi(t_0)$. One may verify that the exact solution reads $Z_1(t) = \dot{\Phi}(t)/\Phi(t)$. Now, by differential inequality \eqref{thm1de} and since $Z(t_0)=Z_1(t_0)$, we have $Z(t) \geq Z_1(t)$ for all $t\in\mathcal{I}$. Integration and exponentiation preserve this ordering, and yield
\begin{equation}
Y(t) = Y(t_0)\exp\left( \int_{t_0}^t Z(s)ds\right) \geq Y(t_0)\exp\left( \int_{t_0}^t Z_1(s)ds\right) = \Phi(t) ,
\end{equation}
since $Y(t_0)\exp\left( \int_{t_0}^t Z_1(s)ds\right) = Y(t_0)\exp\left( \int_{t_0}^t\frac{\dot{\Phi}(s)}{\Phi(s)}ds\right) = \Phi(t)$.
Then, for this choice of initial data, $|Y(t)|\geq \Phi(t)$ for all $t\in\mathcal{I}$. This completes the proof. $\blacksquare$
\end{proof}

We note that these inequalities for $G(t)$ are all sufficient conditions for the prescribed time interval including large-time asymptotic behavior if $\mathcal{I}$ is unbounded. There may be specific problems for which these conditions are not necessary. Classifying such dynamics would involve an advanced study of oscillation theory, and we do not address this here, as our goal is to show that these kinds of sufficient conditions are consistent with the standard Turing conditions for static domains. In linear stability theory, one is often interested in the onset of exponential growth of a small perturbation, and for this case we have the following corollary:

\begin{corollary}\label{cor1}
Consider $\Phi(t) = \mu(t)\exp(\delta t)$ for some $\delta >0$ (where we scale with $\mu(t)$ since the factor of $\mu(t)^{-1}$ in \eqref{pert2} will moderate any instability). From Theorem \ref{Thm1}, we have
\begin{equation}
G(t) \leq - \frac{\ddot{\mu}}{\mu} - 2\delta \frac{\dot{\mu}}{\mu} - \delta^2 - \left( \frac{\dot{\mu}}{\mu} + \delta \right) F(t).
\end{equation}
Taking $\delta \rightarrow 0^+$, and strict inequality, we recover the weakest bound for exponential growth over $t\in\mathcal{I}$,
\begin{equation}\label{gf}
G(t) < - \frac{\ddot{\mu}}{\mu} - \frac{\dot{\mu}}{\mu}F(t), \quad \text{for all} \quad  t\in\mathcal{I}.
\end{equation}
\end{corollary}

While exponential instabilities are the standard for discussing the Turing instability and related patterning (as well as any other stability or instability criteria dependent upon temporal eigenvalues), it is tempting to weaken the strength of the instability, in order to further probe the boundary of stability and instability regions. This difference is only notable for fixed-strength bounds with $\delta>0$ (say, when comparing an instability of rate $\exp(\delta t)$ with an instability of rate $t^{\delta}$), as taking the $\delta \rightarrow 0^+$ limit again results in the same bound \eqref{gf}, as seen the the following corollary:

\begin{corollary}\label{alg}
Consider $\Phi(t) = \mu(t)t^{\delta}$ for some $\delta >0$. From Theorem \ref{Thm1}, we have
\begin{equation}
G(t) \leq - \frac{\ddot{\mu}}{\mu} - 2\delta \frac{\dot{\mu}}{\mu}t^{-1} - \delta(\delta -1)t^{-2} - \left( \frac{\dot{\mu}}{\mu} + \delta t^{-1} \right) F(t).
\end{equation}
Taking $\delta \rightarrow 0^+$ to capture the weakest possible growth rate, and strict inequality, we recover the weakest bound for algebraic growth over $t\in\mathcal{I}$,
\begin{equation}
G(t) < - \frac{\ddot{\mu}}{\mu} - \frac{\dot{\mu}}{\mu}F(t), \quad \text{for all} \quad  t\in\mathcal{I}.
\end{equation}
\end{corollary}

Therefore, we conclude that the bound in equation \eqref{gf} of Corollary \ref{cor1} is robust in terms of accounting for the possible rates of instability.

In light of Theorem \ref{Thm1} along with Corollaries \ref{cor1}-\ref{alg} which provide conditions granting a specific growth rate, it is worthwhile to obtain a complementary result on corresponding decay rates.

\begin{theorem}\label{UpperBound}
 Consider the ODE \eqref{gen1} and suppose that 
\begin{equation}\label{genboundUpper}
G(t) > - \frac{\ddot{\mu}}{\mu} - \frac{\dot{\mu}}{\mu}F(t), \quad \text{for all} \quad  t\in\mathcal{I}.
\end{equation}
Then, solutions $Y(t)$ to \eqref{gen1} do not grow faster than any exponential or algebraic function of $t\in\mathcal{I}$, namely $|Y(t)|\leq K \mu(t) e^{\delta t}$ for all $\delta>0$ and $t\in\mathcal{I}$ and similarly $|Y(t)|\leq K \mu(t) (1+t)^\delta$ where $K=\mathcal{O}(|Y(t_0)|)$.
\end{theorem}
\begin{proof}
  First note that one can repeat the whole proof of Theorem \ref{Thm1} with the opposite inequality yielding that a solution, characterised by the initial condition $Y(t_0)=\Phi(t_0),~\dot{Y}(t_0)=\dot{\Phi}(t_0)$, satisfies an upper bound $|Y(t)|\leq \Phi(t)$ and then the claim follows for this fundamental solution from the two corollaries \ref{cor1}, \ref{alg}.

  To finish the proof we need to show a similar estimate for the second fundamental solution of the ODE \eqref{gen1}. We consider again a function $H\geq0$ so that equality in the bound \eqref{genboundUpper} is obtained $$G(t) - H(t) = - \frac{\ddot{\mu}}{\mu} - \frac{\dot{\mu}}{\mu}F(t), \quad \text{for all} \quad  t\in\mathcal{I},$$
  and the same change of variable $Y(t) = Y(t_0)\exp\left( \int_{t_0}^t Z(s)ds\right)$, which puts \eqref{thm1c} into the form of the Riccati equation
\begin{equation}
\dot{Z} = - Z^2 - F(t)Z + \frac{\ddot{\Phi}}{\Phi} + \frac{\dot{\Phi}}{\Phi}F(t) - H(t).
\end{equation}
Note, however, that one cannot capture the second fundamental solution $Y_2(t) = \Phi(t) \int_{t_0}^t \frac{\exp(-\int_{t_0}^\tau F(s) ds)}{\Phi^2(t)} d\tau$ as it corresponds to initial data $Y_2(t_0)=0$ being impossible to be captured by the aforementioned exponential change of variables. Hence we chose instead initial conditions $Y(t_0)=\Phi(t_0)$, $\dot{Y}(t_0)=\frac{1}{\Phi(t_0)}+\dot{\Phi}(t_0)$ corresponding to the sum of the two mentioned fundamental solutions.

The initial condition after the change of variables reads $Z(t_0) = \frac{\dot{Y}(t_0)}{Y(t_0)}=\frac{\dot{\Phi}(t_0)\Phi(t_0)+1}{\Phi(t_0)^2}$, so we have
\begin{equation}\label{thm1deUpper}
\dot{Z}  = - Z^2 - F(t)Z + \frac{\ddot{\Phi}}{\Phi} + \frac{\dot{\Phi}}{\Phi}F(t) - H(t) \leq  -Z^2 - F(t)Z + \frac{\ddot{\Phi}}{\Phi} + \frac{\dot{\Phi}}{\Phi}F(t)
= \dot{Z}_2,
\end{equation}
where we define $Z_2(t)$ as a function satisfying
\begin{equation*}
\dot{Z}_2 = - Z_2^2 - F(t)Z_2 + \frac{\ddot{\Phi}}{\Phi} + \frac{\dot{\Phi}}{\Phi}F(t),
\end{equation*}
with initial data $Z_2(t_0) = \frac{\dot{Y}(t_0)}{Y(t_0)}=\frac{\dot{\Phi}(t_0)\Phi(t_0)+1}{\Phi(t_0)^2}$. As we already know a solution to this Riccati equation we can identify a general solution to it
\begin{equation*}
Z_2 = \frac{\dot{\Phi}}{\Phi}+\frac{1}{\Phi^2} \exp\left(-\int_{t_0}^t F(\tau) d\tau\right)\left[C+\int_{t_0}^t \frac{1}{\Phi^2(\tau)}\exp\left(-\int_{t_0}^\tau F(s) ds\right)d\tau\right]^{-1},
\end{equation*}
while at $t_0$ its value is
\begin{equation*}
Z_2(t_0)=\frac{\dot{\Phi}(t_0)}{\Phi(t_0)}+\frac{1}{\Phi^2(t_0)} C^{-1},
\end{equation*}
and hence to satisfy the prescribed initial condition we set $C=1$.

Finally, by differential inequality \eqref{thm1deUpper} and since $Z(t_0)=Z_2(t_0)$, we have $Z(t) \leq Z_2(t)$ for all $t\in\mathcal{I}$. Integration and exponentiation preserve this ordering, and yield
\begin{equation*}
  Y(t) = Y(t_0)\exp\left( \int_{t_0}^t Z(s)ds\right) \leq \\
  Y(t_0)\exp\left( \int_{t_0}^t Z_2(s)ds\right) = \Phi(t) \exp\left(\frac{\Psi(t)}{(1+\int_{t_0}^t\Psi(\tau)d\tau)}\right) \leq e \Phi(t),
\end{equation*}
where $\Psi(t)=\frac{1}{\Phi(t)^2} \exp\left(-\int_{t_0}^t F(\tau) d\tau\right)\geq 0$ since $Y(t_0)=\Phi(t_0)$, $\exp\left( \int_{t_0}^t Z_2(s)ds\right)=\frac{\Phi(t)}{\Phi(t_0)}$, and as $\frac{\Psi(t)}{1+\int_{t_0}^t\Psi(\tau)d\tau}\leq 1$ due to nonnegativity of $\Psi$.

Then, for this choice of initial data, $|Y(t)|\leq e \Phi(t)$ for all $t\in\mathcal{I}$ and thus an arbitrary solution to the ODE \eqref{gen1} has to satisfy $|Y(t)|<K \Phi(t)$ for all $t\in\mathcal{I}$.

To complete the proof it suffices to choose $\Phi(t)=\mu(t) e^{\delta t}$ and $\Phi(t)=\mu(t) (1+t)^\delta$ followed by taking the limit $\delta\rightarrow 0^+$. $\blacksquare$
\end{proof}

In Section \ref{mainresult}, we will apply Theorem \ref{Thm1} and Corollary \ref{cor1} in order to obtain the natural analogue of the Turing conditions for a system of two reaction-diffusion equations on an evolving domain. In light of the result presented in Corollary \ref{alg}, the bound obtained in these results is sufficiently general to capture transient growth rates leading to instability. Furthermore, in light of Theorem \ref{UpperBound}, we do not expect a weaker bound to be useful, as the reverse strict inequality results in perturbations which are stable. Therefore, Theorem \ref{Thm1} and Corollary \ref{cor1} indeed provide the most general bounds one is likely to obtain.

\subsection{Instability conditions for systems of two reaction-diffusion equations}\label{mainresult}
Due to the time variability of the base state and the actual growth, the nature of our stability result will be time dependent (rather than for $t\rightarrow \infty$ as is true of the classical Turing conditions), with modes losing and perhaps gaining stability over time. This is exactly along the lines of the history dependence observed in \cite{klika2017history}. We shall then phrase the result in terms of a time interval over which the instability is observed. This interval becomes unbounded if the mode remains unstable as $t\rightarrow \infty$. Throughout the time interval on which an instability arises, given by $\Ik$ for each $\rho_k(t)$, we shall require $J_{12} \neq 0$ for all $t \in \Ik$. Otherwise, the equation for the first chemical species would decouple from the second, and either (i) the reaction kinetics would grow without bound for $J_{11}>0$ or (ii) the perturbation \eqref{pert2} can never give instability for $J_{11}<0$ for any arbitrary spatial perturbation, and hence pattern formation would be impossible. Hence, $J_{12}\neq 0$ is a reasonable assumption. Likewise, we shall assume $J_{21} \neq 0$. 

We now apply the results of Theorem \ref{Thm1} and Corollary \ref{cor1} to obtain conditions on the instability of spatial perturbations of the form \eqref{pert2}. As we mentioned above, due to Theorem \ref{UpperBound}, these are the best instability bounds one expects to obtain. In order to invoke Theorem \ref{Thm1} and Corollary \ref{cor1}, we first convert the non-autonomous first-order system into a scalar non-autonomous second-order scalar ODE. The generalization of the Turing conditions is as follows:

\begin{theorem}\label{Thm3}
Consider the evolution of a compact, simply connected, smooth Riemannian manifold $\Omega(t)\subset \mathbb{R}^N$ as in \eqref{coords}, with Laplace-Beltrami operator spectrum $\rho_k(t) \in C^1(\mathcal{I}_k)$ where $\mathcal{I}_k\subseteq (0,\infty)$, such that volume expansion or contraction $\mu(t) \in C^2(\mathcal{I}_k)$ given in \eqref{vol} is independent of space. Assume that $J\in C^1(\mathcal{I}_k)$ is the time-dependent Jacobian matrix of $\mathbf{f}$ evaluated at the spatially homogeneous solution $\mathbf{U}(t)$ to \eqref{uniform}, with $J_{12}, J_{21} \neq 0$ on $\mathcal{I}_k$. For $n=2$ species, the spatially homogeneous state $\mathbf{U}(t)$ for the reaction-diffusion system \eqref{RD_system:eqn} is linearly unstable under a perturbation of the form \eqref{pert2} corresponding to $\rho_k(t)$ for $t\in \mathcal{I}_k$, provided that the inequality
\begin{equation}\begin{aligned}\label{T3}
& \det(J) - \left( d_2 J_{11} + d_1 J_{22} \right)\rho_k + d_1d_2\rho_k^2 \\
& \qquad
 < - \frac{\ddot{\mu}}{\mu} - \frac{\dot{\mu}}{\mu}\left( \left( d_1 + d_2\right)\rho_k - \mathrm{tr}(J) \right) \\
& \qquad\qquad + \max\left\lbrace \frac{\dot{\mu}}{\mu} \frac{\dot{J}_{12}}{J_{12}} - J_{12}\frac{d}{dt}\left( \frac{d_1 \rho_k -J_{11}}{J_{12}}\right) , \frac{\dot{\mu}}{\mu} \frac{\dot{J}_{21}}{J_{21}} - J_{21}\frac{d}{dt}\left( \frac{d_2 \rho_k -J_{22}}{J_{21}}\right)\right\rbrace 
\end{aligned}\end{equation}
holds for all $t\in\mathcal{I}_k$.
\end{theorem}
\begin{proof}
For $n=2$ species, and for each $k=0,1,2,\dots$, \eqref{ode4} reads
\begin{align}
\frac{dV_1}{dt} & = - d_1 \rho_k(t) V_1 + J_{11}V_1 + J_{12} V_2 ,\label{v1a}\\
\frac{dV_2}{dt} & = - d_2 \rho_k(t) V_2 + J_{21}V_1 + J_{22} V_2 .\label{v2a}
\end{align}
Recall that $J=J(\mathbf{U})$, where $\mathbf{U}(t)$ is given by \eqref{uniform}, hence the components of $J$ are in general time-dependent. 

We start with $V_1(t)$. Since $J_{12}\neq 0$ for $t\in \mathcal{I}_k$, we isolate \eqref{v1a} for $V_2(t)$, and use it in \eqref{v2a} to obtain a single second-order ODE for $V_1(t)$, finding
\begin{equation}\begin{aligned}\label{v1b}
&\frac{d^2V_1}{dt^2} + \left\lbrace \left( d_1 + d_2\right)\rho_k - \text{tr}(J) - \frac{\dot{J}_{12}}{J_{12}} \right\rbrace\frac{dV_1}{dt}  \\
& \quad  + \left\lbrace \det(J) - \left( d_2 J_{11} + d_1 J_{22} \right)\rho_k + d_1d_2\rho_k^2  + J_{12}\frac{d}{dt}\left( \frac{d_1 \rho_k -J_{11}}{J_{12}}\right) \right\rbrace V_1 =0.
\end{aligned}\end{equation}
Applying \eqref{gf} of Corollary \ref{cor1} to \eqref{v1b}, we arrive at the sufficient condition
\begin{equation}\begin{aligned}\label{T31}
&\det(J) - \left( d_2 J_{11} + d_1 J_{22} \right)\rho_k + d_1d_2\rho_k^2\\
& \qquad <  - \frac{\ddot{\mu}}{\mu} - \frac{\dot{\mu}}{\mu}\left( \left( d_1 + d_2\right)\rho_k - \mathrm{tr}(J) - \frac{\dot{J}_{12}}{J_{12}} \right) - J_{12}\frac{d}{dt}\left( \frac{d_1 \rho_k -J_{11}}{J_{12}}\right),
\end{aligned}\end{equation}
which implies exponential growth of the $u_1$ component of the perturbation \eqref{pert2}.
We perform similar calculations using $J_{21}\neq 0$ for $t\in \mathcal{I}_k$, in order to obtain a second-order ODE for $V_2(t)$. Applying \eqref{gf} of Corollary \ref{cor1} to this ODE, we arrive at the sufficient condition 
\begin{equation}\begin{aligned}\label{T32}
& \det(J) - \left( d_2 J_{11} + d_1 J_{22} \right)\rho_k + d_1d_2\rho_k^2\\
& \qquad  <  - \frac{\ddot{\mu}}{\mu} - \frac{\dot{\mu}}{\mu}\left( \left( d_1 + d_2\right)\rho_k - \mathrm{tr}(J) - \frac{\dot{J}_{21}}{J_{21}} \right) - J_{21}\frac{d}{dt}\left( \frac{d_2 \rho_k -J_{22}}{J_{21}}\right),
\end{aligned}\end{equation}
which implies exponential growth of the $u_2$ component of the perturbation \eqref{pert2}.

As we only require one of \eqref{T31} or \eqref{T32} to hold for instability, we take the inequality corresponding to the more extreme inequality in \eqref{T31} or \eqref{T32}, resulting in the appearance of a max function in \eqref{T3}. This completes the proof. $\blacksquare$
\end{proof}

These are the conditions for the system \eqref{growthgen} to exhibit an instability corresponding to the $k$th spatial mode for $t\in \mathcal{I}_k$. In practice we shall consider $\mathcal{I}_k$ to be the largest interval on which the hypotheses of Theorem \ref{Thm3} hold, though for transient or sporadic growth periods there may be distinct intervals. These generalize the standard Turing conditions to corresponding conditions on a smoothly time-evolving manifold, though we remark that we do not yet incorporate a generalization of the standard stability of the homogeneous equilibrium in the absence of diffusion. Akin to what is done for classical Turing conditions, one may choose to group all modes which are unstable at time $t$, and the natural definition for this set will be: $\mathcal{K}_t = \{ k \in \mathbb{N} | \{ t \} \cap \Ik \neq \emptyset \}$. Similar generalizations hold when dealing with multi-indices. For higher-dimensional domains with spectra indexed like $\rho_{k_1, \dots, k_\ell}$, we define $\mathcal{I}_{k_1, \dots, k_\ell}$ and $\mathcal{K}_t$ accordingly.

Before continuing, we briefly connect our result to the standard Turing condition for instability on a static domain. We remark that in the case where the domain is static in time, the spectrum $\rho_k$ is constant, $\dot{\mu}=\ddot{\mu}=0$, and all entries in the matrix $J$ are constant. As such, the condition in Theorem \ref{Thm3} reduces to 
\begin{equation}\label{simplify1}
\det(J) - \left( d_2 J_{11} + d_1 J_{22} \right)\rho_k + d_1d_2\rho_k^2 <0\,,
\end{equation}
which is exactly the standard Turing condition on a static manifold. In the case where the manifold is flat and rectangular, or flat and unbounded, the spectrum $\rho_k = |\mathbf{k}|^2$ for some wavenumber vector $\mathbf{k}$ with dimension equal to the dimension of the space. For such a case, \eqref{simplify1} reduces further to
\begin{equation}\label{simplify2}
\det(J) - \left( d_2 J_{11} + d_1 J_{22} \right)|\mathbf{k}|^2 + d_1d_2|\mathbf{k}|^4 <0\,,
\end{equation}
and this is the Turing condition most commonly seen in the literature \cite{MurrayI}.

\subsection{Instability conditions for higher-order systems}\label{instabhigher}
Returning to the higher-order systems discussed in Section \ref{hoformulation}, we have the following analogue of Theorem \ref{Thm3} for coupled pairs ($n=2$) of systems taking the form \eqref{higher1}:

\begin{theorem}\label{Thm4}
Consider the evolution of a compact, simply connected, smooth Riemannian manifold $\Omega(t)\subset \mathbb{R}^N$ as in \eqref{coords}, with Laplace-Beltrami operator spectrum $\rho_k(t) \in C^1(\mathcal{I}_k)$ where $\mathcal{I}_k\subseteq (0,\infty)$, such that volume expansion or contraction $\mu(t) \in C^2(\mathcal{I}_k)$ given in \eqref{vol} is independent of space. Assume that $J\in C^1(\mathcal{I}_k)$ is the time-dependent Jacobian matrix of $\mathbf{f}$ evaluated at the spatially homogeneous solution $\mathbf{U}(t)$ to \eqref{uniform}, with $J_{12}, J_{21} \neq 0$ on $\mathcal{I}_k$. For $n=2$ species, the spatially homogeneous state $\mathbf{U}(t)$ for the system \eqref{higher1} is linearly unstable under a perturbation of the form \eqref{pert2} corresponding to $\rho_k(t)$ for $t\in \mathcal{I}_k$, provided that the inequality
\begin{equation}\begin{aligned}\label{T4}
& \det(J) + \left( d_2 J_{11} + d_1 J_{22} \right)\mathcal{P}(-\rho_k) + d_1d_2\left( \mathcal{P}(-\rho_k)\right)^2 \\
& \qquad
 < - \frac{\ddot{\mu}}{\mu} + \frac{\dot{\mu}}{\mu}\left( \left( d_1 + d_2\right)\mathcal{P}(-\rho_k) + \mathrm{tr}(J) \right) \\
& \qquad\qquad + \max\left\lbrace \frac{\dot{\mu}}{\mu} \frac{\dot{J}_{12}}{J_{12}} + J_{12}\frac{d}{dt}\left( \frac{d_1 \mathcal{P}(-\rho_k) +J_{11}}{J_{12}}\right) , \frac{\dot{\mu}}{\mu} \frac{\dot{J}_{21}}{J_{21}} + J_{21}\frac{d}{dt}\left( \frac{d_2 \mathcal{P}(-\rho_k) +J_{22}}{J_{21}}\right)\right\rbrace 
\end{aligned}\end{equation}
holds for all $t\in\mathcal{I}_k$.
\end{theorem}

The proof of Theorem \ref{Thm4} is similar to that of Theorem \ref{Thm3}, so we omit it. 

In the scalar case ($n=1$), it is also possible to have instability, provided that the degree of $\mathcal{P}$ is at least two. In the case where $\mathcal{P}(y)=y$, as for standard reaction-diffusion systems, the scalar reaction-diffusion system admits spatial perturbations \eqref{pert2} which grow or decay like
\begin{equation}\label{scalar1}
\frac{d{V}_k}{dt} = -d \rho_k(t) {V}_k + f'({U}){V}_k.
\end{equation} 
Equation \eqref{scalar1} is exactly solvable, and we have
\begin{equation}
|V_k(t)| = |C_k| \exp\left(\int_{t_0}^t f'(U(s))ds - d\int_{t_0}^t \rho_k(s)ds\right) \leq |C_k| \exp\left(\int_{t_0}^t f'(U(s))ds\right) = |V_0(t)|\,,
\end{equation}
since $d>0$ and $\rho_k(t)\geq 0$. Hence, diffusion is always stabilising in the scalar case of one reaction-diffusion equation without any outside forcing. However, in the higher-order case, the structure of the polynomial $\mathcal{P}$ can permit instability in the scalar ($n=1$) case of \eqref{higher1}, 
\begin{equation}\label{higher5}
\frac{\partial u}{\partial t} + \nabla_{\Omega(t)} \cdot (\mathbf{Q}u) = d \mathcal{P}\left(\nabla_{\Omega(t)}^2\right) u + f(u),
\end{equation}
and we give conditions for such an instability in the following Theorem:

\begin{theorem}\label{Thm5}
Consider the evolution of a compact, simply connected, smooth Riemannian manifold $\Omega(t)\subset \mathbb{R}^N$ as in \eqref{coords}, with Laplace-Beltrami operator spectrum $\rho_k(t) \in C^1(\mathcal{I}_k)$ where $\mathcal{I}_k\subseteq (0,\infty)$, such that volume expansion or contraction $\mu(t) \in C^2(\mathcal{I}_k)$ given in \eqref{vol} is independent of space. Assume that $f'(U)\in C^1(\mathcal{I}_k)$, where $U(t)$ is the spatially uniform state satisfying \eqref{uniform} which in the scalar case reads
$\frac{dU}{dt}=-\frac{\dot{\mu}}{\mu}U + f(U)$. For $\mathcal{P}$ of degree at least two, the spatially homogeneous state $U(t)$ for the system \eqref{higher5} is linearly unstable under a perturbation of the form \eqref{pert2} corresponding to $\rho_k(t)$ for $t\in \mathcal{I}_k$, provided that the inequality
\begin{equation}\label{T5}
d\mathcal{P}(-\rho_k(t)) + f'(U(t)) - \frac{\dot{\mu}}{\mu} >0 
\end{equation}
holds for all $t\in\mathcal{I}_k$.
\end{theorem}
\begin{proof}
For the scalar equation \eqref{higher5}, the linear stability of the $k$th spatial perturbation is given by the scalar form of \eqref{ode4},
\begin{equation}\label{scalar2}
\frac{d\mathbf{V}_k}{dt} = d \mathcal{P}(-\rho_k(t)) \mathbf{V}_k + f'(\mathbf{U})\mathbf{V}_k,
\end{equation} 
and solving \eqref{scalar2} we obtain
\begin{equation}
V_k(t) = C_k \exp\left( d\int_{t_0}^t \mathcal{P}(-\rho_k(s))ds + \int_{t_0}^t f'(U(s))ds \right)\,.
\end{equation}
There is then growth of the perturbation \eqref{pert2} provided 
\begin{equation}
\frac{|V_k(t)|}{\mu(t)} \geq |C_k|e^{\delta (t-t_0) - \log(\mu(t_0))}
\end{equation}
for some $\delta >0$, noting that we have assumed an exponential growth rate. (As pointed out in Section \ref{cp}, an exponential growth rate is sufficiently general.) Rearranging, we find
\begin{equation}
d\int_{t_0}^t \mathcal{P}(-\rho_k(s))ds + \int_{t_0}^t f'(U(s))ds - \log(\mu(t))) \geq \delta (t-t_0) - \log(\mu(t_0))\,.
\end{equation}
Both sides of the inequality are equal at $t=t_0$, so we may differentiate the inequality since the left hand side must grow faster than the right hand side in order for there to be a transient instability over $t\in \mathcal{I}_k$, and we find
\begin{equation}
d\mathcal{P}(-\rho_k(t)) + f'(U(t)) - \frac{\dot{\mu}}{\mu} \geq \delta \,. 
\end{equation}
Taking strict inequality in the limit $\delta \rightarrow 0^+$, we have the stated inequality \eqref{T5}. This completes the proof. $\blacksquare$
\end{proof}

We will give explicit examples of transient instabilities and pattern formation in fourth-order scalar systems in Section \ref{shsec}.

In Section \ref{hoformulation}, we also commented that the most generqal higher-order spectral problem may not simply involve Laplace-Beltrami eigenfunctions $\psi_k$ but also other eigenfunctions $\hat{\psi}_k$, depending on the higher-order boundary conditions. In this case, upon considering the stationary coordinates \eqref{efxn}, one instead has the spectral problem
$\mathcal{P}\left(\nabla_{\Omega(t)}^2\right)\hat{\psi}_k = -\zeta_k(t) \hat{\psi}_k$. Through a similar approach to Theorems \ref{Thm4}-\ref{Thm5}, we find

\begin{theorem}\label{Thm6}
Consider the evolution of a compact, simply connected, smooth Riemannian manifold $\Omega(t)\subset \mathbb{R}^N$ as in \eqref{coords}. Let the differential operator $\mathcal{P}\left(\nabla_{\Omega(t)}^2\right)$ have corresponding spectrum $\zeta_k(t) \in C^1(\mathcal{I}_k)$ where $\mathcal{I}_k\subseteq (0,\infty)$, such that volume expansion or contraction $\mu(t) \in C^2(\mathcal{I}_k)$ given in \eqref{vol} is independent of space. Assume that $J\in C^1(\mathcal{I}_k)$ is the time-dependent Jacobian matrix of $\mathbf{f}$ evaluated at the spatially homogeneous solution $\mathbf{U}(t)$ to \eqref{uniform}, with $J_{12}, J_{21} \neq 0$ on $\mathcal{I}_k$. 
Then, the spatially homogeneous state $U(t)$ for the system \eqref{higher5} corresponding to $n=1$ species is linearly unstable under a perturbation of the form \eqref{pert2} corresponding to $\zeta_k(t)$ for $t\in \mathcal{I}_k$, provided that the inequality
\begin{equation}
-d\zeta_k(t) + f'(U(t)) - \frac{\dot{\mu}}{\mu} >0 
\end{equation}
holds for all $t\in\mathcal{I}_k$. Meanwhile, for $n=2$ species, the spatially homogeneous state $\mathbf{U}(t)$ for the system \eqref{higher1} is linearly unstable under a perturbation of the form \eqref{pert2} corresponding to $\rho_k(t)$ for $t\in \mathcal{I}_k$, provided that the inequality
\begin{equation}\begin{aligned}
& \det(J) - \left( d_2 J_{11} + d_1 J_{22} \right)\zeta_k(t) + d_1d_2\left( \zeta_k(t)\right)^2 \\
& \qquad
 < - \frac{\ddot{\mu}}{\mu} - \frac{\dot{\mu}}{\mu}\left( \left( d_1 + d_2\right)\zeta_k(t) + \mathrm{tr}(J) \right) \\
& \qquad\qquad + \max\left\lbrace \frac{\dot{\mu}}{\mu} \frac{\dot{J}_{12}}{J_{12}} + J_{12}\frac{d}{dt}\left( \frac{J_{11}-d_1\zeta_k(t)}{J_{12}}\right) , \frac{\dot{\mu}}{\mu} \frac{\dot{J}_{21}}{J_{21}} + J_{21}\frac{d}{dt}\left( \frac{J_{22}-d_2 \zeta_k(t)}{J_{21}}\right)\right\rbrace 
\end{aligned}\end{equation}
holds for all $t\in\mathcal{I}_k$.
\end{theorem}

This is the most general bound on such higher-order problems. However, for the sake of simulations, we restrict our attention to simple higher-order problems with boundary conditions yielding a straightforward collection of eigenfunctions. Still, Theorem \ref{Thm6} is more general than Theorems \ref{Thm4}-\ref{Thm5}, and should be regarded as the primary result for such systems, as for the most general problems it will include spectral contributions that might be ignored by Theorems \ref{Thm4}-\ref{Thm5}.

\subsection{Asymptotic stability results for reaction-diffusion systems}
In this section we shall analyse the special case of asymptotic stability, i.e. $\Ik=(T,\infty)$.

\subsubsection{No pattern in the large-time limit for unbounded growth}
Consider first the case of unbounded growth when $\rho_k(t)\rightarrow 0$ for all $k$ (which is the case for purely dilatational growth, $\rho_k(t)=\rho_0(0)/r^2(t)$). Then the condition \eqref{T3} for instability of the $k$-th wavemode reads
\begin{equation}\label{nopattern}
\det(J) < - \frac{\ddot{\mu}}{\mu} + \frac{\dot{\mu}}{\mu} \mathrm{tr}(J)  + \max\left\lbrace \frac{\dot{\mu}}{\mu} \frac{\dot{J}_{12}}{J_{12}} + J_{12}\frac{d}{dt}\left( \frac{J_{11}}{J_{12}}\right) , \frac{\dot{\mu}}{\mu} \frac{\dot{J}_{21}}{J_{21}} + J_{21}\frac{d}{dt}\left( \frac{J_{22}}{J_{21}}\right)\right\rbrace \quad \text{as} \quad t\rightarrow \infty \,.
\end{equation}
Note the obvious that in this case the condition for instability is independent of the wavemode number. Thus, if this condition is satisfied, then asymptotically all spatial modes lead to instability. Therefore we expect that such a situation would not yield a reasonable pattern (not having arbitrary small lengthscale) and further note that this inequality is also the one satisfied by $\rho_0(t)=0$ at finite times, indicating spatially homogeneous instability. Hence, for asymptotically large times, spatial instabilities (of arbitrary mode number) coincide with homogeneous instabilities which are typically precluded.

Therefore, there is no diffusion driven patterning on domains undergoing unbounded growth for asymptotically large time. Any spatial patterning under such growth must therefore occur due to transient dynamics. This is quite distinct from the static, bounded domain case, where diffusive instabilities retain their dominance in the large-time limit.

\subsubsection{Saturating growth scenarios}
If the growth stopped at a finite time or is saturating at a finite size then the above observations about the asymptotic stability reduce exactly to classical diffusion-driven instability condition. Indeed, in this case we have $\lim_{t\rightarrow \infty}\rho_k(t)=\overline{\rho}_k>0$, while $\lim_{t\rightarrow \infty}\dot{\mu}=0$ and $\lim_{t\rightarrow \infty}\dot{J}_{ij}=0$, so $\mathbf{U}^*$ is the homogeneous steady state solution satisfying $f(\mathbf{U}^*)=0$ and the sufficient conditions read
\begin{equation}
  \det(J) - \left( d_2 J_{11} + d_1 J_{22} \right)\overline{\rho}_k + d_1d_2 \overline{\rho}_k^2 < 0,
\end{equation}
which is exactly the Turing condition which one would derive on the static domain $\overline{\Omega}$ \cite{dhillon2017bifurcation} with $\mathcal{I}_k$ unbounded. As $d_1d_2\overline{\rho}_k^2$ is dominant for large $k$, the spatial modes of high frequency are stable, in line with classical Turing conditions \cite{Murray}.

\subsubsection{Asymptotic stability of the base state}
It is worth briefly discussing the asymptotic stability of the base state solution of \eqref{uniform}, and to do so we consider two cases. 
 
First, suppose that $\lim_{t\rightarrow \infty}\frac{\dot{\mu}}{\mu} = \nu$, a constant. This is true, for example, in the case of exponential growth of a domain. The dynamics of the base state \eqref{uniform} then read
\begin{equation}\label{basestate1}
\dot{\mathbf{U}} = -\nu\mathbf{U} + \mathbf{f}(\mathbf{U})\quad \text{as} \quad t\rightarrow \infty\,.
\end{equation}
As this equation is autonomous in the large-time limit, to obtain a steady state we set $-\nu\mathbf{U} + \mathbf{f}(\mathbf{U})=\mathbf{0}$, and denote by $\mathbf{U}^\dag$ a solution of this algebraic equation. Note that $\mathbf{U}^\dag$ is not in general equal to $\mathbf{U}^*$, which is the solution of the algebraic equation $\mathbf{f}(\mathbf{U})=\mathbf{0}$ as discussed in Section \ref{sec.GenLinStabAnal}. In particular, $\mathbf{U}^\dag = \mathbf{U}^*$ when $\nu = 0$, i.e., when the change in volume expansion or contraction is zero in the asymptotic limit $t\rightarrow \infty$. The important thing to note here is that in the case where $\lim_{t\rightarrow \infty}\frac{\dot{\mu}}{\mu}$ exists and is equal to some constant $\nu \in \mathbb{R}$, the appropriate spatially uniform base state is $\mathbf{U}^\dag$ satisfying $-\nu\mathbf{U} + \mathbf{f}(\mathbf{U})=\mathbf{0}$ rather than the base state in the absence of domain evolution, $\mathbf{U}^*$, which satisfies $\mathbf{f}(\mathbf{U})=\mathbf{0}$. 

Assuming a linear perturbation of the form 
\begin{equation}\label{basestate2}
\mathbf{U}(t) = \mathbf{U}^\dag + \epsilon \mathbf{V}(t)\,,
\end{equation}
equation \eqref{basestate1} is linearized as
\begin{equation}\label{basestate3}
\dot{\mathbf{V}}=\left( -\nu I + J\right)\mathbf{V}\,,
\end{equation}
where $J = \frac{\partial \mathbf{f}}{\partial \mathbf{U}}$ evaluated at $\mathbf{U} = \mathbf{U}^\dag$. In particular, the Jacobian matrix is now constant as it is evaluated at the steady state $\mathbf{U}^\dag$. In the standard way, for the $n=2$ species case, we have from \eqref{basestate3} that the steady state $\mathbf{U}^\dag$ is asymptotically stable provided that the following necessary conditions are satisfied:
\begin{subequations}\begin{align}
& \text{tr}\left( -\nu I + J\right) = -2\nu + \text{tr}(J) < 0\,, \\
& \det \left( -\nu I + J\right) = \det(J) - \nu \text{tr}(J) + \nu^2 > 0\,.
\end{align}\end{subequations}
These are natural analogues of the standard necessary conditions for reaction kinetics to be stable on static domains, $\text{tr}(J)<0$ and $\det(J) >0$.

Of course, for more complicated domain evolution, the limit $\lim_{t\rightarrow \infty}\frac{\dot{\mu}}{\mu}$ need not exist, and this leads us to our second case. In this case, we perturb a time-dependent solution of \eqref{uniform} akin to what we did in \eqref{basestate2}. However, this is equivalent to a perturbation of the form \eqref{pert2} with a spatially homogeneous mode (the $k=0$ mode, which always exists for the Neumann problem on manifolds with boundary, as well as for manifolds without boundary). In light of Theorem \ref{UpperBound}, we anticipate a condition akin to that given in Theorem \ref{Thm3} only with a sign reversed. Indeed, carrying out a similar analysis, and invoking Theorem \ref{UpperBound}, we find that a necessary condition for the stability of a spatially uniform yet time varying base state solution to \eqref{uniform} reads:
\begin{equation}
\det(J) > - \frac{\ddot{\mu}}{\mu} + \frac{\dot{\mu}}{\mu}\mathrm{tr}(J) + \max\left\lbrace \frac{\dot{\mu}}{\mu} \frac{\dot{J}_{12}}{J_{12}} + J_{12}\frac{d}{dt}\left( \frac{J_{11}}{J_{12}}\right) , \frac{\dot{\mu}}{\mu} \frac{\dot{J}_{21}}{J_{21}} + J_{21}\frac{d}{dt}\left( \frac{J_{22}}{J_{21}}\right)\right\rbrace \,.
\end{equation}
This is the complement of the condition given in Theorem \ref{Thm3}, for the zeroth mode $k=0$ with $\rho_0(t) =0$. This condition is also complementary to that given in \eqref{nopattern}, which makes sense, as \eqref{nopattern} was the unrealistic sufficient condition for all modes (even the zeroth mode) to result in an instability.

\subsection{Equal diffusion coefficients}
To explore whether equal diffusion coefficients permit pattern formation it is advantageous to transform the evolution equations for the perturbation \eqref{ode4} via
\begin{equation} \label{eq.transf4AltPert}
  \mathbf{V} = \exp(-P_k(t) {D}) \mathbf{A}, \quad P_k(t) = \int_{t_0}^t \rho_k(s) ds
\end{equation}
into
\begin{equation} \label{eq.AltPert}
  \dot{\mathbf{A}} = \exp(P_k(t){D}) {J} \exp(-P_k(t){D}) \mathbf{A} = \begin{pmatrix} J_{11} & J_{12} \exp(P_k(t)(d_1-d_2) \\ J_{21} \exp(P_k(t)(d_2-d_1) & J_{22} \end{pmatrix} \mathbf{A}.
\end{equation}
Note that for finite time intervals or when $\rho_k(t)\rightarrow 0$ the transformation is such that $\mathbf{A}$ is stable iff $\mathbf{V}$ is stable. Otherwise, i.e. for finite positive limit $\lim_{t\rightarrow \infty}\rho_k(t)=\rho_k^*>0$, stability in the original variable $\mathbf{V}$ is guaranteed only if $\mathbf{A}$ grows at least as fast as the matrix exponential $\exp(\rho_k^* {D}t)$.

To use Theorem \ref{Thm1} and Corollary \ref{cor1}, we rewrite the alternative relation for perturbation evolution \eqref{eq.AltPert} as a second-order equation
\begin{equation}
  \ddot{A}_1 - \dot{A}_1 \underbrace{\left(\mathrm{tr} ({J}) +\frac{\dot{J}_{12}}{J_{12}} + (d_1-d_2) \rho_k(t)\right)}_{-F(t;k)} + A_1 \underbrace{\left[\det(J) + J_{11} \left(\frac{\dot{J}_{12}}{J_{12}} + (d_1-d_2) \rho_k(t)\right)-\dot{J}_{11}\right]}_{G(t;k)} = 0,
\end{equation}
where the equation for $A_2$ is the same with swapped indices $1\leftrightarrow 2$.

In the case where all diffusion coefficients are equal, $D=d_1I$, hence $\exp\left( P_k(t) D \right)= e^{d_1 P_k(t)}I$, and we have $\dot{\mathbf{A}}=J(\mathbf{U}) \mathbf{A}$. Then, $\mathbf{V}_k = \exp\left( - P_k(t) D \right)\mathbf{A}=e^{-d_1 P_k(t)}\mathbf{A}$, where $\mathbf{A}$ depends on reaction kinetics at the base state through $J(\mathbf{U})$. As $d_1>0$ and $P_k(t)>0$, the contribution of diffusion is stabilizing, with any instability arising only from a combination of domain evolution (through the $\mu(t)$ term in \eqref{ode4}) and reaction kinetics, precluding spatial patterning due to diffusive instabilities.

In addition, an application of Corollary \ref{cor1} reveals that no instability (pattern) can be expected for equal diffusion coefficients for finite time intervals or for unbounded growth with $\rho_k(t)\rightarrow 0$ as $\rho_k(t)$ vanishes from both $F(t;k)$ and $G(t;k)$. Finally, in the asymptotic case with $\lim_{t\rightarrow \infty} \rho_k(t)=\overline{\rho_k}$ the functions $F,~G$ are also independent of $k$, however, the requirement of exponential growth  at least as fast as the matrix exponential $\exp(\rho_k^* {D}t)$ results in a dependence of the upper and lower bounds on $k$ only via the exponential bound $\delta=d_j \overline{\rho_k}$. Because both terms in the bounds depending on wave number $k$ are negative, the boundary between the upper and lower bound (hence the threshold for instability) is more stringent than a sufficient condition for instability without diffusion. Therefore equal diffusion coefficients do not allow the emergence of spatial patterns (not only being diffusion driven) even on evolving domains.

\subsection{Transient breakdown of the continuum assumption}
As noted above, any understanding of transient dynamics is welcomed, especially on evolving domains. The sufficient condition (the threshold for instability) given in Theorem \ref{Thm3} also allows us to study such effects. Indeed, while \cite{klika2017history} explored transient growth, and showed that rapid growth can lead to arbitrarily-large wavemode excitation, their results can be found as a special case of the instability results presented here.

We shall use the above transformation \eqref{eq.transf4AltPert}. From Corollary \ref{cor1}, transient exponential growth then happens when
\begin{equation}
  \det ({J}) + J_{j,j} \left(\frac{\dot{J}_{j,{\neg j}}}{J_{j,\neg j}} + (d_j-d_{\neg j}) \rho_k(t)\right)-\dot{J}_{j,j} < -\frac{\ddot{\mu}}{\mu} +\frac{\dot{\mu}}{\mu} \left(\mathrm{tr} (J) +\frac{\dot{J}_{j, \neg j}}{J_{j, \neg j}} + (d_j-d_{\neg j}) \rho_k(t)\right)\,,
\end{equation}
where $\neg j$ denotes the index not being $j$ (e.g., if $j=1$ then $\neg j=2$). Focusing on large wavenumbers $k$, we find that they become unstable if
\begin{equation}
  J_{j,j} (d_j-d_{\neg j}) \rho_k(t) < \frac{\dot{\mu}}{\mu}  (d_j-d_{\neg j}) \rho_k(t)\,,
\end{equation}
or equivalently, when
\begin{equation}
  J_{j,j} (d_j-d_{\neg j}) < \frac{\dot{\mu}}{\mu}  (d_j-d_{\neg j}).
\end{equation}
One immediately observes that for sufficiently rapid growth, this inequality is satisfied (for one of the $j$, since we assume $d_1 \neq d_2$), and hence fast growth (measured by $\dot{\mu}/\mu$) always yields transient exponential growth for arbitrarily large wavenumbers, provided that the time-dependent Jacobian entries remain bounded. Such an instability entails a breakdown of the linear analysis as exemplified in \cite{klika2017history}. On the other hand, if the Jacobian entries also change rapidly (recall that they depend on the dynamics of the base state governed by \eqref{uniform}, and hence on the quantity $\dot{\mu}/\mu$), then this effect will be suppressed. In practice, we do not observe transient exponential growth for arbitrarily large wavenumbers in our simulations, at any time.

\section{Applications to reaction-diffusion systems in one space dimension}\label{sec4}
We illustrate the analytical instability conditions given in Theorem \ref{Thm3} by considering various case studies consisting of specific growth functions and domain geometries, some of which extend studies in the literature, and others of which have seemingly never been considered due to their complexity in the face of existing methods. We note that the conditions given in Theorem \ref{Thm3} are sufficient for an instability to grow over a specified time interval, but are insufficient to determine if a given time period is sufficient to observe a heterogeneity forming in a simulation of the full system, as this will depend on the specific nonlinearities involved. Nevertheless, we aim to show that the linear stability analysis captures a variety of solution features observed in numerical simulations.

We first consider systems in one spatial dimension in the present section, before moving onto more complicated configurations in the following sections. Before this, we provide a brief discussion of the reaction kinetics and numerical schemes used. 

\subsection{Reaction kinetics}
We consider the Schnakenberg, or activator-depleted, reaction kinetics as a very simple example which is used extensively in the literature \cite{gierer1972theory, schnakenberg1979simple}. The kinetics and homogeneous equilibria at $t=t_0$ read
\begin{equation}\label{Schnack}
\mathbf{f}(u_1,u_2) = \begin{pmatrix} a-u_1+u_1^2u_2 \\ b-u_1^2u_2 \end{pmatrix}, \quad \mathbf{U^*} = \begin{pmatrix} a+b \\ \frac{b}{(a+b)^2} \end{pmatrix},
\end{equation}
where $a,b$ will be taken as non-negative real parameters. We will also consider the FitzHugh-Nagumo kinetics to demonstrate the applicability of our results to an oscillatory base state giving rise to Hopf and Turing-Hopf bifurcations \cite{fitzhugh1955mathematical, keener1998mathematical, nagumo1962active}. The kinetics and homogeneous equilibria at $t=t_0$ read
\begin{equation}\label{FHN}
\mathbf{f}(u_1,u_2) = \begin{pmatrix} c\left (u_1-\frac{u_1^3}{3}+ u_2-i_0\right ) \\ \frac{a-u_1-bu_2}{c} \end{pmatrix},\quad \mathbf{U^*} = \begin{pmatrix} U_1^* \\ \frac{a-U_1^*}{b} \end{pmatrix},
\end{equation}
where $a,b,c$ and $i_0$ are taken as non-negative constants, and $U_1^*$ will be the root of $c(U_1^*-U_1^{*3}/3- (a-U_1^*)-i_0)=0$. For the parameters we will use, this equation will have a unique real root, and so the system will have a unique steady state solution.

\subsection{Numerical approach}
We simulate \eqref{growthgen} with the kinetics \eqref{Schnack} using the finite element solver COMSOL, version 5.3, with which we discretize the manifolds using second-order finite elements (which will be triangular and tetrahedral in the higher dimensional examples). We used Matlab to compute the evolution of the homogeneous state, and to generate $\Ik$ according to Theorem \ref{Thm3}. We verified simulations in various static domain cases (1-D intervals, 2-D spheres) using the Matlab package Chebfun \cite{townsend2013extension}, in addition to convergence checks in spatial and time discretizations. In all simulations, we used a relative tolerance of $10^{-5}$, and fixed an initial time step of $10^{-6}$. We used COMSOL's default backward difference formula method which then adaptively updated the time step beyond this initial value. In 1-D we used $10^4$ finite elements, and for higher-dimensional simulations used at least $10^4$ elements, though this varied for each geometry. Some restrictions were used on the maximum allowable time step to prevent behaviors such as the loss of modes in the initial perturbations. Convergence in time was checked by restricting the maximum time step, and convergence in space was determined via computing solutions across varied numbers of finite elements, and comparing the norm of solutions over time and space. 

We emphasize that in the cases of fast or non-monotonic domain growth, extreme care is needed due to the non-autonomous nature of the spatial operator. We note that one advantage of this choice of finite element software, as well as the restriction to dilational growth, is that it allows for simple implementations of growing manifolds where the growth is directed in particular directions in the ambient space. This is because the Laplace-Beltrami operator on a surface of dimension $N$ can be constructed from the Laplace operator in the ambient space $\mathbb{R}^{N+1}$, so that dilation of a particular coordinate in $\mathbb{R}^{N+1}$ allows a natural construction of the time-dependent Laplace-Beltrami operator on the surface. We note that there exist many other choices for numerical methods for such problems \cite{barreira2011surface, macdonald2013simple}.

Initial data is taken to be of the form $\mathbf{u}(0,\mathbf{x}) = (I+{\zeta}(\mathbf{x}))\mathbf{U^*}$, where $I$ is the identity matrix and ${\zeta} = \text{diag}(\zeta_1,\zeta_2)$ are normally distributed perturbations which are independent across space and for each morphogen. Specifically, for each $\mathbf{x}\in\Omega^*$ and $i=1,2$, we take $\zeta_i(\mathbf{x}) \sim \mathcal{N}(0,10^{-1})$. We have also considered smaller initial perturbations for each case, and note that whether or not a pattern persists despite transient periods of growth and decay is highly dependent on the size of the perturbation. For this reason, we use this reasonably large perturbation for all simulations, as the finite-time effects we study are intrinsically linked to observing growth of finite perturbations. For each geometry we show simulations using the same realization of the initial data throughout, though for a given size of perturbation (the variance of $\zeta$), we observe qualitatively similar dynamics for different realizations.

We consider two relevant sets to help visualize our instability criterion. We will consider these sets as functions of time. The first is a generalization of a time-dependent \emph{generalized Turing space} which is the set of all parameters for which Theorem \ref{Thm3} predicts an instability for some $k \geq 0$. Here we will consider as an example the non-negative parameters $(a,b) \in \mathbb{R}^2_+$ for the kinetics given by \eqref{Schnack}, but of course generalizing these definitions is straightforward. We then define such a space, for a given time $t$ as: $\mathcal{T}_t = \left \{(a,b) \in \mathbb{R}^2_+ | \cup_{k\geq 0} ( \{ t \} \cap \Ik )\neq \emptyset \right \}$. Of course one can generalize $\mathcal{T}_t$ to a set of times, say $\mathcal{S}([t_1,t_2])=\cup_{t_1 \leq t\leq t_2} \mathcal{T}_t$, rather than the singleton time, but for our purposes we prefer to think of these as sets parameterized by time. We separately plot the space corresponding to homogeneous instabilities, which are times $t\in \mathcal{I}_0$, so that one may consider Turing spaces which exclude these points. Similarly, for fixed parameters, we may be interested in plotting an analogue of the classical dispersion relation which indicates which wavenumbers $k$ are excited as a function of time $t$. We define this dispersion set to be: $\mathcal{K}_t = \{ k \in \mathbb{N} | \{ t \} \cap \Ik \neq \emptyset \}$.

We will compare these time-dependent sets to the quasi-static Turing space and dispersion relations. These are given by ignoring the non-autonomous nature of the system, and treating the domain length as a parameter in the classical static Turing conditions. While such quasi-static conditions are not formally valid, we will demonstrate cases where they do seem to capture the qualitative behaviour of the system, in addition to cases where they fail. We will also compare the observed modes from full numerical solutions, deduced via the Fast Fourier transform, with predictions from our linear stability theory in the 1-D setting to provide evidence of the applicability of our analysis. We will only plot the single Fourier mode with the largest absolute power, corresponding to the wavemode with the largest component of an expansion of the full spatial solution. If the variation of the solution across the domain is less than $1\%$ of its mean value, then we set the largest mode to $k=0$, essentially neglecting small variations from the homogeneous base state.

Finally, we again reiterate that the instability criterion given by Theorem \ref{Thm3} only tells us if the $k$th mode is growing or not at some time, but not directly the growth rate (though a bound on this can be inferred via Corollary \ref{cor1} which we will use to identify the fastest growing mode at a given time). Additionally, analysis of nonlinearities is necessary to determine conditions for whether or not a pattern fully develops and persists, or undergoes subsequent instabilities, such as peak-splitting. Nevertheless, we have exhaustively explored this condition numerically and confirmed that patterns typically develop if the parameter set is within the Turing space for a sufficiently long time, or equivalently that at least one mode remains unstable for a sufficient period. 

\subsection{Isotropic evolution of a line segment}\label{interval}
The simplest and most commonly studied example in the literature is a uniformly growing line segment. We define $\Omega(t) \subset \mathbb{R}$ by $\Omega(t) = [0,r(t)]$. The moving coordinate is $X=r(t)x$, for $x\in [0,1]$, and we find $\rho_k(t) = \frac{\pi^2 k^2}{r(t)^2}$ and $\mu(t)=r(t)$. We will use this simple geometric setting to explore various Turing spaces and dispersion relations for a variety of growth functions $r(t)$, to demonstrate how the instability regions change, particularly away from the well-studied case of slow growth. Our main aim is to show that the instability criterion in Theorem \ref{Thm3} can effectively capture instabilities in this time-dependent setting, and how it differs radically from either quasi--static approaches \cite{varea1999turing}, or the small corrections due to slow growth previously reported in the literature \cite{klika2017history, madzvamuse2010stability}. 

Under the criterion given in Theorem \ref{Thm3}, we find that the $k$th perturbation of the form \eqref{pert2} corresponding to $\rho_k(t)= \frac{\pi^2 k^2}{r(t)^2}$ is unstable over some interval $t\in \mathcal{I}_k$, provided that the inequality
\begin{equation}\begin{aligned}\label{excond1d}
& \det(J) - \left( d_2 J_{11} + d_1 J_{22} \right)\frac{\pi^2k^2}{r^2}+ d_1d_2\frac{\pi^4k^4}{r^4} \\
& \qquad
 < - \frac{\ddot{r}}{r} - \frac{\dot{r}}{r}\left( \left( d_1 + d_2\right)\frac{\pi^2 k^2}{r^2} - \mathrm{tr}(J) \right) \\
& \qquad\qquad + \max\left\lbrace \frac{\dot{r}}{r} \frac{\dot{J}_{12}}{J_{12}} - J_{12}\frac{d}{dt}\left( \frac{d_1 \pi^2 k^2}{r^2 J_{12}}-\frac{ J_{11}}{J_{12}}\right) , \frac{\dot{r}}{r} \frac{\dot{J}_{21}}{J_{21}} - J_{21}\frac{d}{dt}\left( \frac{d_2 \pi^2 k^2 }{r^2 J_{21}} - \frac{J_{22}}{J_{21}}\right)\right\rbrace 
\end{aligned}\end{equation}
holds for all $t\in\mathcal{I}_k$. In the special case where $r(t)=L$ for constant $L>0$, hence the domain is static, the condition \eqref{excond1d} reduces to 
\begin{equation}
\det(J) - \left( d_2 J_{11} + d_1 J_{22} \right)\frac{\pi^2k^2}{L^2}+ d_1d_2\frac{\pi^4k^4}{L^4} < 0\,,
\end{equation}
which is exactly the classical Turing condition for a static one dimensional domain $[0,L]$.

A specific form of growth which is somewhat popular in the literature is exponential growth, which takes the form $r(t) = r(0)\exp(s t)$, $s>0$. In addition to biological plausibility, another reason for the popularity of exponential growth is that it allows for the volume expansion term to take the form $\frac{\dot{\mu}}{\mu} = N s$ (where $N$ is the dimension of the space domain), a constant, which greatly simplifies the dynamics of the spatially uniform system. Exponential isotropic growth of surfaces in $\mathbb{R}^3$ was extensively studied in \cite{toole2014pattern}, albeit under the assumption of a time-independent base state for \eqref{uniform}.

\begin{figure}
\centering

\includegraphics[width=0.3\textwidth]{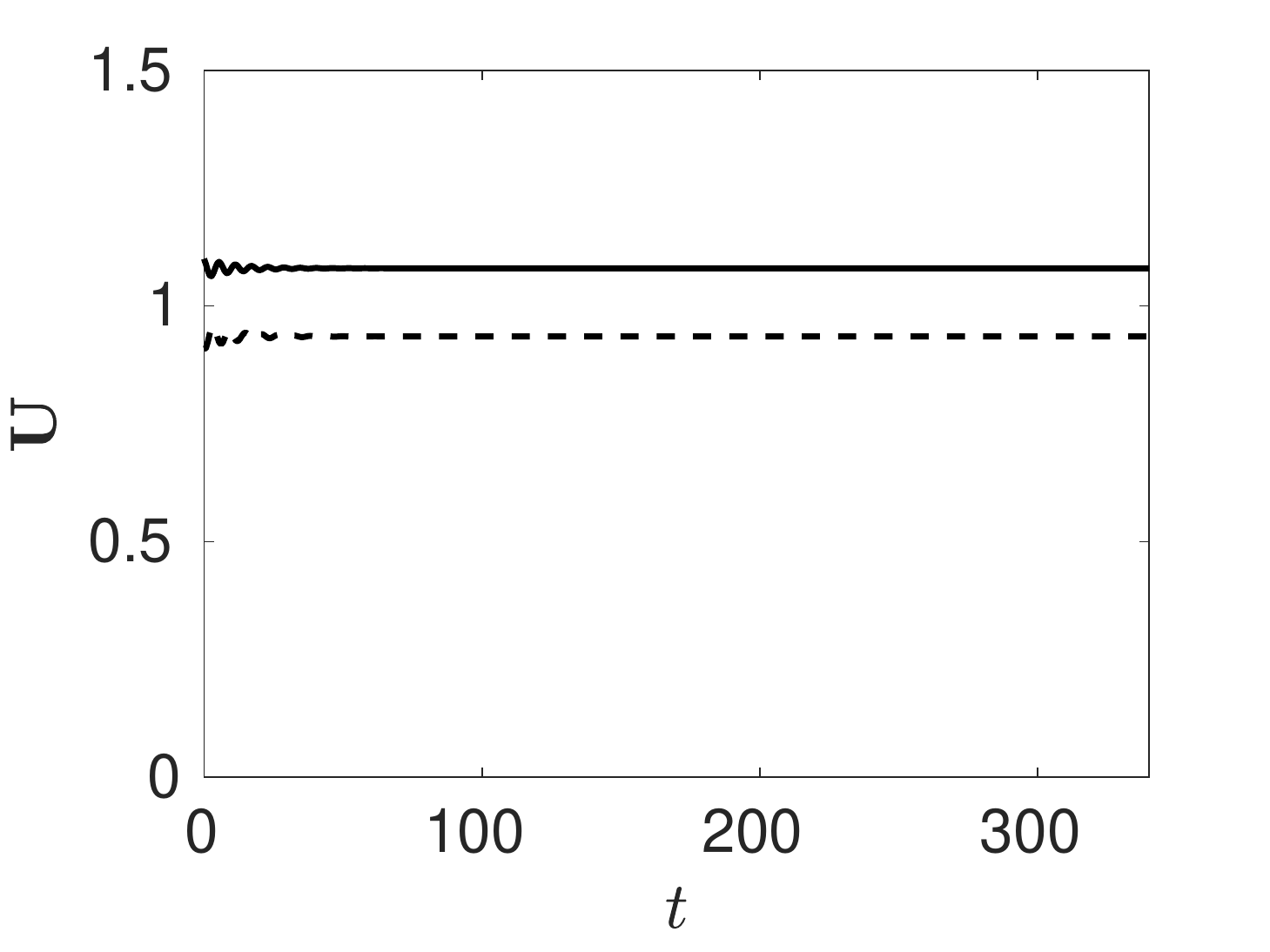}
\includegraphics[width=0.3\textwidth]{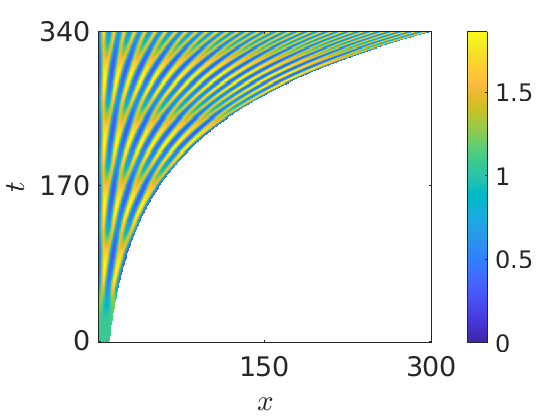}
\includegraphics[width=0.3\textwidth]{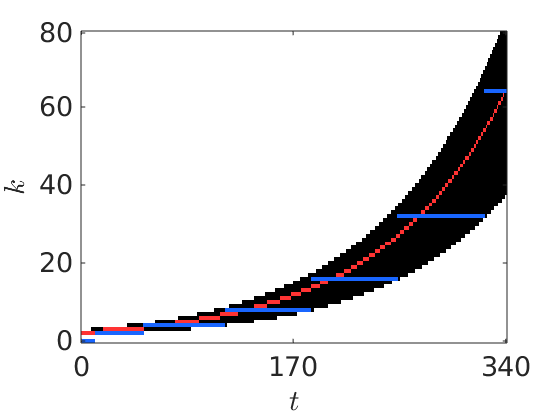}

(a)i \hspace{4cm} (a)ii \hspace{4cm} (a)iii

\includegraphics[width=0.3\textwidth]{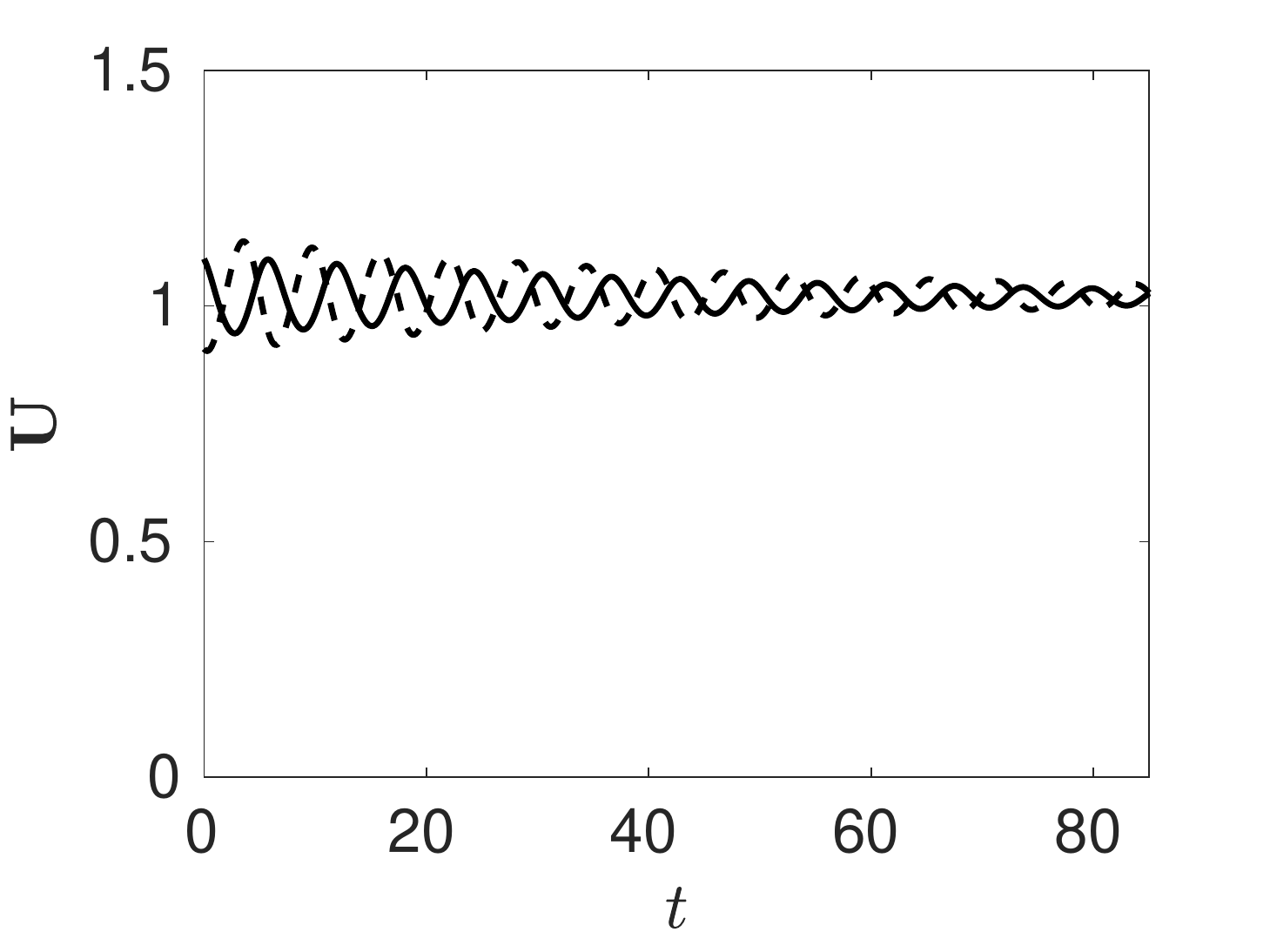}
\includegraphics[width=0.3\textwidth]{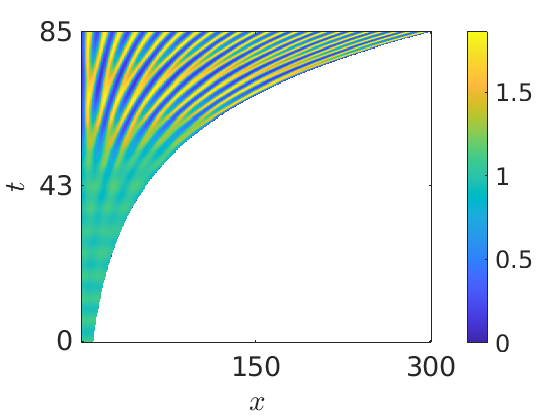}
\includegraphics[width=0.3\textwidth]{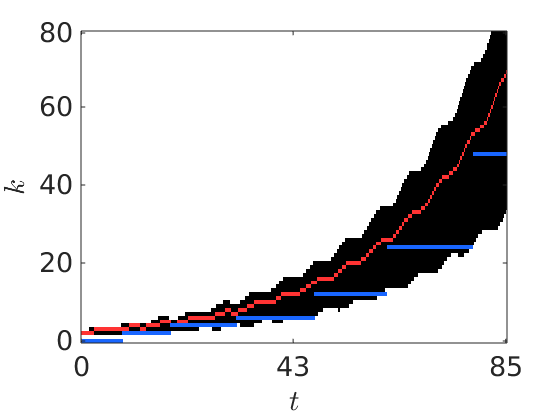}

(b)i \hspace{4cm} (b)ii \hspace{4cm} (b)iii 

\includegraphics[width=0.3\textwidth]{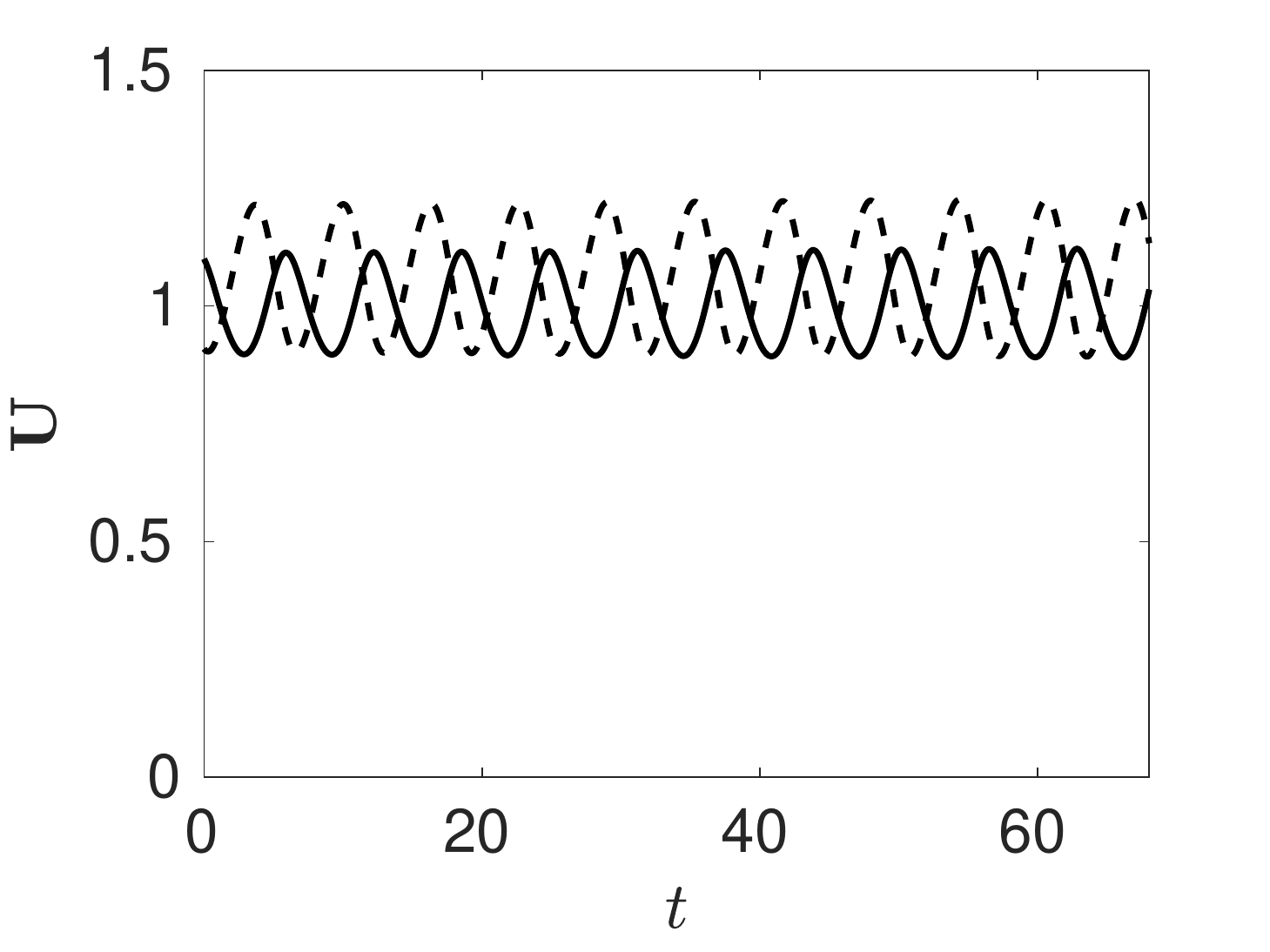}
\includegraphics[width=0.3\textwidth]{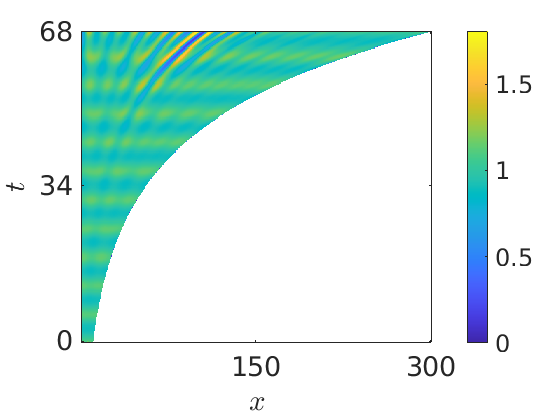}
\includegraphics[width=0.3\textwidth]{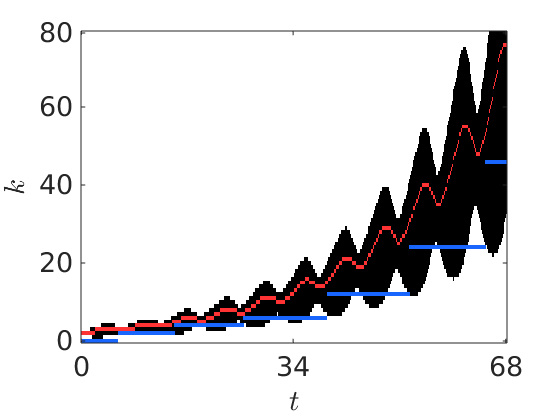}

(c)i \hspace{4cm} (c)ii \hspace{4cm} (c)iii 

\includegraphics[width=0.3\textwidth]{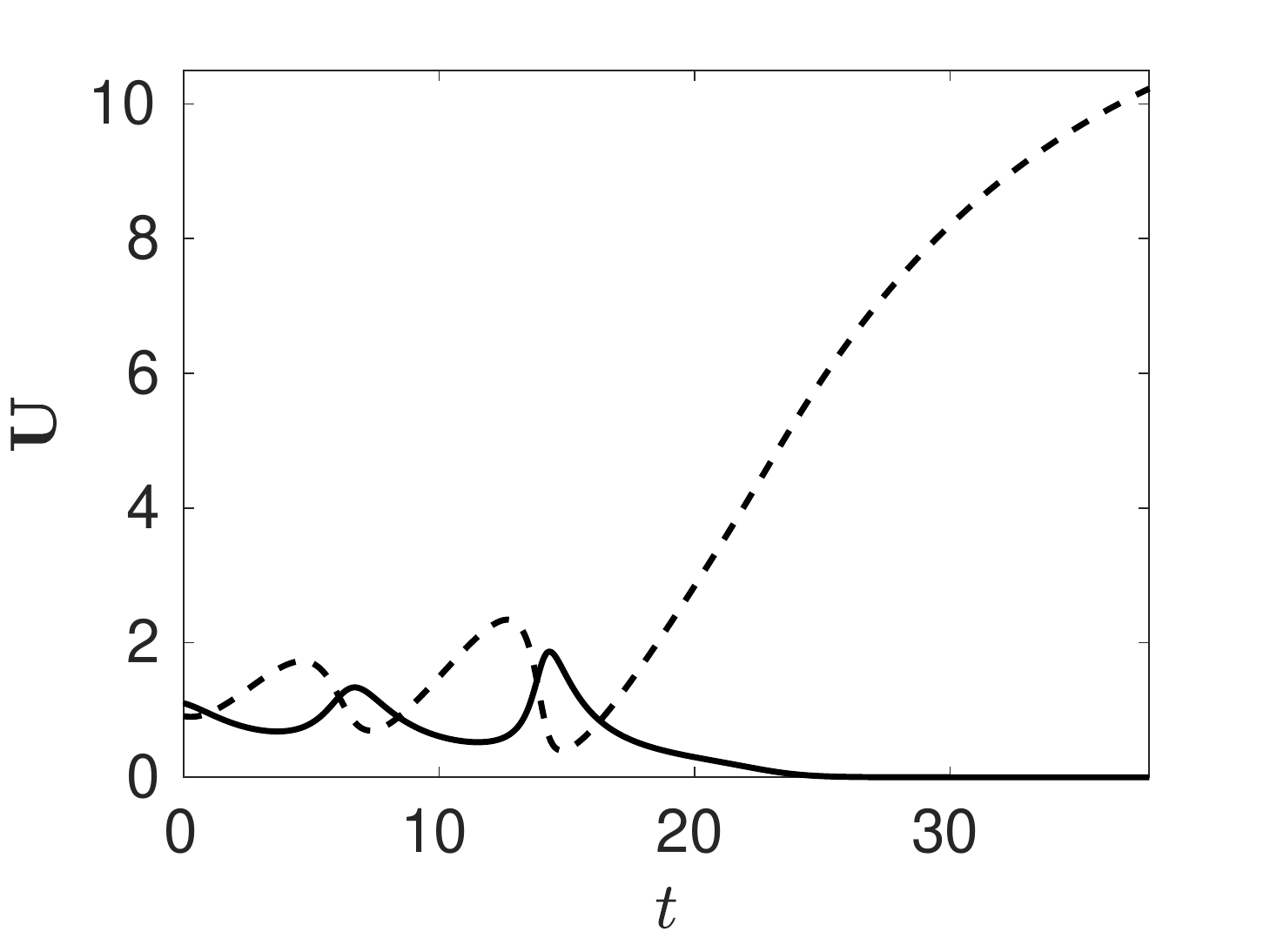}
\includegraphics[width=0.3\textwidth]{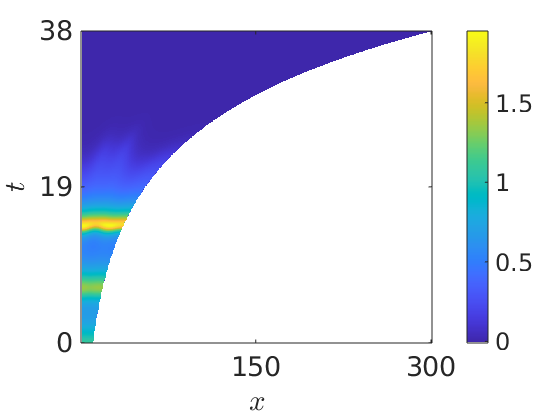}
\includegraphics[width=0.3\textwidth]{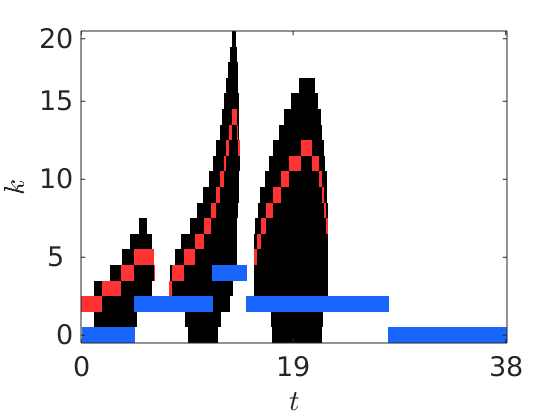}

(d)i \hspace{4cm} (d)ii \hspace{4cm} (d)iii

\vspace{-0.1in}
\caption{Plots corresponding to the kinetics \eqref{Schnack} with parameters $a=0$, $b=1.1$, $d_1=1$, and $d_2=10$. The domain is taken to grow as $r(t) = 10\exp(st)$ for growth rates $s=0.01, 0.04, 0.05$, and $0.09$ in rows (a)-(d) respectively. In all simulations we take the final time such that the domain has grown to $30$ times its initial size. In column (i) we plot solutions of the homogeneous base state solution of \eqref{uniform} over time, with $U_1$ given by the dark line and $U_2$ by the dashed line. In column (ii) we show plots of the PDE solution $u_1$ over space and time. In column (iii) we plot the dispersion set $\mathcal{K}_t$ in black, with the theoretically maximally growing mode in red and the largest frequency component of the FFT of $u_1(x,t)$ from the full numerical solution in blue. NB: The temporal and mode axes have different ranges for different growth rates.\label{expplots}}
\end{figure}

We choose parameters of the kinetics \eqref{Schnack} and an initial domain of size $r(0)=10$ for which the system would be on the boundary of the Turing space for a static domain, only admitting a single unstable wavenumber $k=1$. We then simulate \eqref{growthgen} until the domain has grown to $r(t_f)=30r(0)$. We show our results in Fig.~\ref{expplots}. In each row, we plot solutions to the uniform base state $\mathbf{U}$ from Equation \eqref{uniform} in the first column, the PDE solution $u_1$ in the second column, and the dispersion set $\mathcal{K}_t$ in the third column, with each row demonstrating an increasing growth rate. We observe that the dynamics of the uniform base state plays a substantial role in determining both $\mathcal{K}_t$, and consequently the evolution of the pattern. 

As exponential growth leads to an autonomous planar system, we observe that the decaying oscillations in Fig.~\ref{expplots}(b)i are due to a stable spiral, and that the oscillations in Fig.~\ref{expplots}(c)i are due to a Hopf bifurcation which has created a stable limit cycle. These decaying and persistent oscillations have an impact on the timescale over which a pattern can emerge, and we only see the onset of a pattern near the end of the simulation time in Fig.~\ref{expplots}(c)ii. The fastest growth rate results in a uniform base state which grows far from the original kinetic equilibrium, and pattern formation is no longer possible. As the set of unstable wavenumbers grows exponentially, there is a hysteresis effect such that if a perturbation has not left the base state sufficiently early on, then a pattern cannot form, whereas a developed pattern persists. The quasi-static Turing space is identical to that shown in Fig.~\ref{expplots}(a)iii, and due to the choice of the growth, is independent of the growth rate (up to relabelling time). Hence the qualitative differences in the third column are all manifestations of the non-autonomous nature of the growth. 

The largest modal components observed (in blue) roughly follow a peak-splitting mode doubling process, which is most apparent for the slowest growth case in Fig.~\ref{expplots}(a)ii, but breaks down for faster growth as in Fig.~\ref{expplots}(c)ii, as anticipated by \cite{ueda2012mathematical}. While the linear analysis does not precisely predict these observed modes, it does give a qualitative insight into the processes leading to these patterned states. Specifically, the numerically observed modes all follow a period of time wherein that specific mode has been unstable and allowed to grow away from the homogeneous base state. Additionally, the lower-frequency solutions seen in the faster growing domains can also be explained as, due to oscillations, the system does not remain in a state admitting a given unstable mode for nearly as long as it does for the slower growth cases.

\begin{figure}
\centering
\includegraphics[width=0.32\textwidth]{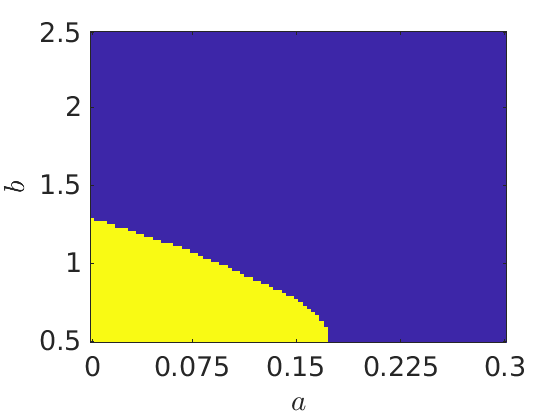}
\includegraphics[width=0.32\textwidth]{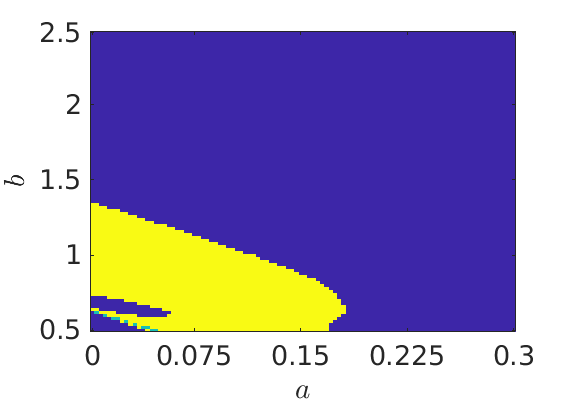}
\includegraphics[width=0.32\textwidth]{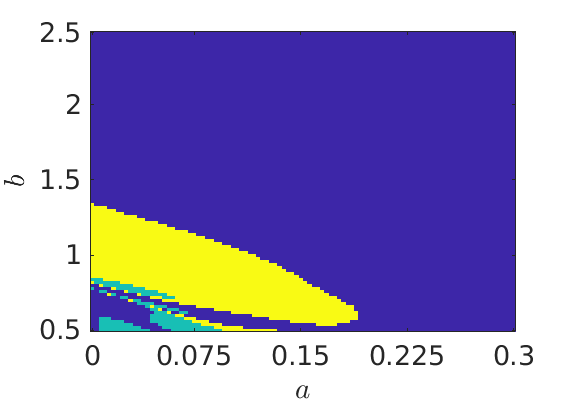}

\hspace{0.2cm} (i) \hspace{4.7cm} (ii) \hspace{4.7cm} (iii) 

\vspace{-0.1in}
\caption{Turing spaces $\mathcal{T}_t$ corresponding to the kinetic parameters in \eqref{Schnack} with parameters $a=0$, $b=1.1$, $d_1=1$, and $d_2=10$. The domain is taken to grow exponentially like $r(t) = 10\exp(0.01t)$, and the Turing spaces are computed at $t=0, 10, 20$ in columns (i)-(iii), respectively. A parameter set which has an unstable mode in $k=1,\dots,200$ at time $t$ is given in yellow (light), a point for which $t \in \mathcal{I}_0$ (i.e.~homogeneous instability) is in teal (medium), and all other points are colored blue (dark) which indicates stability of the homogeneous state.\label{turingplots}}
\end{figure}

Next we consider Turing spaces, $\mathcal{T}_t$, at different instances in time, in Fig.~\ref{turingplots}. The first column shows the initial Turing space, which is equivalent to the quasi-static space obtained by just incorporating the growth rate into the kinetics \cite{madzvamuse2010stability}, and specifically (i) is equivalent to the static Turing space without growth. As expected, we observe little change from a small kinetic addition at $t=0$, but for larger times we see previously-unstable regions become stable, and regions becoming unstable to homogeneous perturbations, as well as new regions becoming unstable as the Turing space expands around the edges. Such observations are in line with the results of \cite{klika2017history}, though we remark that these spaces are not equivalent as our approach accounts for discrete wavenumbers, and does not need the assumption of slow growth.

\begin{figure}
\centering
\includegraphics[width=0.3\textwidth]{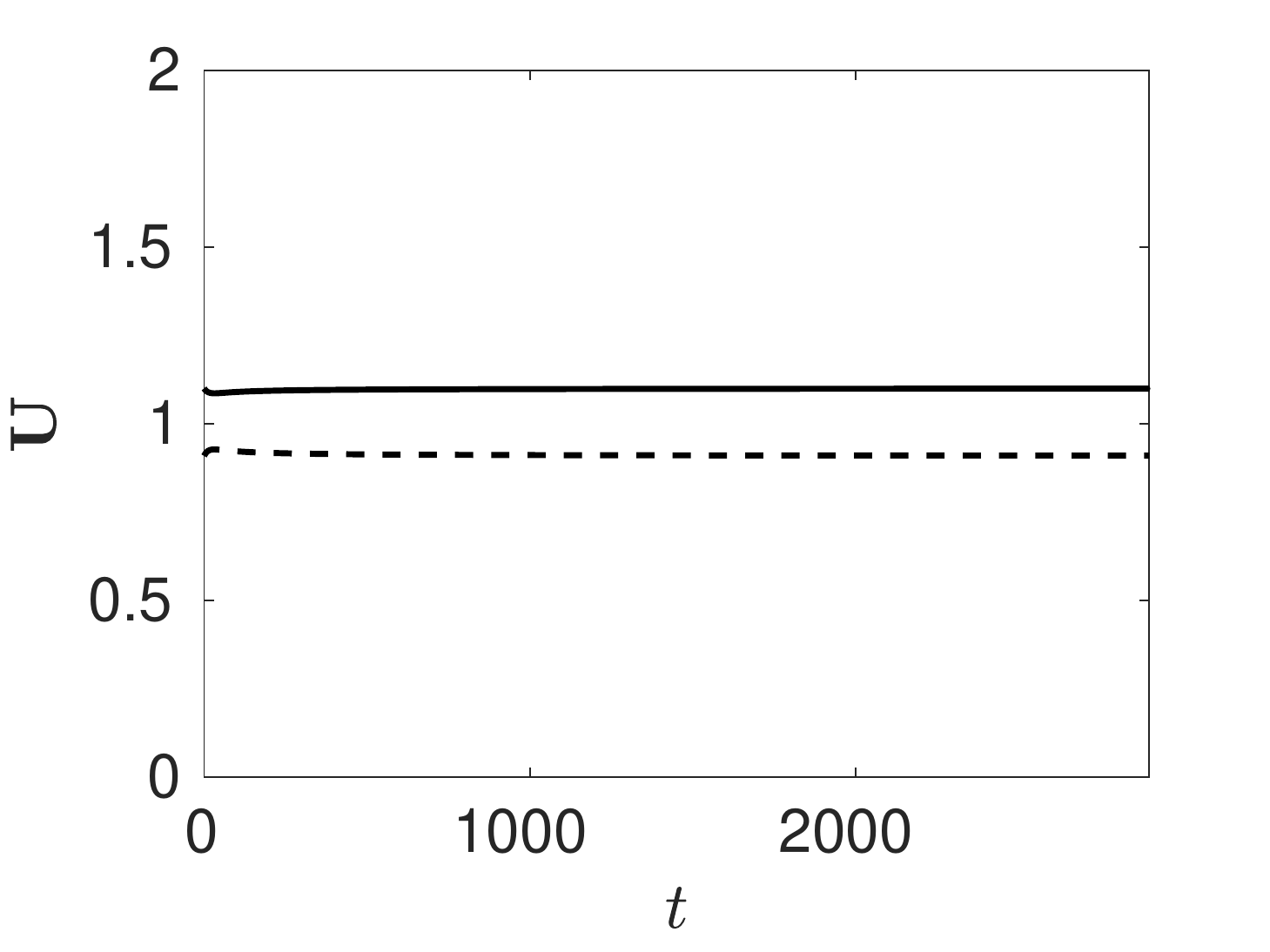}
\includegraphics[width=0.3\textwidth]{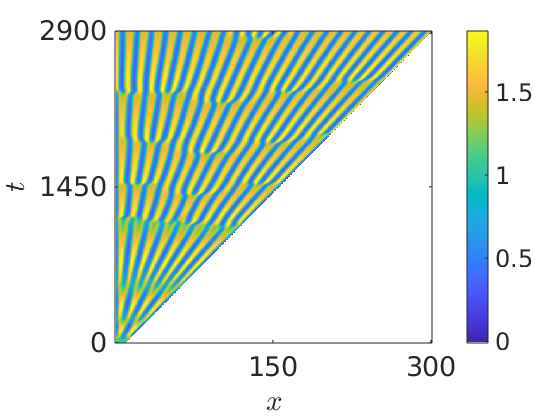}
\includegraphics[width=0.3\textwidth]{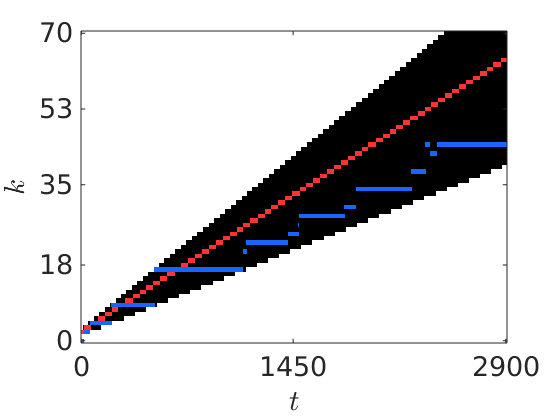}

(a)i \hspace{4cm} (a)ii \hspace{4cm} (a)iii 

\includegraphics[width=0.3\textwidth]{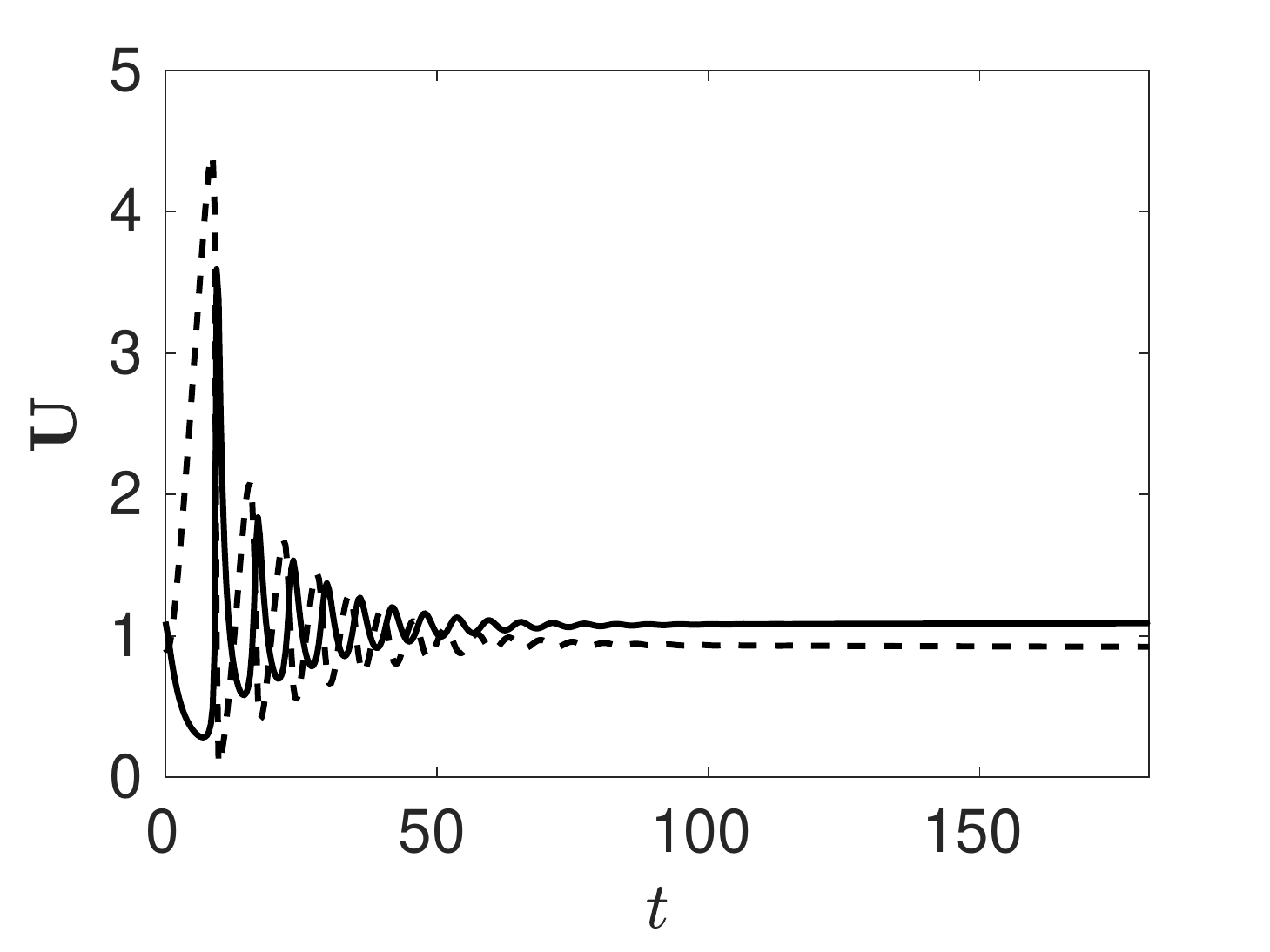}
\includegraphics[width=0.3\textwidth]{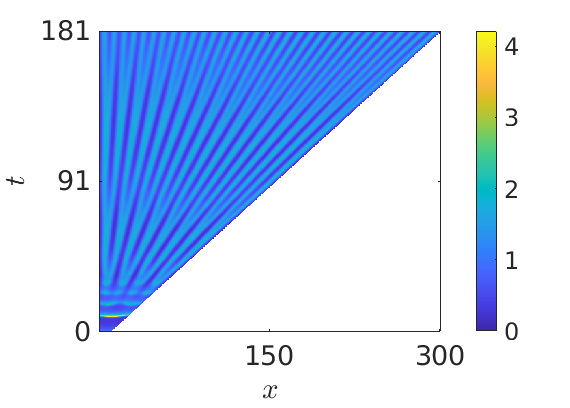}
\includegraphics[width=0.3\textwidth]{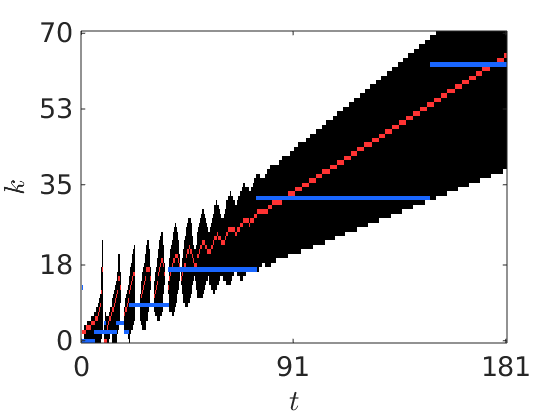}

(b)i \hspace{4cm} (b)ii \hspace{4cm} (b)iii 

\includegraphics[width=0.3\textwidth]{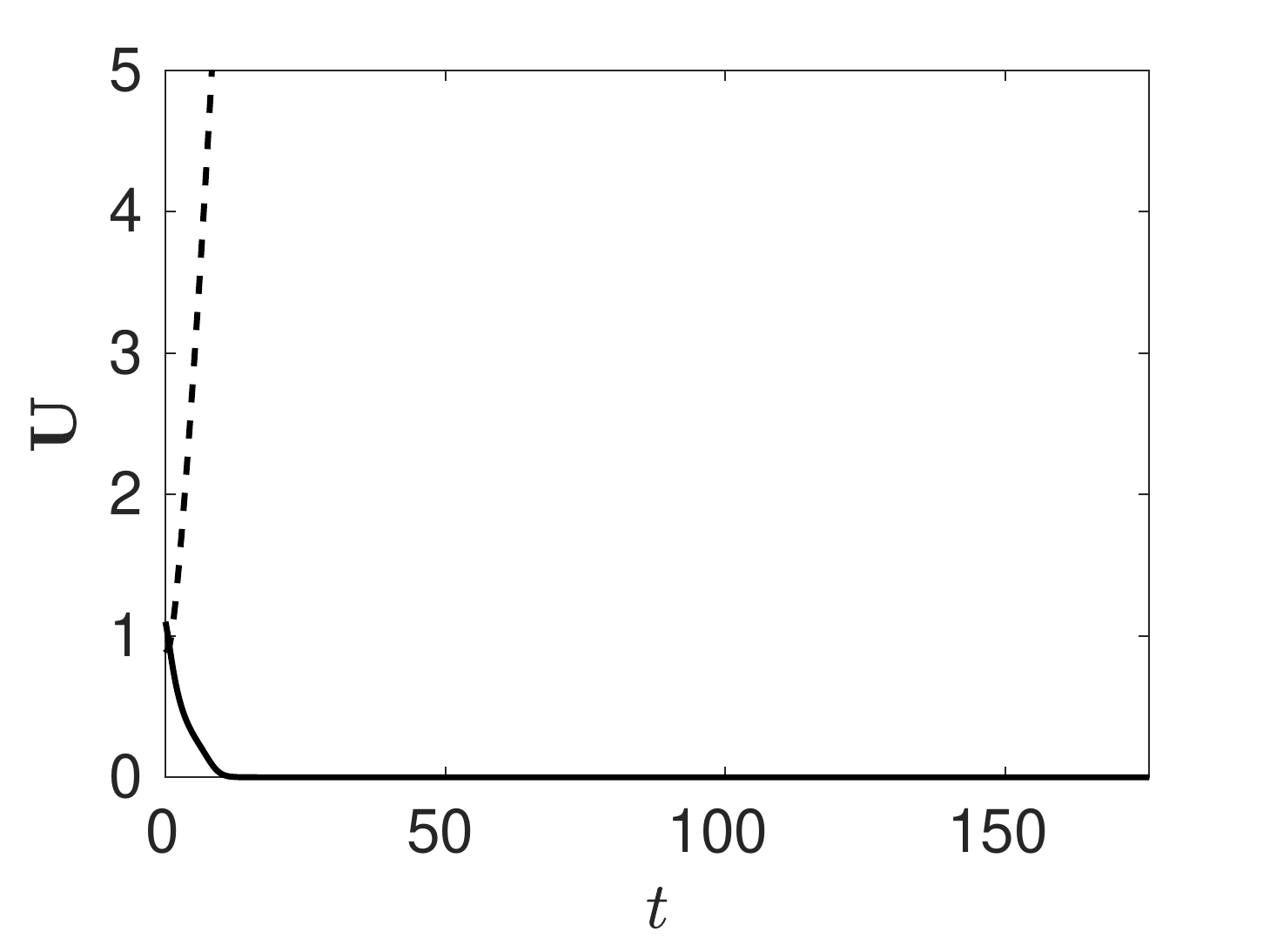}
\includegraphics[width=0.3\textwidth]{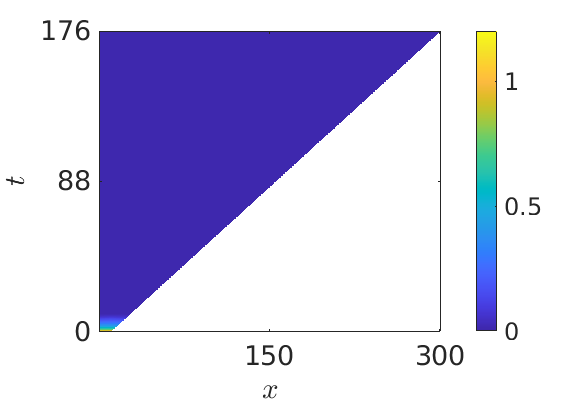}
\includegraphics[width=0.3\textwidth]{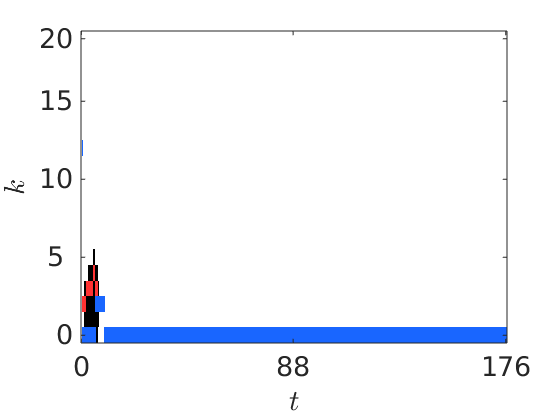}

(c)i \hspace{4cm} (c)ii \hspace{4cm} (c)iii 

\vspace{-0.1in}
\caption{Plots corresponding to the kinetics \eqref{Schnack} with parameters $a=0$, $b=1.1$, $d_1=1$, and $d_2=10$. The domain is taken to grow as $r(t) = 10(1+st)$ for growth rates $s=0.01, 0.16$, and $0.165$ in rows (a)-(c), respectively. In all simulations we take the final time such that the domain has grown to $30$ times its initial size. In column (i) we plot solutions of the homogeneous base state solution of \eqref{uniform} over time, with $U_1$ given by the dark line and $U_2$ by the dashed line. In column (ii) we show plots of the PDE solution $u_1$ over space and time. In column (iii) we plot the dispersion set $\mathcal{K}_t$ in black, with the theoretically maximally growing mode in red and the largest frequency component of the FFT of $u_1(x,t)$ from the full numerical solution in blue. NB: The temporal and mode axes have different ranges for different growth rates.\label{linplots}}
\end{figure}

We consider linear growth in Fig.~\ref{linplots}, with increasing growth rates in each subsequent row. Other than similar transient effects to before, the final modes observed are similar in each case except as the growth rate surpasses $s=0.16$. Slightly beyond this point, by $s=0.165$, the steady state of the uniform base states is no longer stable, and instead we see in Fig.~\ref{linplots}(c)i that $U_1$ tends toward $0$, and $U_2$ diverges to infinity. We remark that this destabilization of the uniform base state's long-time behavior can be observed in both the dispersion sets and space-time plots. The concentration of $u_2$ increases in time uniformly as the domain expands. This phenomenon is inherently non-autonomous, and depends strongly on the initial condition; for other choices of $\boldmath{U}(0)$ we observe different behaviors.  In Fig.~\ref{linplots}(b)ii,iii, we see sharp oscillations with increasing amplitudes before a pattern is allowed to form, suggesting a kind of excitability inherent in the transient dynamics.

\begin{figure}
\centering
\includegraphics[width=0.3\textwidth]{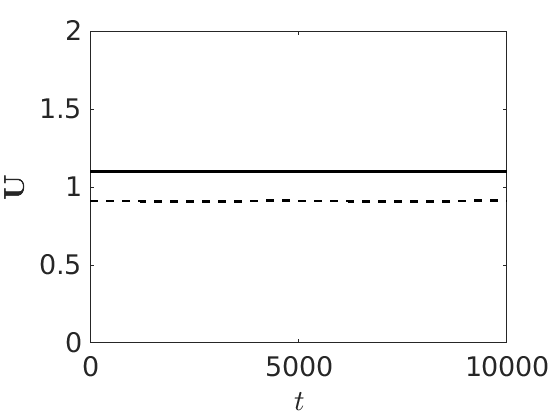}
\includegraphics[width=0.3\textwidth]{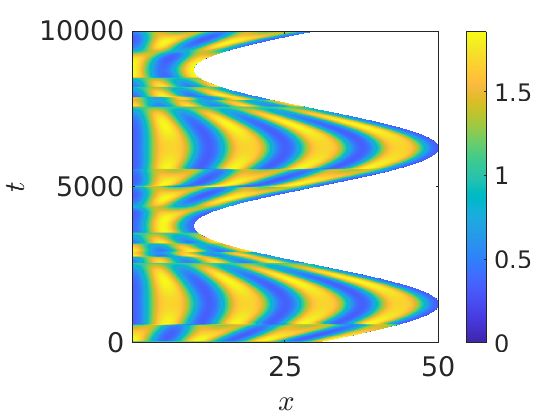}
\includegraphics[width=0.3\textwidth]{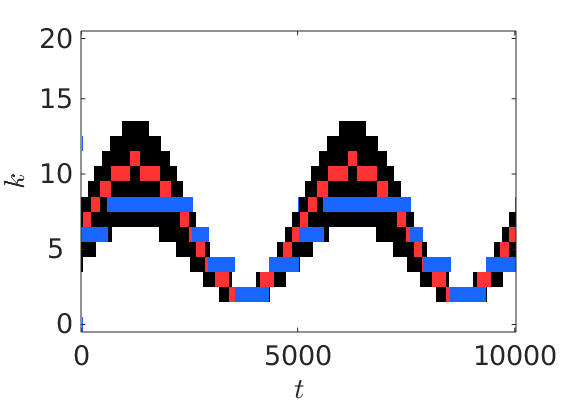}

(a)i \hspace{4cm} (a)ii \hspace{4cm} (a)iii 

\includegraphics[width=0.3\textwidth]{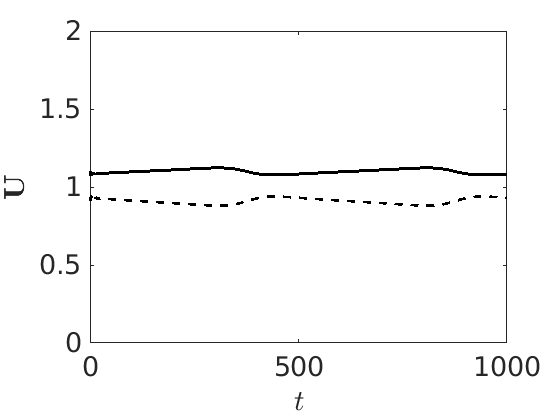}
\includegraphics[width=0.3\textwidth]{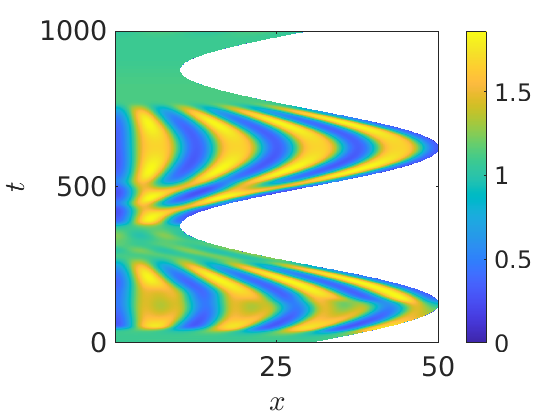}
\includegraphics[width=0.3\textwidth]{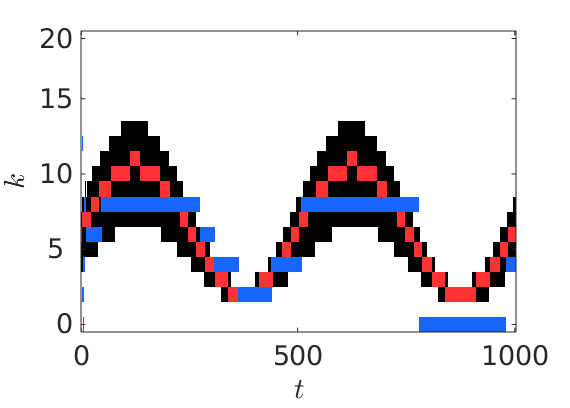}

(b)i \hspace{4cm} (b)ii \hspace{4cm} (b)iii

\includegraphics[width=0.3\textwidth]{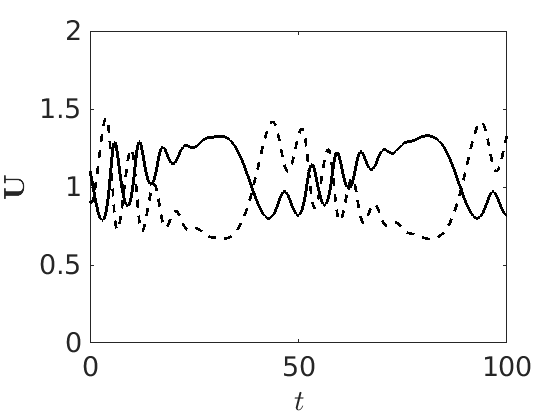}
\includegraphics[width=0.3\textwidth]{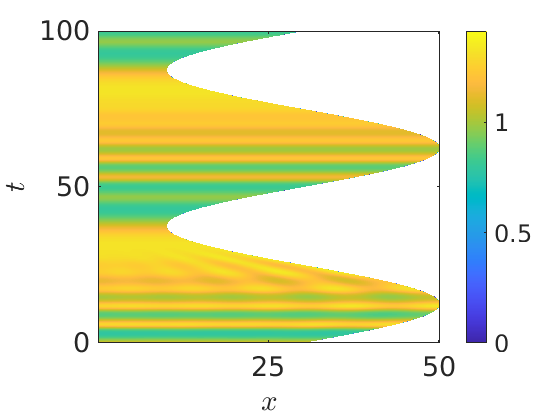}
\includegraphics[width=0.3\textwidth]{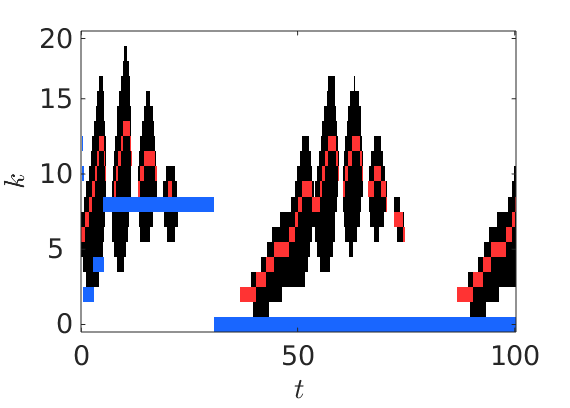}

(c)i \hspace{4cm} (c)ii \hspace{4cm} (c)iii 

\vspace{-0.1in}
\caption{Plots corresponding to the kinetics \eqref{Schnack} with parameters $a=0$, $b=1.1$, $d_1=1$, and $d_2=10$. The domain is taken to evolve as $r(t) = 30(1+(2/3)\sin(4\pi t /t_f))$ for different final times (and hence growth rates) $t_f=10^4, 10^3$, and $10^2$ in rows (a)-(c) respectively. In column (i) we plot solutions of the homogeneous base state solution of \eqref{uniform} over time, with $U_1$ given by the dark line and $U_2$ by the dashed line. In column (ii) we show plots of $u_1$ over space and time. In column (iii) we plot the dispersion set $\mathcal{K}_t$ in black, with the theoretically maximally growing mode in red and the largest frequency component of the FFT of $u_1(x,t)$ from the full numerical solution in blue. NB: The temporal axes have different ranges for different growth rates.\label{sinplots}}
\end{figure}

There are a wide variety of more complex kinds of domain evolution one could consider, especially if we allow expansion and contraction rather than monotonic growth. As a simple example of this, we consider periodic growth and contraction given by a sinusoidal function in Fig.~\ref{sinplots}. Again, the set $\mathcal{K}_t$ in the slow case (a)iii is identical to the quasi-static approximation, and the observed modes follow this reasonably well. While the dispersion set only has extremely small changes in (b)iii, we see that the base states in this case slowly oscillate in (b)i, and that the pattern seemingly disappears during the height of contractions in (b)ii, only to reappear later. In the case of more rapid oscillations, the uniform base states oscillate irregularly, and spatial pattern formation is only intermittent (see $t\in [25,40]$) and fails to persist. One key observation that is clear from the plots of $\mathcal{K}_t$, is that contraction of the domain is a highly stabilizing effect, as during the contracting period of Figs.~\ref{sinplots}(b,c)iii, we see substantially fewer unstable modes than during the expanding phase.

While growth has been heavily studied in the reaction-diffusion literature, contraction or other complex forms of domain evolution have not been as thoroughly explored. While these periodically expanding and contracting domains may be somewhat exaggerated from realistic examples, we remark that large contractions have been observed in the blastocyst stage of mice embryos \cite{shimoda2016time}. Such contractions can lead to a decrease in volume of as much as $20\%$, and potentially play important roles in morphogenesis, likely altering local chemical concentrations in addition to mechanical effects. While most biological media are undergoing expansion and growth, we anticipate that a more nuanced and accurate representation of morphogenesis will necessarily involve processes such as contraction. Indeed, if the domain contraction is strong enough, it forces the solution to be spatially heterogeneous (yet still oscillatory in time) and can suppress future patterning. We comment further on this point later.

In many of our plots and simulations, the blue line corresponding to the maximal mode in the FFT lies mostly within the shaded instability region. In other cases, particularly those where the evolution of the instability region is not strictly monotone in time, there is a slight lag in the full nonlinear system in responding to instabilities, which is why the maximal mode will sometimes extend outside of the black region. It is important to note that, at the onset of instability, the maximal mode resulting from the instability lies in the shaded region, since the new pattern is set by the instability. This maximal mode can then persist for a time even as the instability region shifts, due to nonlinear terms stabilizing the fully nonlinear simulations. However, the maximal mode gradually loses stability, and a new dominant unstable mode is selected within the shaded instability region which is valid at that time. This process continues. 

For more extreme cases, such as that shown in Fig. \ref{sinplots}(c), there is a strong contraction of the domain leading to all modes becoming stable. This stabilizes a uniform solution, and upon the later expansion of the domain, the concentration remains uniform yet oscillates. The reason for this is that upon the second expansion of the domain, the uniform solution is not acted on by a spatial perturbation (as it was going into the first expansion). Without a spatial perturbation, the higher spatial modes are never activated, despite the domain change, and there is hence no Turing instability. This is in particular seen in Fig. \ref{sinplots}(c)(ii), where after the first domain contraction the later dynamics are spatially uniform and simply oscillate in time as the density of the chemicals change along with the domain size. Therefore, it appears as those through solutions maintaining some spatial heterogeneity over time are susceptible to later bifurcations leading to new spatial patterns (such as peak splitting leading to a doubling of localized structures at each bifurcation), as the spatial variations provide enough noise to permit successive Turing instabilities as the domain grows. However, those solutions for which domain evolution suppresses spatial heterogeneity resulting in a spatially uniform state, subsequent dynamics associated to domain evolution do not appear sufficient to initiate later Turing instabilities. While our linear instability analysis compares well with full numerical simulations in this regard, a more rigorous nonlinear analysis focused specifically on this behaviour would possibly elucidate this suppression of pattern formation.

\subsection{Isotropic evolution of an excitable medium}\label{sec_FHN}
We now consider the reaction kinetics \eqref{FHN} with parameters corresponding to the Turing (but not Turing--Hopf) space for a static domain (see \cite{sanchez2018turing} for bifurcation diagrams). We consider linear and step wise growth functions to demonstrate the impact that an excitable system has on pattern formation. Theorem \ref{Thm3} is also useful in determining when spatial modes can destabilize a homogeneous but oscillating base state on a static domain, such as that which occurs generically when the kinetics have undergone a Hopf bifurcation. We remark that linear analysis is insufficient to completely characterize instabilities which involve competition between both unstable Turing and Hopf modes, and generally the behavior can depend on the initial perturbation in addition to the parameters. Nevertheless, we demonstrate here that Theorem \ref{Thm3} can give some insight into when these Hopf modes can occur, which is a prerequisite to both purely oscillatory or spatiotemporal dynamics involving the competition of modes from both kinds of instabilities. Additionally we demonstrate how the solution to Equation \eqref{uniform} precisely determines the possibility of oscillatory dynamics.

\begin{figure}
\centering
\includegraphics[width=0.3\textwidth]{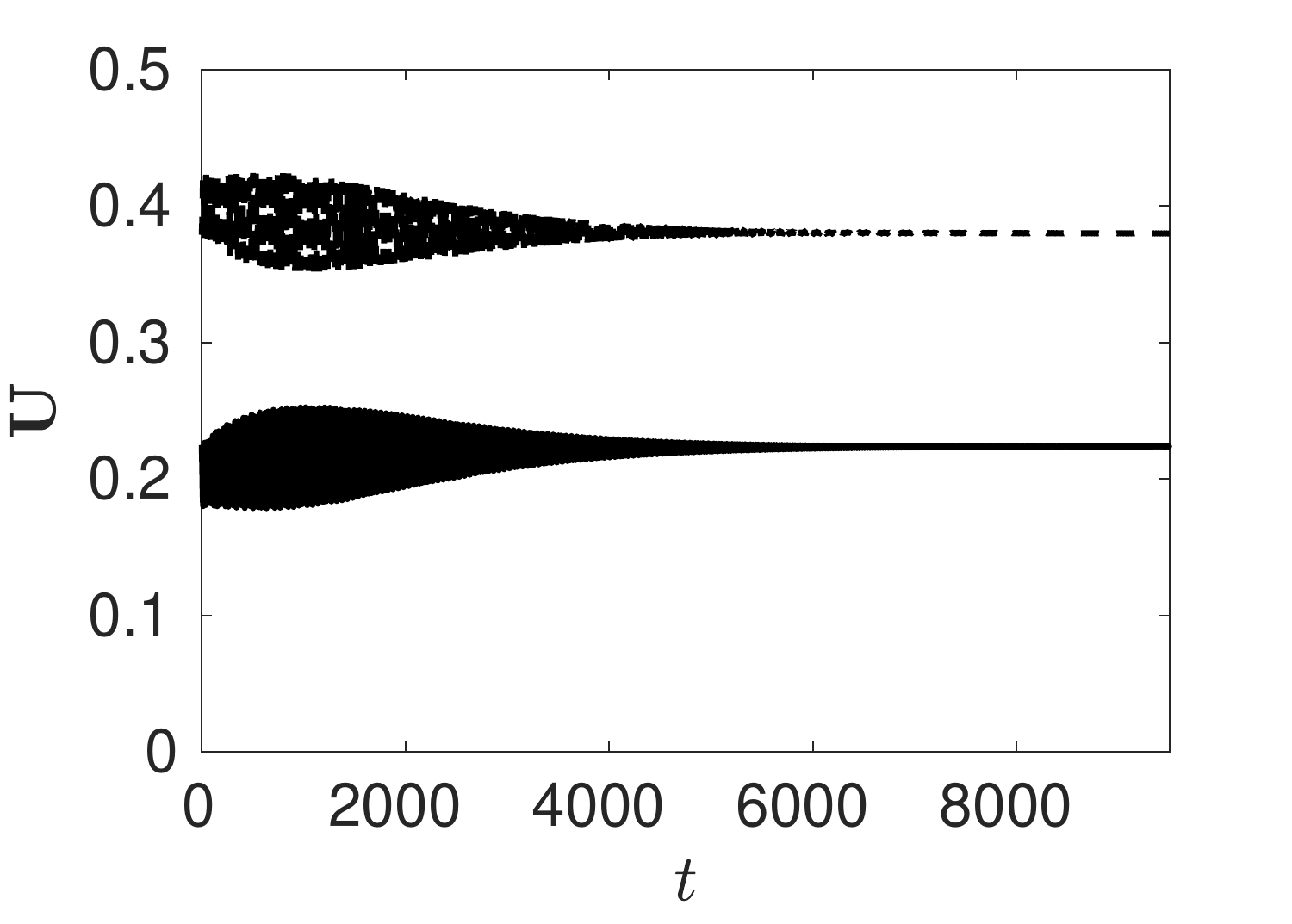}
\includegraphics[width=0.3\textwidth]{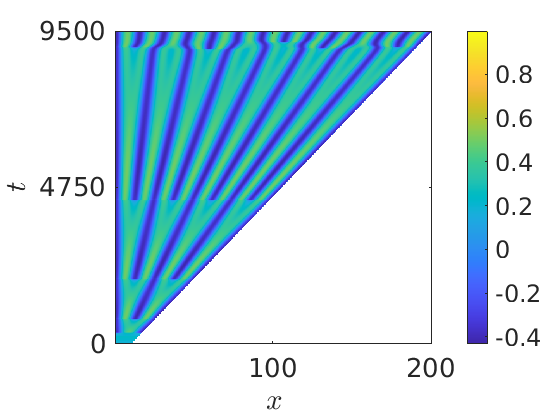}
\includegraphics[width=0.3\textwidth]{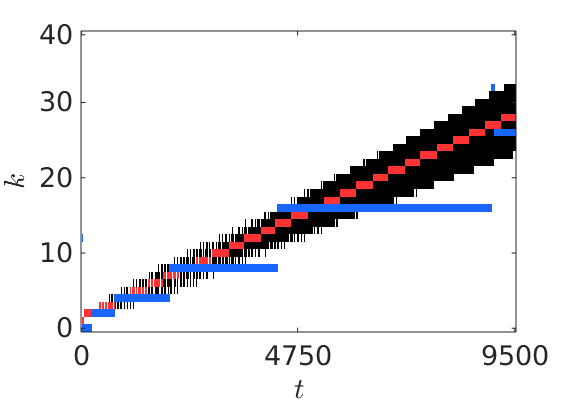}

(a)i \hspace{4cm} (a)ii \hspace{4cm} (a)iii

\includegraphics[width=0.3\textwidth]{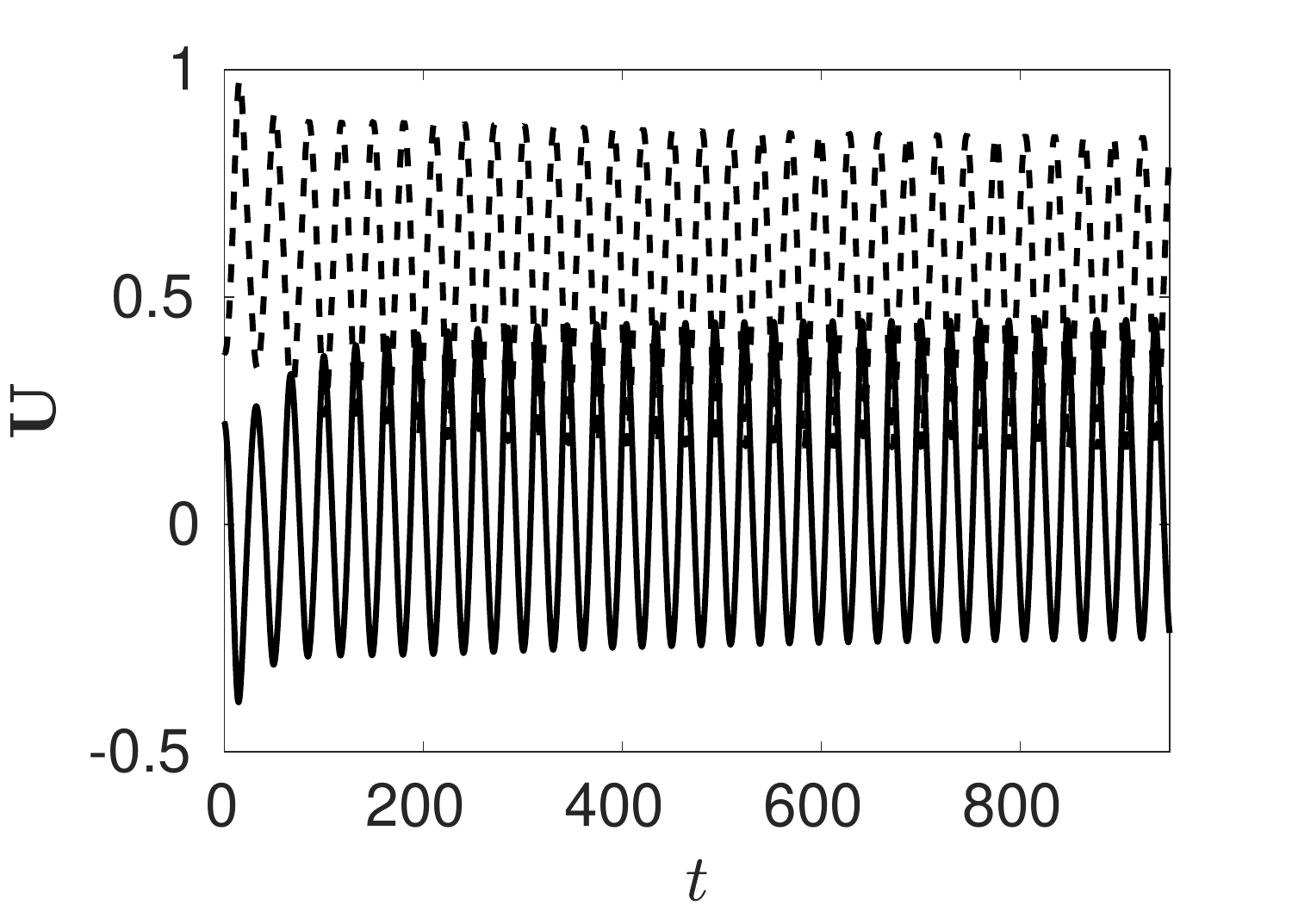}
\includegraphics[width=0.3\textwidth]{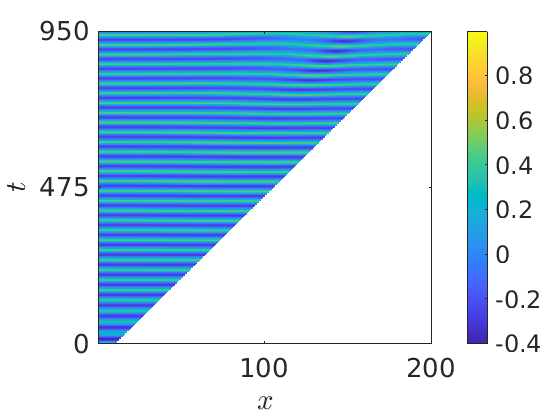}
\includegraphics[width=0.3\textwidth]{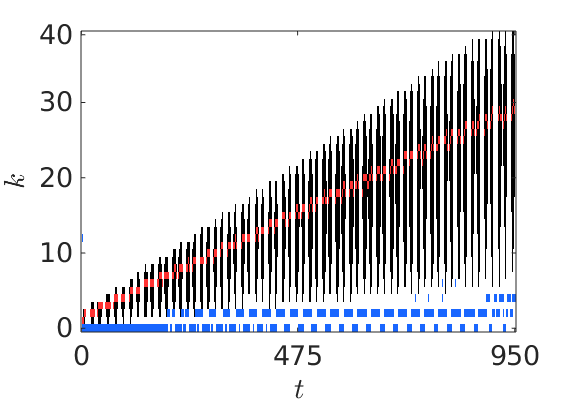}

(b)i \hspace{4cm} (b)ii \hspace{4cm} (b)iii 

\vspace{-0.1in}
\caption{Plots corresponding to the kinetics \eqref{FHN} with parameters $a=0.6$, $b=0.99$, $c=1.02$, $i_0=0.6$, $d_1=1$, and $d_2=1.7$. The domain is taken to grow as $r(t) = 10(1+st)$ up to a final size $r(t_f)=200$ for (a) $s=0.002$ and (b) $0.02$. In column (i) we plot solutions of the homogeneous base state solution of \eqref{uniform} over time, with $U_1$ given by the dark line and $U_2$ by the dashed line. In column (ii) we show plots of the PDE solution $u_1$ over space and time. In column (iii) we plot the dispersion set $\mathcal{K}_t$ in black, with the theoretically maximally growing mode in red and the largest frequency component of the FFT of $u_1(x,t)$ from the full numerical solution in blue. NB: The temporal axes have different ranges for different growth rates.\label{FHNlinplots}}
\end{figure}

In Fig.~\ref{FHNlinplots} we consider the linear growth case. For very small growth rates we recover the quasi-static dispersion relation (not shown), but as the growth rate is increased we observe transient oscillations as the base state slowly spirals back to its steady state value (Fig.~\ref{FHNlinplots}(a)). As the growth rate is increased further, the initial disturbance from the kinetic steady state leads to a sustained oscillation (Fig.~\ref{FHNlinplots}(b)i), which persists even when the growth is no longer substantially influencing the dynamics. The oscillatory base state leads to a dispersion set which is no longer a simply connected set, such that modes oscillate between growing for some time and decaying for others, which prevents the formation of spatial patterns. This occurs because even without growth, the base state dynamics are excitable such that both a stable steady state and a stable limit cycle coexist for these parameters, and growth provides the necessary perturbation to transition between the two attracting sets. We do note that over longer time periods, spatiotemporal patterns appear to form as indicated by the FFT results in Fig.~\ref{FHNlinplots}(b)iii. There structures do not persist for long, however, with the patterns forming, dissipating, and then forming again, due to the strongly oscillatory state and resulting non-monotone instability region.

\begin{figure}
\centering
\includegraphics[width=0.3\textwidth]{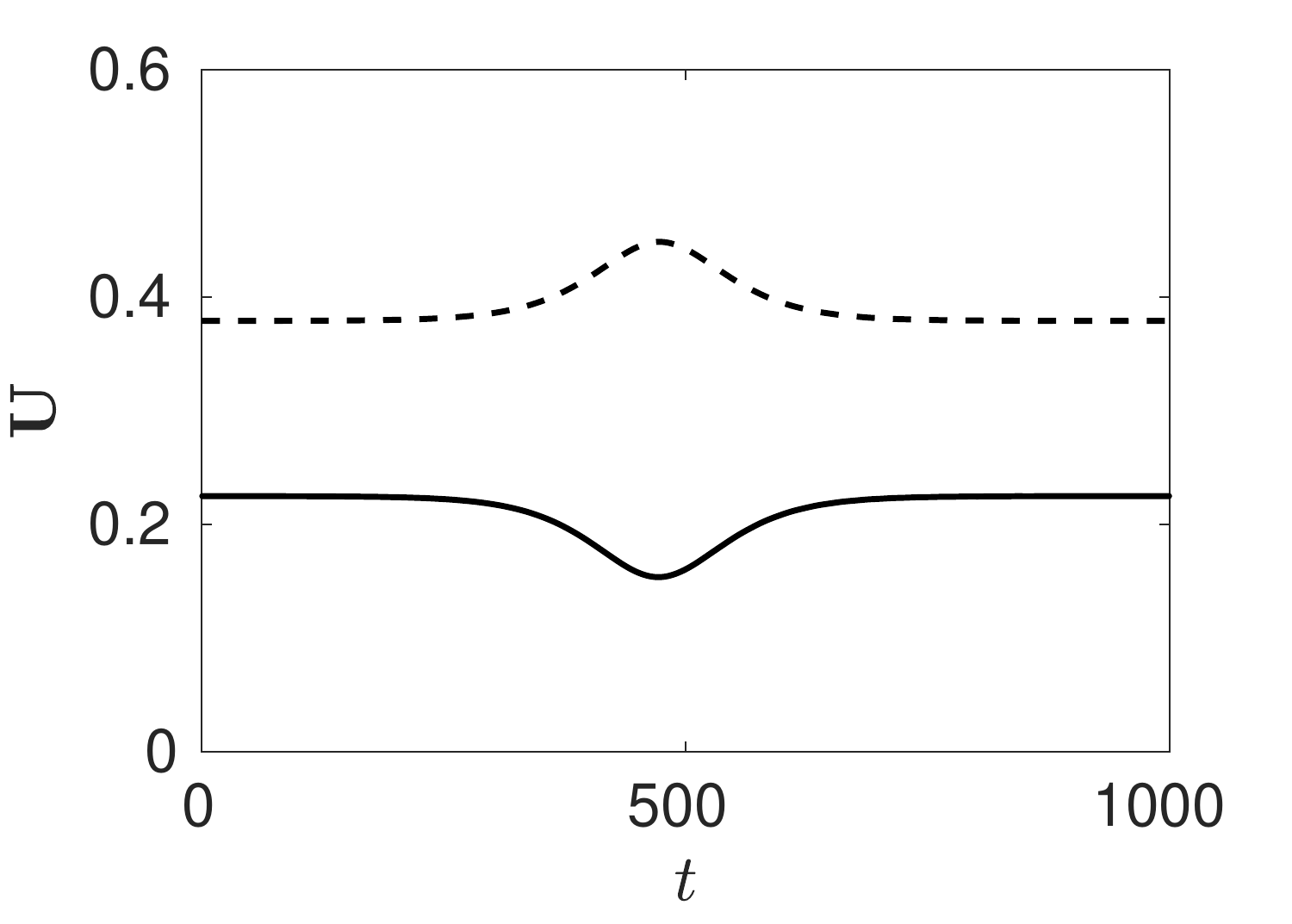}
\includegraphics[width=0.3\textwidth]{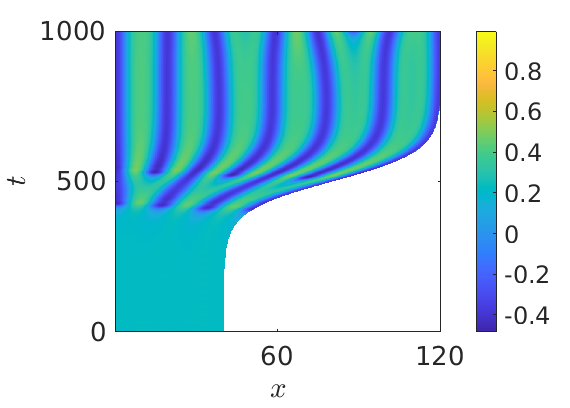}
\includegraphics[width=0.3\textwidth]{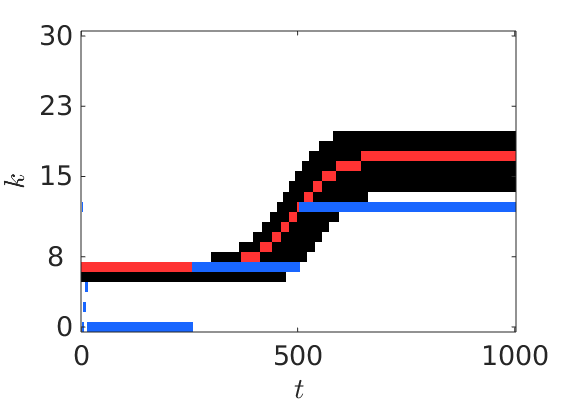}

(a)i \hspace{4cm} (a)ii \hspace{4cm} (a)iii 

\includegraphics[width=0.3\textwidth]{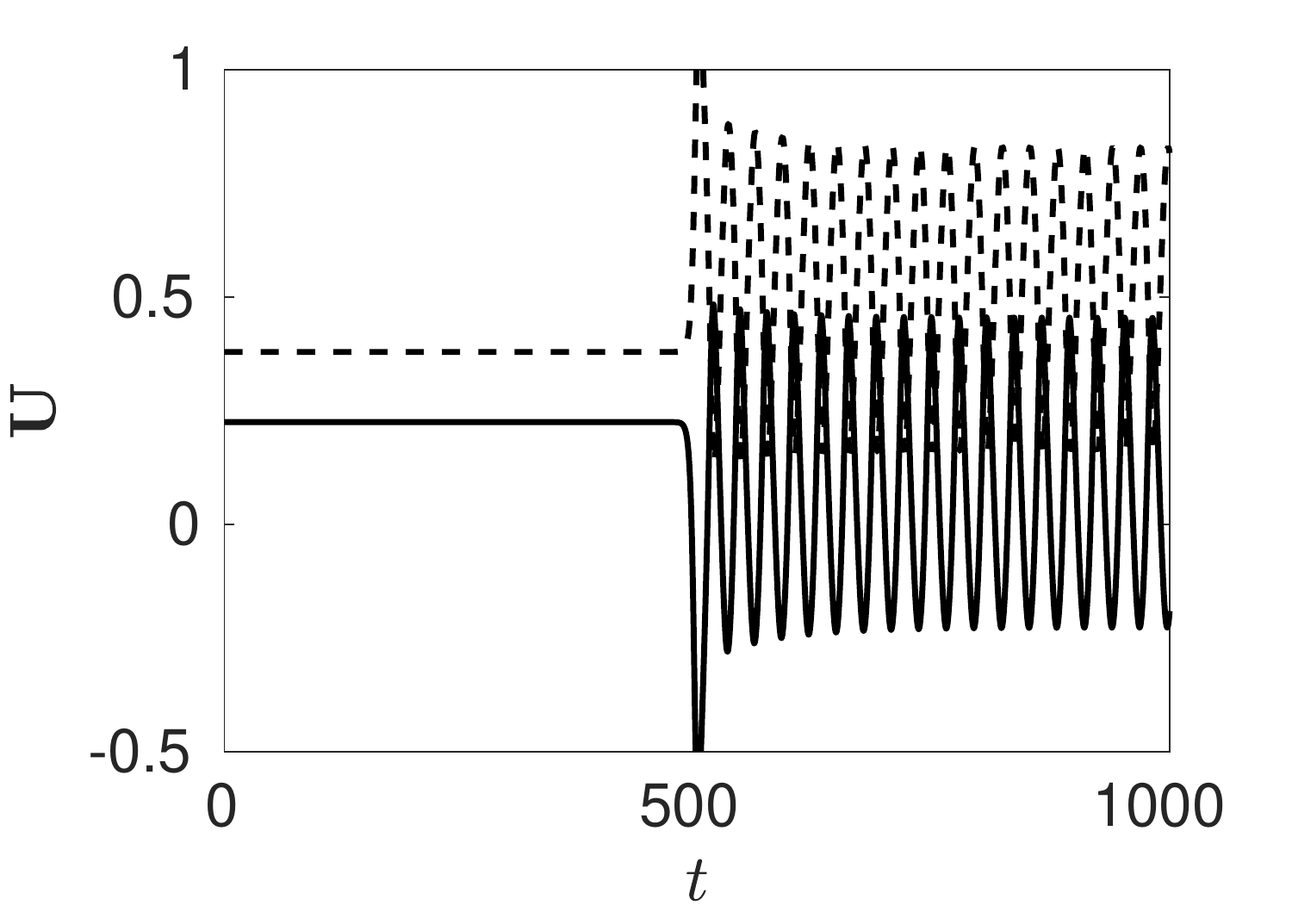}
\includegraphics[width=0.3\textwidth]{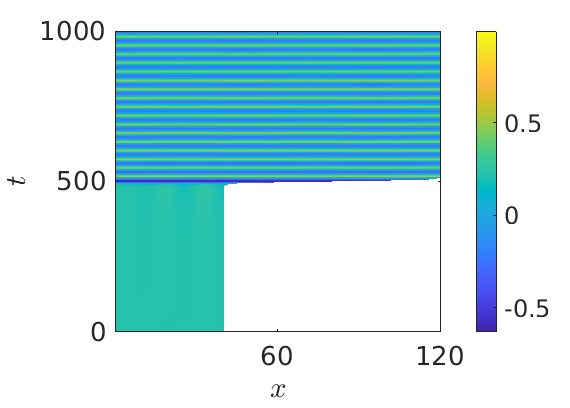}
\includegraphics[width=0.3\textwidth]{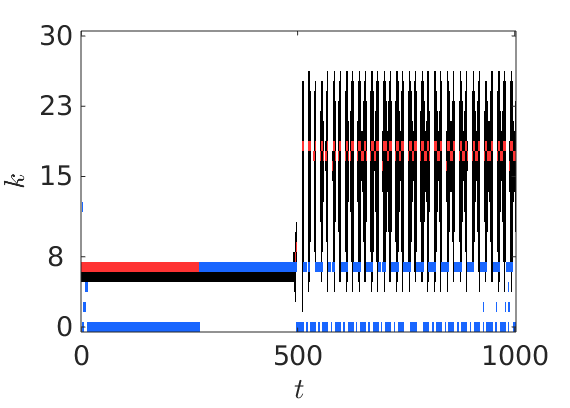}

(b)i \hspace{4cm} (b)ii \hspace{4cm} (b)iii

\vspace{-0.1in}
\caption{Plots corresponding to the kinetics \eqref{FHN} with parameters $a=0.6$, $b=0.99$, $c=1.02$, $i_0=0.6$, $d_1=1$, and $d_2=1.7$. The domain is taken to grow as $r(t) = 40(2+\tanh(s(t-t_f/2)))$ with $t_f=10^3$ for (a) $s=0.01$, (b) $s=0.2$. In column (i) we plot solutions of the homogeneous base state solution of \eqref{uniform} over time, with $U_1$ given by the dark line and $U_2$ by the dashed line. In column (ii) we show plots of $u_1$ over space and time. In column (iii) we plot the dispersion set $\mathcal{K}_t$ in black, with the theoretically maximally growing mode in red and the largest frequency component of the FFT of $u_1(x,t)$ from the full numerical solution in blue.\label{FHNtanhplots}}
\end{figure}

Similarly, in Fig.~\ref{FHNtanhplots} we observe that a short but rapid domain expansion can induce the same type of multistability. If the increase in the size of the domain is sufficiently slow, a connected dispersion set is recovered. In fact, the quasi-static approach would always generate such a continuous set, as it cannot account for the possibility of an oscillatory base state. While stepwise growth is less simple to analyze than that of linear or exponential growth, it has physiological significance in a number of organisms which exhibit pulsatile growth spurts between periods of slow or stagnant growth during development \cite{beloussov2003geometro, feijo2001cellular} which we model by a rapid smooth expansion. 

Fig.~\ref{FHNtanhplots}(b)iii again shows that an oscillatory base state can result in transient patterns rather than a single persistent patter, like what was seen in previous examples. Fig.~\ref{FHNtanhplots}(a)iii, however, shows something new and fairly interesting. The initial pattern (corresponding to a dominant mode of $k\sim 7$) destabilizes, with a new pattern selected during the short but rapid growth near $k \sim 12$. After this, this pattern is locked in, even though the instability region soon after becomes fixed between $14\leq k \leq 20$ once the growth of the domain finishes. Since the pattern was formed during this short growth period, it lies adjacent to but just outside of the instability region for $t\rightarrow \infty$. This highlights an interesting case where the Turing pattern would not have been detected in the asymptotic limit of $t\rightarrow \infty$, even though the dynamics of the perturbation become autonomous in this limit. This again highlights the strong role of hysteresis in forming patterns when domain evolution is involved.

\section{Applications to reaction-diffusion systems on isotropically growing manifolds}\label{mfds}
We next give examples of domains which evolve in more than one spatial dimension. 

\subsection{Isotropic evolution of a circular disk in $\mathbb{R}^2$}
Turing conditions for reaction-diffusion systems on a static disk were recently obtained in \cite{sarfaraz2018domain}, and taking growth and volume expansion terms to zero we recover their conditions as a special case. Numerical simulations and experimental results for a specific application of Lengyel-Epstein reaction kinetics on a growing disk were given in \cite{preska2014target}, although external forcing on the reaction-diffusion model was employed. More recent experiments in a radially expanding domain have been performed which show a crucial mode selection phenomenon induced by the speed of the growth \cite{konow2019turing}.

\begin{figure}
\centering
\includegraphics[width=0.32\textwidth]{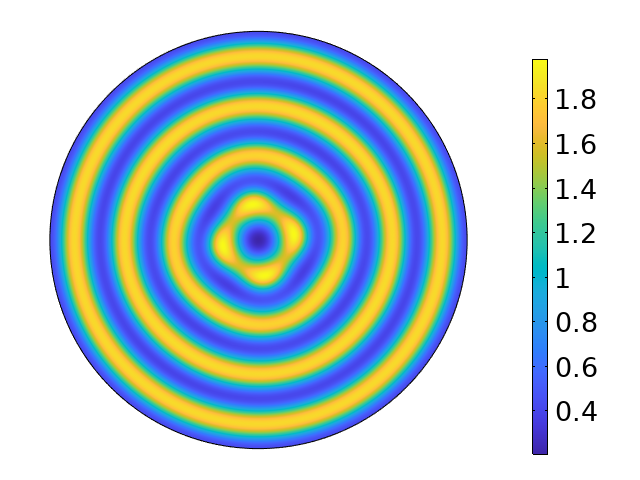}
\includegraphics[width=0.32\textwidth]{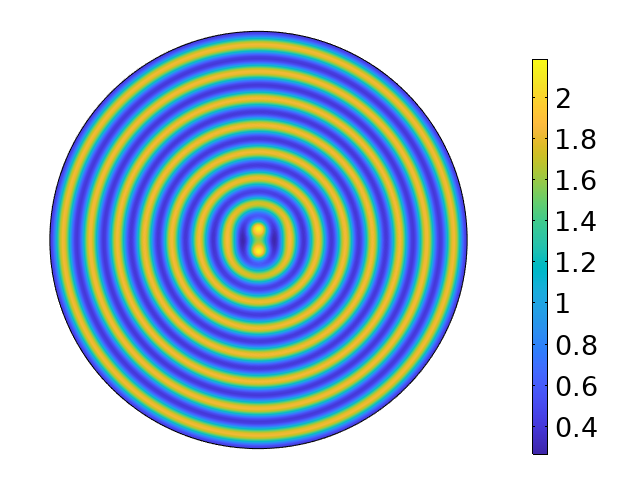}
\includegraphics[width=0.32\textwidth]{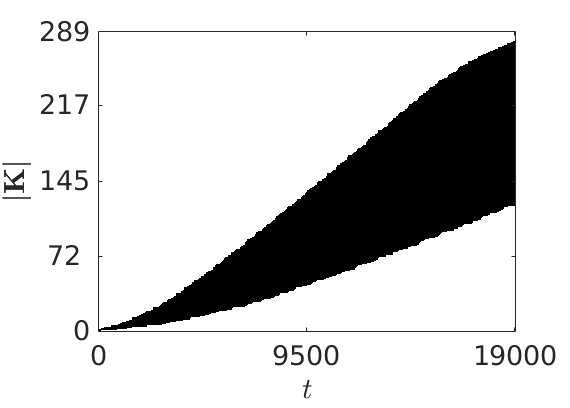}

(a)i \hspace{5cm} (a)ii \hspace{5cm} (a)iii 

\includegraphics[width=0.32\textwidth]{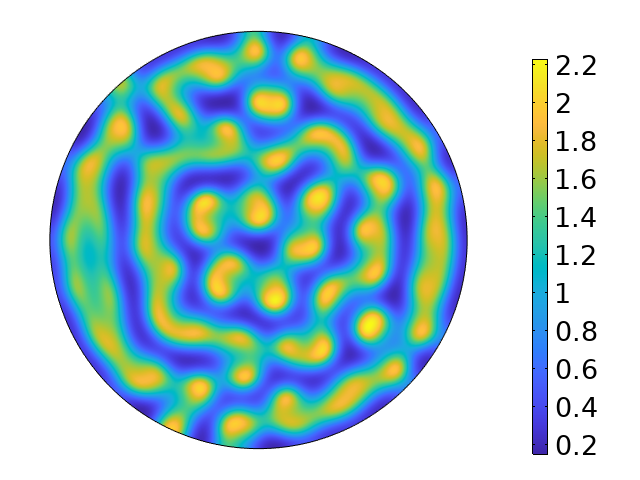}
\includegraphics[width=0.32\textwidth]{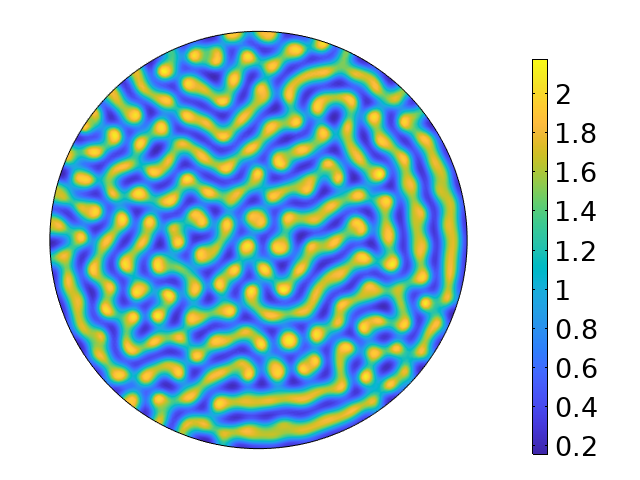}
\includegraphics[width=0.32\textwidth]{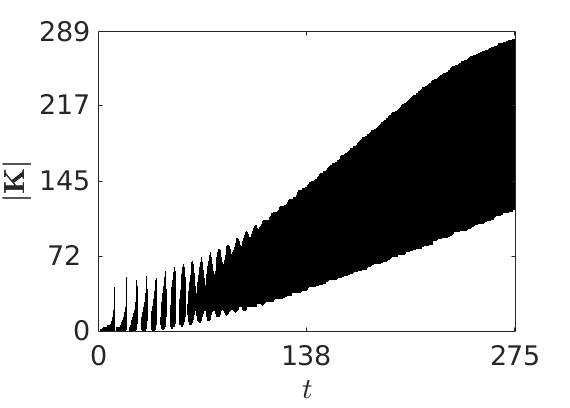}

(b)i \hspace{5cm} (b)ii \hspace{5cm} (b)iii 

\vspace{-0.1in}
\caption{Plots of $u_1$ corresponding to the kinetics \eqref{Schnack} with parameters $a=0$, $b=1.1$, $d_1=1$, and $d_2=10$ on a linearly isotropically growing disk with $r(t) = 4(1+st)$ for (a) $s=0.001$ and (b) $0.069$. We take $t=0.5t_f$ and $t=t_f$ in (i,ii), with dispersion plots $\mathcal{K}_t$ shown in (iii). For each case, we take $t_f$ so that that the domain has grown to $20$ times its initial size. For the dispersion plots, we order the $\tilde{j}_{\ell,k}$ by magnitude and plot dispersion sets in this order, where $|\mathbf{K}|$ denotes the index of this ordering.\label{circplots}}
\end{figure}

We consider $\mathbf{X}=(r(t)x_1\cos(2\pi x_2),r(t)x_1\sin(2\pi x_2))$, $x_1,x_2\in [0,1]$, for the isotropic evolution of a circular disk. Given Neumann data on the circle $|\mathbf{X}|=r(t)$, we have $\rho_{\ell,k}(t) = \frac{\tilde{j}_{\ell,k}^2}{r(t)^2}$, where $\tilde{j}_{\ell,k}$ denotes the $k$th positive root of the derivative of the Bessel function of the first kind $\mathcal{J}_\ell(x)$, $\ell = 0,1,2,\dots$, and $\mu(t) = r(t)^2$.

Under the criterion given in Theorem \ref{Thm3}, we find that the $(\ell,k)$th perturbation of the form \eqref{pert2} corresponding to $\rho_{\ell,k}(t) = \frac{\tilde{j}_{\ell,k}^2}{r(t)^2}$ is unstable over some interval $t\in \mathcal{I}_{\ell,k}$, provided that the inequality
\begin{equation}\begin{aligned}\label{exdisk}
& \det(J) - \left( d_2 J_{11} + d_1 J_{22} \right)\frac{\tilde{j}_{\ell,k}^2}{r^2} + d_1d_2\frac{\tilde{j}_{\ell,k}^4}{r^4} \\
& \qquad
 < - 2\frac{\ddot{r}}{r}  - 2\left( \frac{\dot{r}}{r}\right)^2 - 2\frac{\dot{r}}{r}\left( \left( d_1 + d_2\right)\frac{\tilde{j}_{\ell,k}^2}{r^2} - \mathrm{tr}(J) \right) \\
& \qquad\qquad + \max\left\lbrace 2\frac{\dot{r}}{r} \frac{\dot{J}_{12}}{J_{12}} - J_{12}\frac{d}{dt}\left( \frac{d_1 \tilde{j}_{\ell,k}^2}{r^2 J_{12}} -\frac{J_{11}}{J_{12}}\right) , 2\frac{\dot{r}}{r} \frac{\dot{J}_{21}}{J_{21}} - J_{21}\frac{d}{dt}\left( \frac{d_2 \tilde{j}_{\ell,k}^2}{r^2 J_{21}} - \frac{J_{22}}{J_{21}}\right)\right\rbrace 
\end{aligned}\end{equation}
holds for all $t\in \mathcal{I}_{\ell,k}$. In the special case where $r(t)=R$ for constant $R>0$, the condition \eqref{exdisk} reduces to 
\begin{equation}
\det(J) - \left( d_2 J_{11} + d_1 J_{22} \right)\frac{\tilde{j}_{\ell,k}^2}{R^2} + d_1d_2\frac{\tilde{j}_{\ell,k}^4}{R^4} < 0\,,
\end{equation}
which is exactly the classical Turing condition for a static circular disk of radius $R$.

We show simulations on a growing disk in Fig.~\ref{circplots}. We see that for slower growth rates, the successive instabilities lead to symmetric patterning, though as the growth rate is increased additional unstable modes lead to less-robust patterning, with additional irregularity in the pattern structure. This is analogous to the one-dimensional case where robustness is only attained for certain growth rates \cite{crampin2001reaction, ueda2012mathematical}. For larger growth rates, there is no Turing pattern and $u_2$ grows uniformly while $u_1$ decays to zero everywhere. We note that due to the difference in $\mu(t)$ in the setting of a two-dimensional manifold, the specific value at which this instability in the base state occurs is different from the one-dimensional setting shown in Fig.~\ref{linplots}. Regarding the dispersion plots, we select a suitably large subset of them and then sort the resulting eigenvalues $\tilde{j}_{\ell,k}$ by magnitude to obtain a one-dimensional dispersion set analogous to those shown in the previous section, though omit any Fourier analysis of the numerical simulations or maximal growth rates. As anticipated, we observe broadly similar curves for different growth rates, though the faster rate admits transient oscillations when the domain is small, as seen earlier for one-dimensional cases.

\subsection{Isotropic evolution of an equilateral triangle}
We found no studies of Turing patterns on equilaterial triangles (static or growing), yet this is an example for which our results can be easily applied. Consider an equilateral triangle defined by
$\Omega(t)=\left\lbrace (r(t)x_1,r(t)x_2) ~|~ 0 \leq x_1 \leq 1, ~ 0 \leq x_2 \leq \sqrt{3}x_1, ~  x_2 \leq \sqrt{3}(1-x_1)\right\rbrace$. Due to self-similarity of this domain, growth of a single face is equivalent to growth of all three faces, with the domain remaining an equilateral triangle for all $t\geq 0$. The spectrum for $\Omega(0)$ is given by $\rho_{k_1,k_2}(0) = \frac{16\pi^2}{9}\left( k_1^2 + k_1 k_2 + k_2^2\right)$, for $(k_1,k_2)\in \mathbb{N}^2$, with the spectrum on the evolving domain then given by $\rho_{k_1,k_2}(t) = \frac{16\pi^2}{9r(t)^2}\left( k_1^2 + k_1 k_2 + k_2^2\right)$. As the domain is flat and planar, we have $\mu(t) = r(t)^2$.

Under the criterion given in Theorem \ref{Thm3}, we find that the $(k_1,k_2)$th perturbation of the form \eqref{pert2} corresponding to $\rho_{k_1,k_2}(t)$ is unstable over some interval $t\in \mathcal{I}_{k_1,k_2}$, provided that the inequality
\begin{equation}\begin{aligned}\label{extriangle}
& \det(J) - \left( d_2 J_{11} + d_1 J_{22} \right)\frac{16\pi^2\left( k_1^2 + k_1 k_2 + k_2^2\right)}{9r^2} + \frac{256\pi^4d_1d_2\left( k_1^2 + k_1 k_2 + k_2^2\right)^2}{81r^4} \\
& \qquad
 < - 2\frac{\ddot{r}}{r} - 2\left( \frac{\dot{r}}{r}\right)^2 - 2\frac{\dot{r}}{r}\left( \frac{16\pi^2\left( d_1 + d_2\right)\left( k_1^2 + k_1 k_2 + k_2^2\right)}{9r^2} - \mathrm{tr}(J) \right) \\
& \qquad\qquad + \max\left\lbrace 2\frac{\dot{r}}{r} \frac{\dot{J}_{12}}{J_{12}} - J_{12}\frac{d}{dt}\left( \frac{16\pi^2d_1\left( k_1^2 + k_1 k_2 + k_2^2\right)}{9r^2 J_{12}} -\frac{J_{11}}{J_{12}}\right) ,\right.\\
& \qquad\qquad\qquad\qquad \left. 2\frac{\dot{r}}{r} \frac{\dot{J}_{21}}{J_{21}} - J_{21}\frac{d}{dt}\left( \frac{16\pi^2d_2\left( k_1^2 + k_1 k_2 + k_2^2\right)}{9r^2 J_{21}}  -\frac{J_{22}}{J_{21}}\right)\right\rbrace 
\end{aligned}\end{equation}
holds for all $t\in\mathcal{I}_k$. In the special case of a static equilateral triangle with constant area $A>0$, we take $r(t)=\sqrt{\frac{4\sqrt{3}}{3}A}$, and the condition \eqref{extriangle} reduces to 
\begin{equation}
\det(J) - \left( d_2 J_{11} + d_1 J_{22} \right)\frac{4\sqrt{3}\pi^2\left( k_1^2 + k_1 k_2 + k_2^2\right)}{9A} + \frac{16\pi^4d_1d_2\left( k_1^2 + k_1 k_2 + k_2^2\right)^2}{27A^2} < 0\,,
\end{equation}
which is the classical Turing condition for a static equilateral triangle of area $A$.

We consider linear isotropic growth in Fig.~\ref{triangplots}. We observe a more ordered formation of stripes in the case of slower growth, whereas spots and disordered connections appear for faster growth. We note that the quasi-static dispersion sets are identical to Fig.~\ref{triangplots}(a)iii, but the inset in Fig.~\ref{triangplots}(b)iii demonstrates the impact of a faster growth rate.

\begin{figure}
\centering
\includegraphics[width=0.32\textwidth]{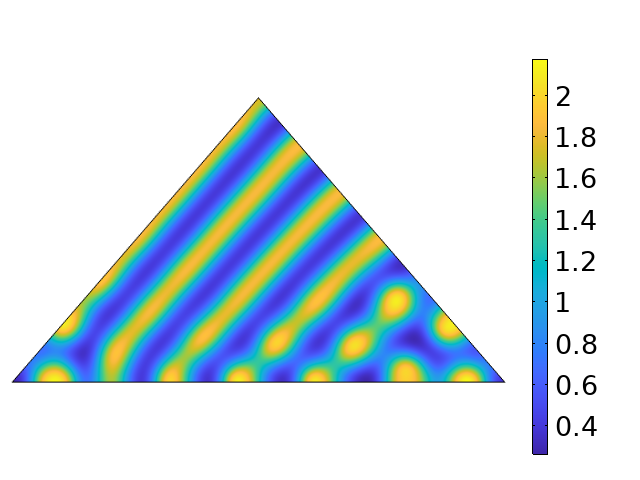}
\includegraphics[width=0.32\textwidth]{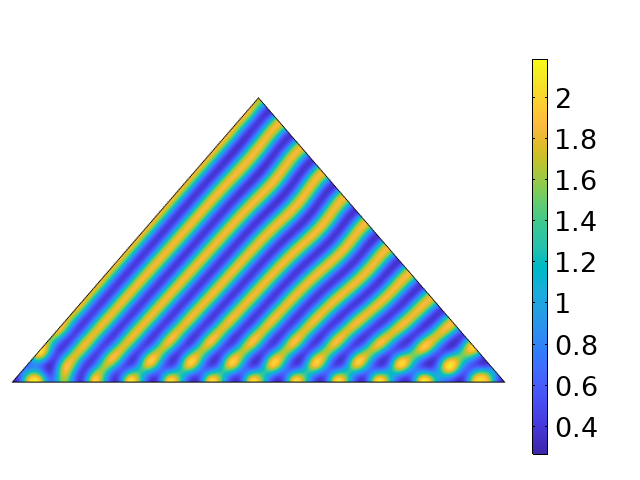}
\includegraphics[width=0.32\textwidth]{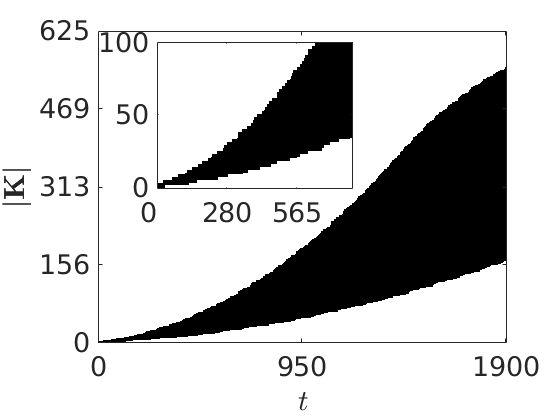}

(a)i \hspace{5cm} (a)ii \hspace{5cm} (a)iii 

\includegraphics[width=0.32\textwidth]{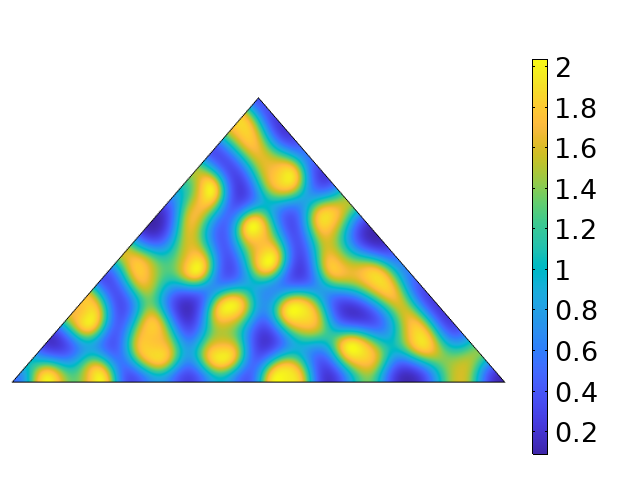}
\includegraphics[width=0.32\textwidth]{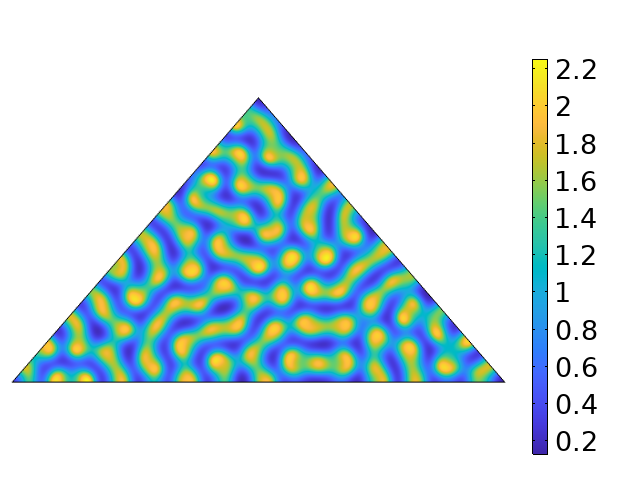}
\includegraphics[width=0.32\textwidth]{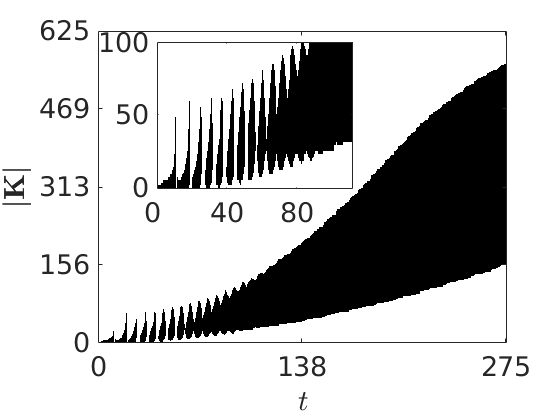}

(b)i \hspace{5cm} (b)ii \hspace{5cm} (b)iii

\vspace{-0.1in}
\caption{Plots of $u_1$ corresponding to the kinetics \eqref{Schnack} with parameters $a=0$, $b=1.1$, $d_1=1$, and $d_2=10$ on an isotropically growing equilateral triangle with $r(t) = 8(1+st)$ for (a) $s=0.01$ (b) $0.069$. In all simulations we take the final time such that the domain has grown to $20$ times its initial size. Figures are shown at times (i) $t=0.5t_f$ and (ii) $t=t_f$.  We use the notation $|\mathbf{K}|$ to denote a sequential numbering of these ordered states by magnitude.\label{triangplots} }
\end{figure}

\begin{figure}
\centering

\includegraphics[width=0.4\textwidth]{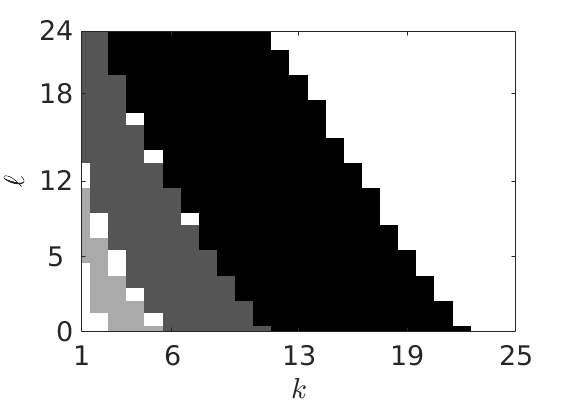}
\includegraphics[width=0.36\textwidth]{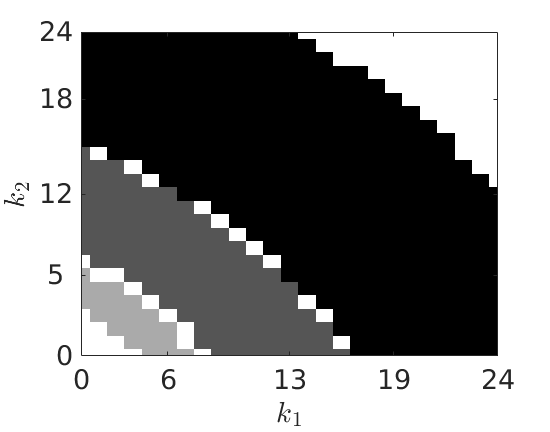}

\hspace{0.5cm} (a) \hspace{5.5cm} (b) 

\vspace{-0.1in}
\caption{Plots of unstable modes over both indices, corresponding to dispersion curves for (a) a disk (similar to the dynamics of Fig.~\ref{circplots}) and (b) an equilateral triangle (similar to the dynamics of Fig.~\ref{triangplots}), both with linear growth rate $s=0.01$. The different shaded regions are unstable modes for times $t=300, 850, 1900$, with lighter grays corresponding to earlier times. These reslts correspond to the dispersion sets $\mathcal{K}_{300}$, $\mathcal{K}_{850}$, and $\mathcal{K}_{1900}$, respectively.\label{Kscircplots}}
\end{figure}

We also demonstrate the time dependence of sets $\mathcal{K}_t$ for more than one index in Fig.~\ref{Kscircplots}, where we compare results for the disk and triangle. We observe an increasing band of unstable wavenumbers emanating from the origin. Although dispersion sets $\mathcal{K}_t$ agree between both domains in a qualitative sense, we see that the unstable modes in the triangular case are not bounded by lines as in the case of the disk, but instead by circular arcs, due to the form of the Laplace-Beltrami eigenvalues for each respective domain.

\subsection{Isotropic evolution of a sphere $S^N \subset \mathbb{R}^{N+1}$}
There have been various pattern formation studies on static 2-spheres \cite{chaplain2001spatio, krause2018emergent, rozada2014stability, varea1999turing}. An exponentially growing 2-sphere was considered by \cite{gjorgjieva2007turing}, although their analysis was quasi-static, thereby ignoring the role of transients in the dynamics of \eqref{uniform}. Similar assumptions were made in \cite{toole2013turing}. Numerical simulations of pattern formation on growing 2-spheres, as well as anisotropic growth of 2-spheres into ellipsoids, were obtained in \cite{krause2018influence}.

\begin{figure}
\centering
\includegraphics[width=0.32\textwidth]{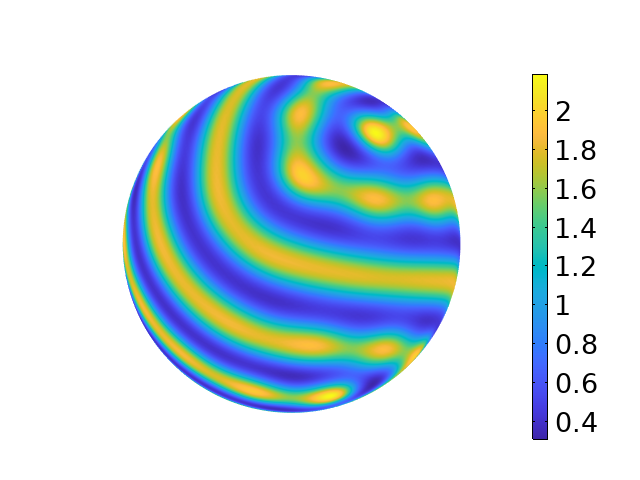}
\includegraphics[width=0.32\textwidth]{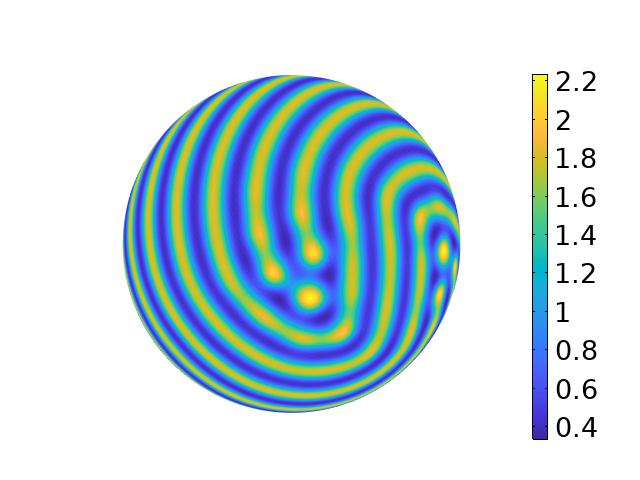}
\includegraphics[width=0.32\textwidth]{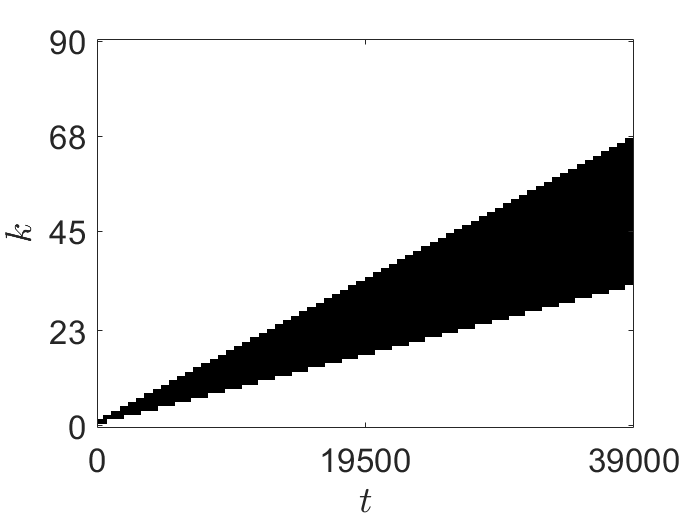}

(a)i \hspace{5cm} (a)ii \hspace{5cm} (a)iii 

\includegraphics[width=0.32\textwidth]{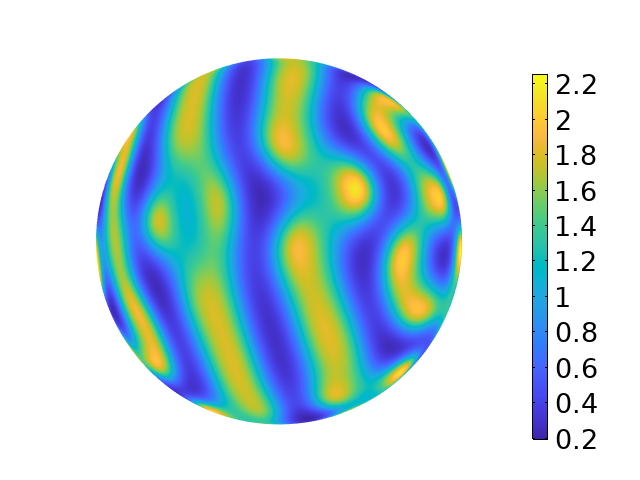}
\includegraphics[width=0.32\textwidth]{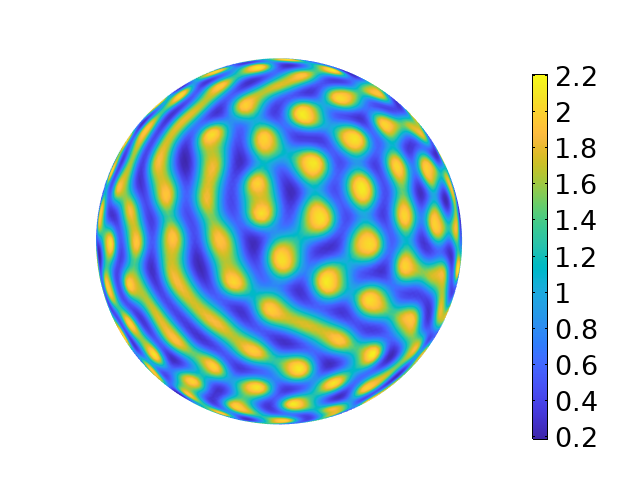}
\includegraphics[width=0.32\textwidth]{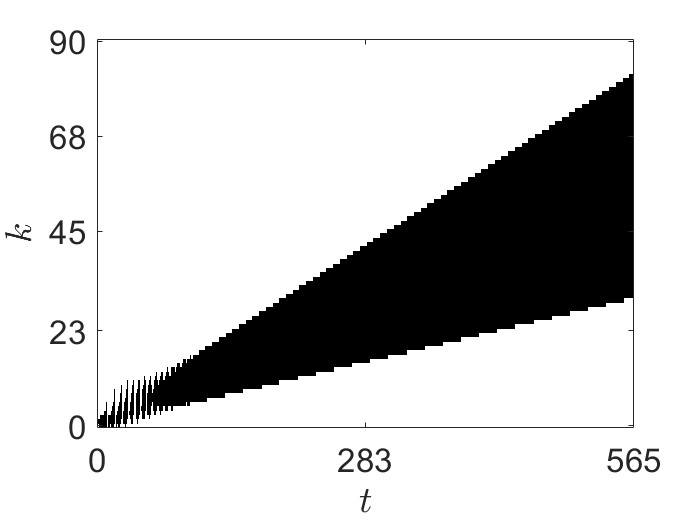}

(b)i \hspace{5cm} (b)ii \hspace{5cm} (b)iii 

\vspace{-0.1in}
\caption{Plots of $u_1$ corresponding to the kinetics \eqref{Schnack} with parameters $a=0$, $b=1.1$, $d_1=1$, and $d_2=10$ on an isotropically growing 2-sphere at different points in time. The domain is taken to grow with $r(t) = 2(1+st)$ for (a) $s=0.001$, (b) $0.069$. In all simulations we take the final time such that the domain has grown to $r(t_f)=80$; figures are shown at (i) $t = 0.5t_f$, $t=t_f$. Respective dispersion sets $\mathcal{K}_t$ are shown in (iii).\label{sphereplots}}
\end{figure}

The unit $N$-sphere ${S}^{N}\subset \mathbb{R}^{N+1}$ has spectrum $\rho_k(0)= k(k+N-1)$, $k=0,1,2,\dots$, hence $\rho_k(t) = \frac{k(k+N-1)}{r(t)^2}$ and $\mu(t) = r(t)^N$. These eigenvalues will have increasingly large multiplicity corresponding to different eigenfunctions on ${S}^{N}$. That is to say, for a fixed $k$, there can exist multiple distinct $\psi_k(\mathbf{x})$ in the general perturbation \eqref{pert2}. As such, each of these distinct spatial eigenfunctions can yield patterning when the perturbation corresponding to $\rho_k(t)$ is unstable.

Under the criterion given in Theorem \ref{Thm3}, we find that the $k$th perturbation of the form \eqref{pert2} corresponding to $\rho_k(t) = \frac{k(k+N-1)}{r(t)^2}$ is unstable over some interval $t\in \mathcal{I}_k$, provided that the inequality
\begin{equation}\begin{aligned}\label{exsphere}
& \det(J) - \left( d_2 J_{11} + d_1 J_{22} \right)\frac{k(k+N-1)}{r^2} + \frac{d_1d_2k^2(k+N-1)^2}{r^4} \\
& \qquad
 <  -N\frac{\ddot{r}}{r} - N(N-1)\left( \frac{\dot{r}}{r}\right)^2 - N \frac{\dot{r}}{r}\left( \frac{\left( d_1 + d_2\right)k(k+N-1)}{r^2} - \mathrm{tr}(J) \right) \\
& \qquad + \max\left\lbrace N \frac{\dot{r}}{r} \frac{\dot{J}_{12}}{J_{12}} - J_{12}\frac{d}{dt}\left( \frac{d_1k(k+N-1)}{r^2 J_{12}}  -\frac{J_{11}}{J_{12}}\right) , N \frac{\dot{r}}{r} \frac{\dot{J}_{21}}{J_{21}} - J_{21}\frac{d}{dt}\left( \frac{d_2k(k+N-1)}{r^2 J_{21}}  -\frac{J_{22}}{J_{21}}\right)\right\rbrace 
\end{aligned}\end{equation}
holds for all $t\in\mathcal{I}_k$. In the special case where $r(t)=R$ for constant $R>0$, hence the domain is a static $N$-sphere of radius $R$, the condition \eqref{exsphere} reduces to 
\begin{equation}
\det(J) - \left( d_2 J_{11} + d_1 J_{22} \right)\frac{k(k+N-1)}{R^2} + \frac{d_1d_2k^2(k+N-1)^2}{R^4} < 0\,,
\end{equation}
which is exactly the classical Turing condition for a static $N$-sphere.

We first obtain solutions on $S^2$ in Fig.~\ref{sphereplots} for different rates of linear growth. The unstable modes are qualitatively the same as in Fig.~\ref{linplots}, up to a rescaling due to domain size. This suggests that the difference in volume expansion in this case does not have a substantial effect on these dispersion sets. In order to better understand the role of volume expansion, in Fig.~\ref{KplotsNspherelinear} we compare dispersion sets for spheres of different dimension undergoing linear growth. The spectrum for each is similar, yet due to differences in the volume expansion term, we find that the dispersion sets collapse for large enough $N$, since for such cases volume expansion is far more rapid, resulting in more rapid dilution of the spatially homogeneous state which prevents spatial patterning. Hence, there are indeed differences in patterning due to volume expansion, yet depending on the growth function selected, these differences may manifest only for large $N$.

\begin{figure}
\centering
\includegraphics[width=0.32\textwidth]{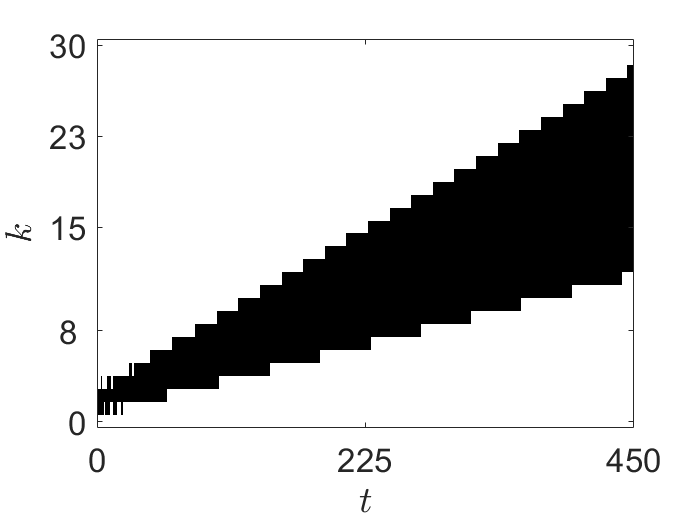}
\includegraphics[width=0.32\textwidth]{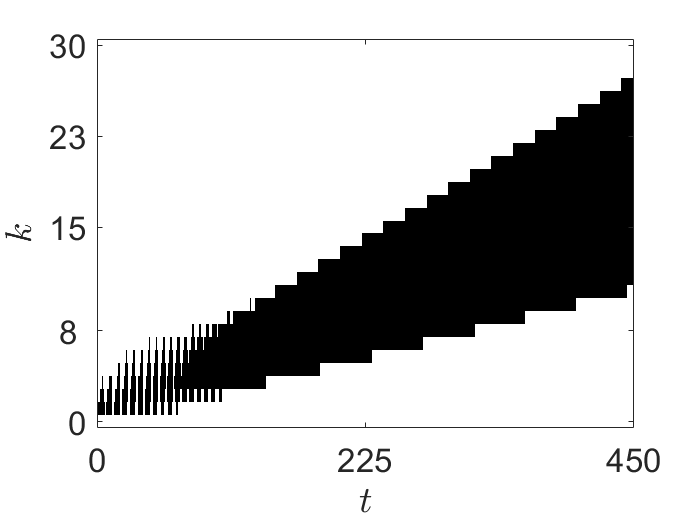}
\includegraphics[width=0.32\textwidth]{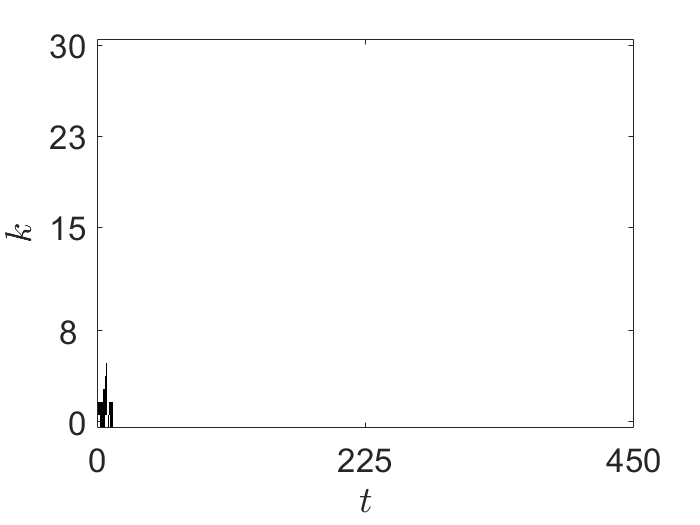}

\hspace{0.1cm} (a) \hspace{4.6cm} (b) \hspace{4.6cm} (c) 

\vspace{-0.1in}
\caption{Dispersion sets $\mathcal{K}_t$ for dynamics corresponding to the kinetics in Fig.~\ref{sphereplots} on isotropically growing $N$-spheres $S^N$ with (a) $N=2$, (b) $N=4$, (c) $N=6$. We consider an isotropic linear growth $r(t) = 3(1+st)$ with rate parameter $s=0.02$, and allow growth until $t_f = 450$, as which point the final radius is $r(t_f)=30$. Although the radius change is the same, the rate of volume expansion is considerably faster as the dimension increases.\label{KplotsNspherelinear}}
\end{figure}

\section{Domain evolution with area or volume conservation}\label{sec44}
Much literature on evolving domains considers strict growth or expansion of the domain ($\dot{\mu}(t)>0$). However, there are a variety of situations for which area or volume should be preserved while the underlying domain evolves such as in the buckling of intestinal crypts during development \cite{drasdo2001individual}, and cell shape changes \cite{hunding1985morphogen, thery2006cell} or stationary shapes of vesicles \cite{seifert1991shape}. For such a case, $\dot{\mu}(t) \equiv 0$, and \eqref{uniform} admits a constant exact solution $\mathbf{U}(t) \equiv \mathbf{U}^*$, akin to what is considered in the Turing conditions for static domains. For this case, $J$ is a constant matrix, and the condition from Theorem \ref{Thm3} reduces to 
\begin{equation}\label{reduced}
\det(J) - \left( d_1 J_{22} + d_2 J_{11} \right) \rho_k  + d_1 d_2 \rho_k^2 < \max\left\lbrace -d_1 \dot{\rho}_k , -d_2 \dot{\rho}_k\right\rbrace.
\end{equation}
This inequality is close to the static domain Turing condition, modified to account for the time dependence of the spectral parameter $\rho_k(t)$. We are unaware of any studies on volume-conserving domain evolution presently considered in the literature, and so give two examples.

\subsection{Area conserving evolution of a rectangular domain}\label{sec_rect}
Regarding asymmetric growth of a rectangular domain, Turing patterning when growth in only one direction with the other direction held fixed was considered in \cite{miguez2006effect}. Consider the evolution of a domain according to the coordinates $\mathbf{X} = (r_1(t)x_1,r_2(t)x_2)$, $x_1,x_2\in [0,1]$, where $r_1(t)r_2(t) =A$, and consider $r_1(0) = A_1$ and $r_2(0) = A_2$ (with $A_1A_2 =A$), along with $(r_1(t_f),r_2(t_f)) = (A_1R^{-1},A_2 R)$ for some constant $R>0$. The domain is then described by $\Omega(t) = [0,r_1(t)]\times [0,r_2(t)]$ with $\Omega(0) = [0,A_1]\times[0,A_2]$, $\Omega(t_f) = [0,A_1 R^{-1}]\times [0,A_2 R]$, and $|\Omega(t)|=A$ for all $t\geq 0$, so this manner of growth does indeed preserve area. The spectrum is $\rho_{k_1,k_2}(t) = \frac{\pi^2k_1^2}{r_1(t)^2}+\frac{\pi^2 k_2^2}{r_2(t)^2}$, $k_1,k_2\in \mathbb{N}$, which is an example were the spectrum is not necessarily monotone in time. Without loss of generality, we choose $r_1(t)=r(t)$ and $r_2(t)=\frac{A}{r(t)}$, where $0<r(t)<\infty$, and we then write $\rho_{k_1,k_2}(t) = \frac{\pi^2k_1^2}{r(t)^2}+\frac{\pi^2 k_2^2}{A^2}r(t)^2$ for $k_1,k_2\in \mathbb{N}$.

Under the criterion given by \eqref{reduced}, we find that the $(k_1,k_2)$th perturbation of the form \eqref{pert2} corresponding to $\rho_{k_1,k_2}(t) = \frac{\pi^2k_1^2}{r(t)^2}+\frac{\pi^2 k_2^2}{A^2}r(t)^2$ is unstable over some interval $t\in \mathcal{I}_{k_1,k_2}$, provided that the inequality
\begin{equation}\begin{aligned}\label{exarea}
& \det(J) - \pi^2\left( d_1 J_{22} + d_2 J_{11} \right) \left(\frac{k_1^2}{r^2}+\frac{k_2^2}{A^2}r^2\right) + \pi^4 d_1 d_2\left(\frac{k_1^2}{r^2}+\frac{k_2^2}{A^2}r^2\right)^2 \\
& \qquad < 2\pi^2\max\left\lbrace -d_1 \left( \frac{k_2^2}{A^2}r - \frac{k_1^2}{r^3}\right)\dot{r} , -d_2 \left( \frac{k_2^2}{A^2}r - \frac{k_1^2}{r^3}\right)\dot{r}\right\rbrace
\end{aligned}\end{equation}
holds for all $t\in \mathcal{I}_{k_1,k_2}$. In the special case of a static rectangle $[0,L_1]\times [0,L_2]$, the condition \eqref{exarea} reduces to 
\begin{equation}
\det(J) - \pi^2\left( d_1 J_{22} + d_2 J_{11} \right) \left(\frac{k_1^2}{L_1^2}+\frac{k_2^2}{L_2^2}\right) + \pi^4 d_1 d_2\left(\frac{k_1^2}{L_1^2}+\frac{k_2^2}{L_2^2}\right)^2 < 0\,,
\end{equation}
which is the classical Turing condition for a static rectangular domain.

\begin{figure}
\centering
\includegraphics[width=0.4\textwidth]{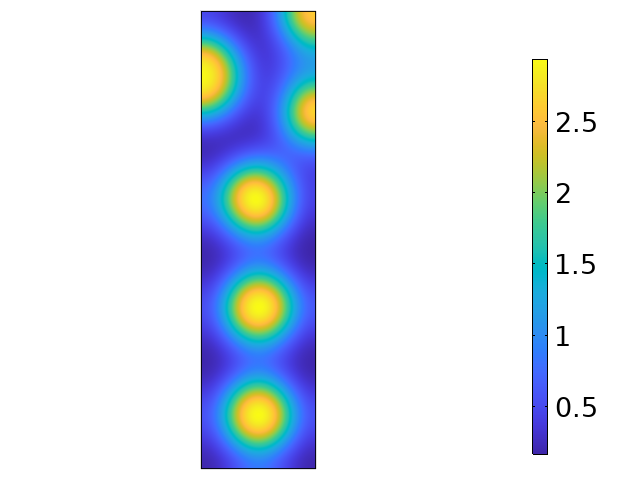}
\includegraphics[width=0.4\textwidth]{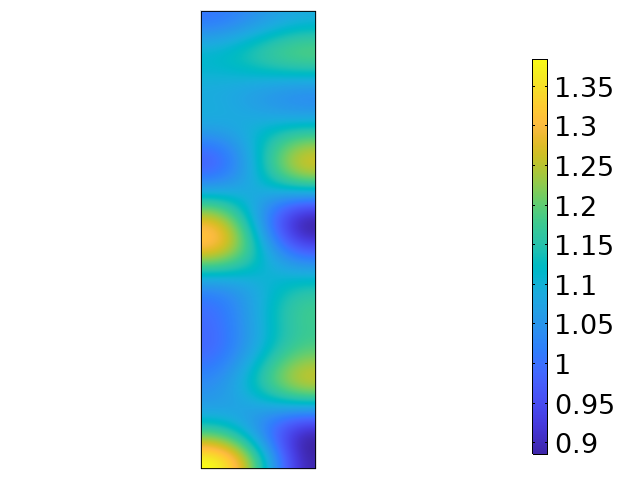}

(a)i \hspace{6cm} (b)i

\includegraphics[width=0.4\textwidth]{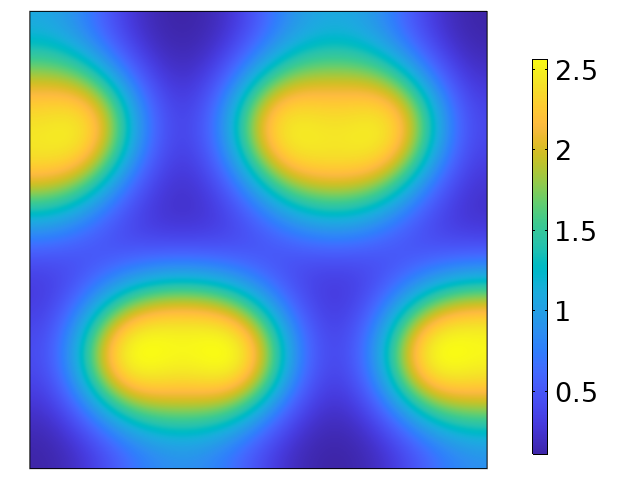}
\includegraphics[width=0.4\textwidth]{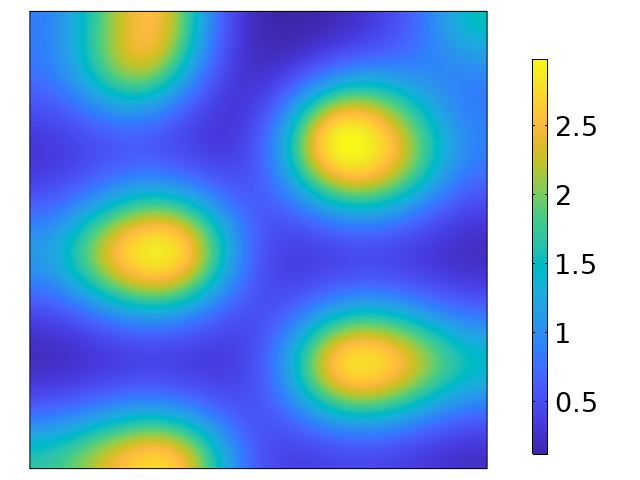}

(a)ii \hspace{6cm} (b)ii

\includegraphics[width=0.4\textwidth]{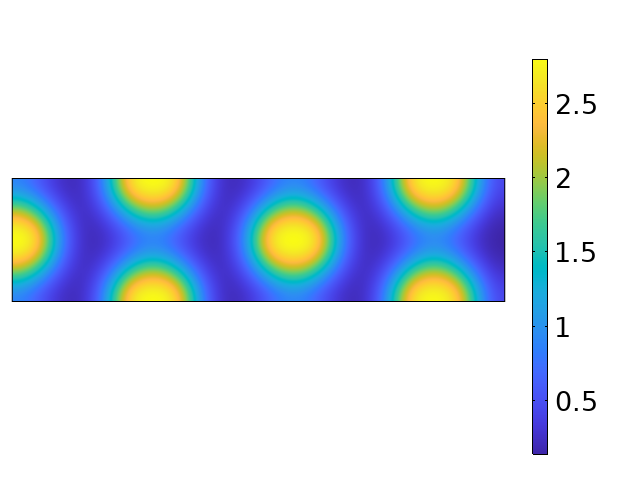}
\includegraphics[width=0.4\textwidth]{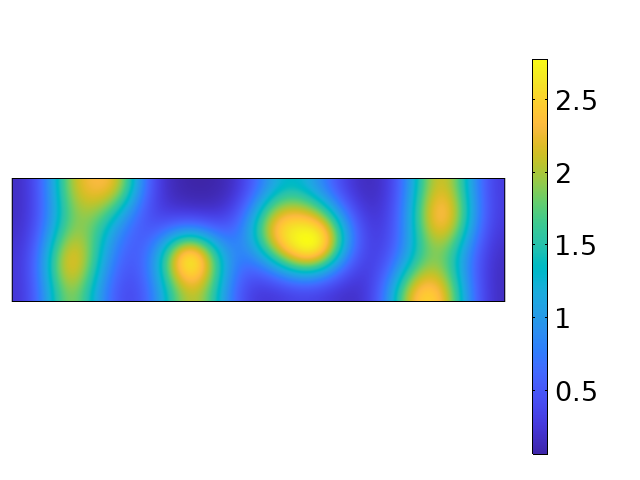}

(a)iii \hspace{6cm} (b)iii

\includegraphics[width=0.4\textwidth]{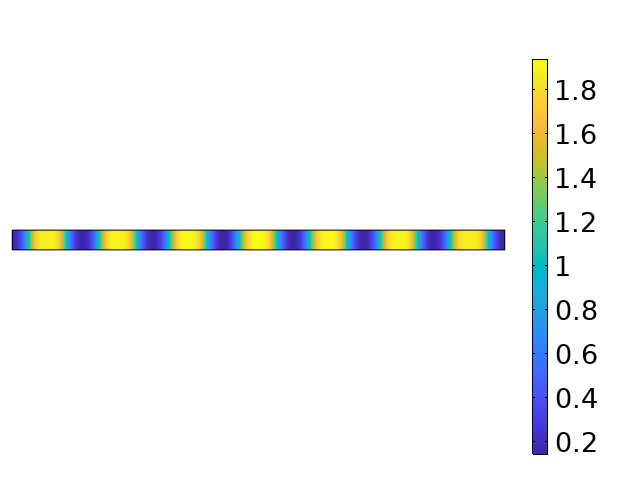}
\includegraphics[width=0.4\textwidth]{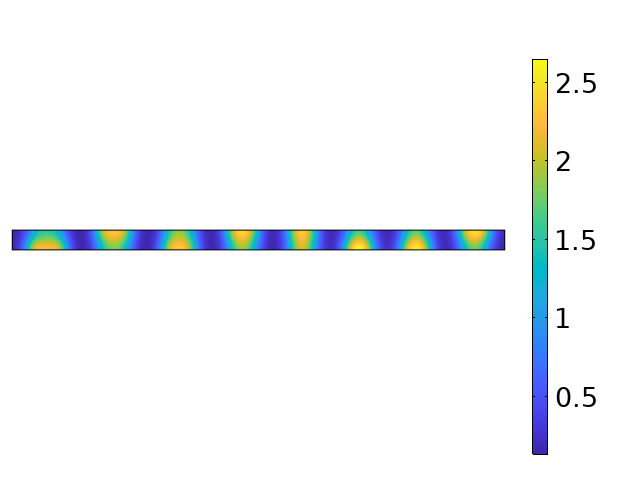}

(a)iv \hspace{6cm} (b)iv

\vspace{-0.1in}
\caption{Plots of $u_1$ corresponding to the kinetics \eqref{Schnack} with parameters $a=0$, $b=1.1$, $d_1=1$, and $d_2=15$ on a rectangular domain with $r_1(t) = 4(1+st)$ and $r_2(t) = 100/(1+st)$ for (a) $s=0.001$, (b) $s=0.1$. The final time $t_f$ is selected so that $r_1(t_f)=100$, $r_2(t_f)=4$. We give plots at $t=0.0625t_f, 0.1667t_f, 0.375t_f, t_f$ in (i)-(iv), respectively.\label{rectangplots} }
\end{figure}

\begin{figure}
\centering
\includegraphics[width=0.45\textwidth]{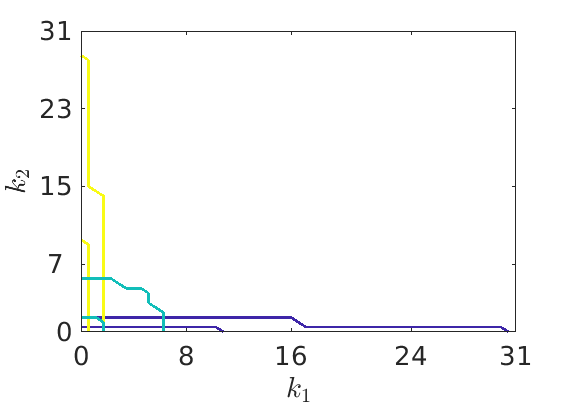}
\includegraphics[width=0.45\textwidth]{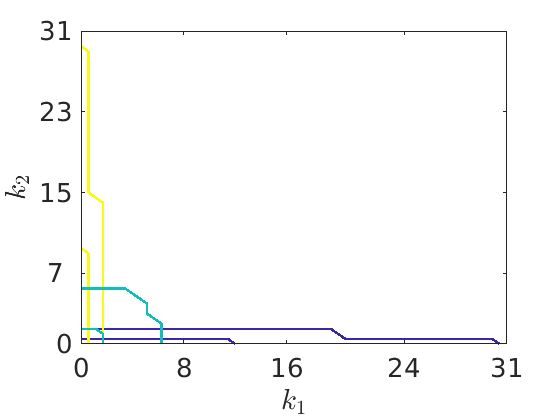}

\hspace{0.2cm} (c)i \hspace{6.6cm} (c)ii

\vspace{-0.1in}
\caption{Plots of dispersion sets $\mathcal{K}_t$ at fixed times, with (a) corresponding to the dynamics in Fig. \ref{rectangplots}(a), and (b) corresponding to the dynamics in Fig. \ref{rectangplots} (b). In particular, we plot $\mathcal{K}_0$ (corresponding to integer pairs $(k_1,k_2)$ bounded by yellow curves), $\mathcal{K}_{0.1667t_{f}}$ (bounded by teal curves), and $\mathcal{K}_{t_f}$ (bounded by purple curves).\label{rectangplotsStab}}
\end{figure}

We consider the evolution of such a domain in Fig. \ref{rectangplots}, showing both slow and fast linear evolution of the domain. For both cases, we find that the thinner rectangular domains admit solutions with more spots than when the rectangle passes through a transient square configuration. The final configuration in both cases similarly admits modes only in the $x$ direction, as one might expect from a quasi-static analysis of thin domains. The faster evolution leads to more disordered structures, as it is further from a true quasi-static picture, and the nonlinear reaction kinetics are unable to stabilize ordered spatial patterns in these transient time periods. We also plot dispersion sets $\mathcal{K}_t$ which characterize the unstable modes. The collection of unstable modes seen in Fig. \ref{rectangplotsStab} corresponds well with the patterns present at each respective time in Fig.\ref{rectangplots}. While the unstable modes are quite similar in each case (with only modest differences present on transient timescales - see $\mathcal{K}_{0.1667t_f}$), the combination of growth with nonlinear selection of patterns does play a role in which unstable modes are selected in the final pattern.

\subsection{Volume conserving evolution of a solid cylinder}\label{sec_cyl}
Three-dimensional Turing patterns have been explored in many systems and geometries, though a complete categorization of such structures and criteria for when they emerge does not yet exist as far as we are aware \cite{callahan1999pattern, de1997twist, leppanen2002new, shoji2007turing}. This is in contrast to the theory in two-dimensions, for which a reasonable classification of patterns exists, at least on rectangular domains \cite{ermentrout1991stripes, shoji2003stripes}. Such emergent structures have even been observed to be suitable for a variety of applications, such as the design of water filters \cite{tan2018polyamide}. Pattern formation on growing cylindrical domains is of strong relevance to plant growth \cite{meinhardt1998models}. Simulations and experimental observations of various Turing patterns in static cylindrical domains were shown in \cite{bansagi2011tomography}, and it was shown that three-dimensional Turing patterns exhibit an extraordinarily richer set of patterns than in one or two dimensions. Finally we remark that \cite{hunding1985morphogen} compared linear analysis and simulations on quasi-static cylinders and spheres to argue that the flattening of cells can have an impact on mitosis, and specifically during cytokinesis when cell shape changes regularly occur. 
 
Consider coordinates $\mathbf{X} = (r_1(t) x_1 \cos (2\pi x_2),r_1(t) x_1 \sin (2\pi x_2),r_2(t) x_3)$, where $x_1,x_2,x_3 \in [0,1]$, which defines a cylindrical domain $\Omega(t)$. We consider $r_1(t)^2r_2(t)=V$ so that the volume of the cylinder is conserved. Choose $r_1(0) = V_1$, $r_2(0) = V_2$ such that $V_1^2 V_2 = V$, and $(r_1(t_f),r_2(t_f)) = \left(V_1 R^{-1}, V_2 R^2\right)$ for some constant $R>0$. Then $\Omega(t) = D(r_1(t))\times [0,r_2(t)]$, where $D(r_1(t))$ denotes a disk of radius $r_1(t)$ centered at the origin, with $\Omega(0) = D(V_1)\times [0,V_2]$ and $\Omega(t_f) = D(V_1 R^{-1})\times \left[0,V_2 R^2\right]$. We have that $|\Omega(t)|=\pi r_1(t)^2r_2(t) = \pi V_1^2 V_2 = \pi V$ for all $t\geq 0$, hence volume is conserved. The spectrum of the Laplace-Beltrami operator over this domain will take the form $\rho_{\ell, k, m}(t) = \frac{\tilde{j}_{\ell ,k}^2}{r_1(t)^2}+ \frac{\pi^2 m^2}{r_2(t)^2}$, for $\ell, k,m\in \mathbb{N}$. Recalling $r_1(t)^2r_2(t)=V$, we write $r_1(t)=r(t)$ and $r_2(t)=\frac{V}{r(t)^2}$ for some function $r(t)$ satisfying $0<r(t)<\infty$. The time-dependent spectrum then becomes $\rho_{\ell, k, m}(t) = \frac{\tilde{j}_{\ell ,k}^2}{r(t)^2}+ \frac{\pi^2 m^2 r(t)^4}{V^2}$.

Under the criterion given in \eqref{reduced}, we find that the $(\ell,k,m)$th perturbation of the form \eqref{pert2} corresponding to $\rho_{\ell, k, m}(t)$ is unstable over some interval $t\in \mathcal{I}_{\ell,k,m}$, provided that the inequality
\begin{equation}\begin{aligned}\label{exvolume}
& \det(J) - \left( d_1 J_{22} + d_2 J_{11} \right) \left(\frac{\tilde{j}_{\ell ,k}^2}{r(t)^2}+ \frac{\pi^2 m^2 r(t)^4}{V^2}\right)  + d_1 d_2 \left(\frac{\tilde{j}_{\ell ,k}^2}{r(t)^2}+ \frac{\pi^2 m^2 r(t)^4}{V^2}\right)^2\\
& \qquad < 2\max\left\lbrace d_1\left(\frac{\tilde{j}_{\ell ,k}^2}{r(t)^3} - \frac{2\pi^2 m^2 r(t)^3}{V^2}\right) , d_2 \left(\frac{\tilde{j}_{\ell ,k}^2}{r(t)^3} - \frac{2\pi^2 m^2 r(t)^3}{V^2}\right)\right\rbrace
\end{aligned}\end{equation}
holds for all $t\in \mathcal{I}_{\ell,k,m}$. In the special case of a static cylinder of radius $R$ and height $H$, the Turing condition \eqref{exvolume} reduces to 
\begin{equation}
\det(J) - \left( d_2 J_{11} + d_1 J_{22} \right)\left(\frac{\tilde{j}_{\ell ,k}^2}{R^2}+ \frac{\pi^2 m^2}{H^2}\right)+ d_1d_2\left(\frac{\tilde{j}_{\ell ,k}^2}{R^2}+ \frac{\pi^2 m^2}{H^2}\right)^2 < 0\,.
\end{equation}

\begin{figure}
\centering
\includegraphics[width=0.4\textwidth]{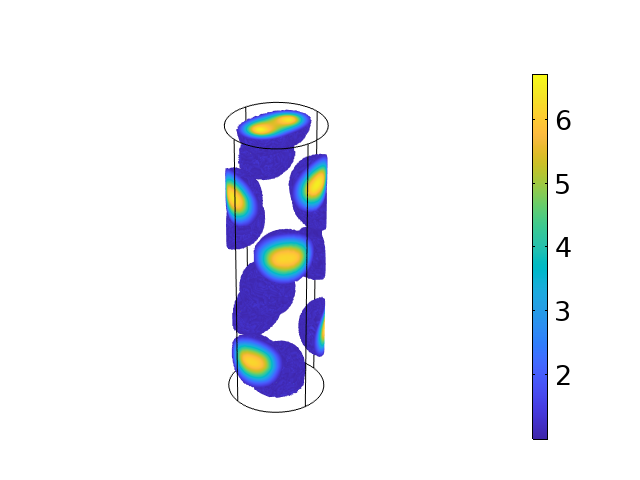}
\includegraphics[width=0.4\textwidth]{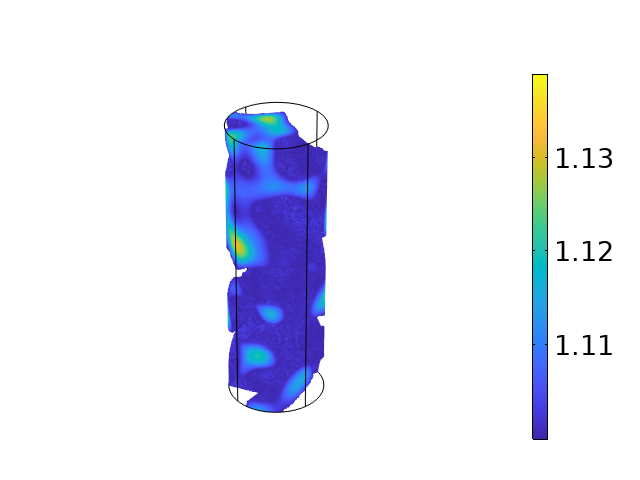}

(a)i \hspace{6cm} (b)i

\includegraphics[width=0.4\textwidth]{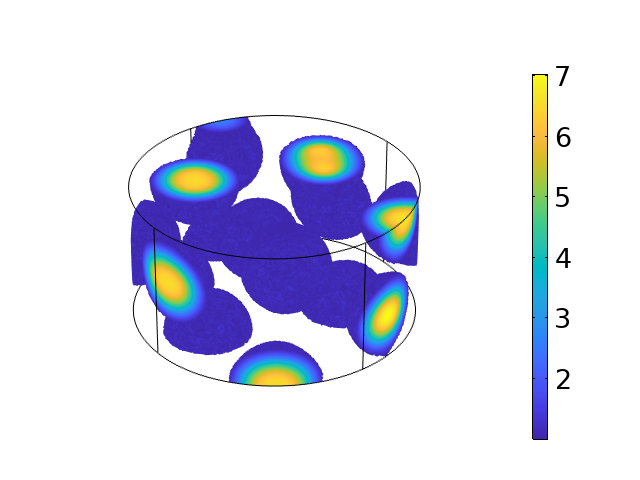}
\includegraphics[width=0.4\textwidth]{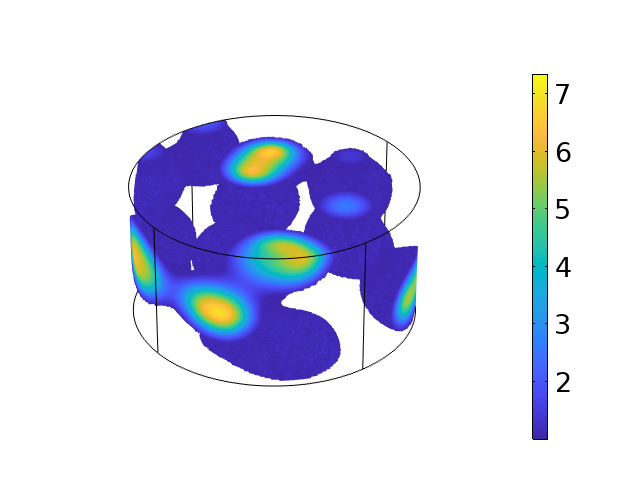}

(a)ii \hspace{6cm} (b)ii

\includegraphics[width=0.4\textwidth]{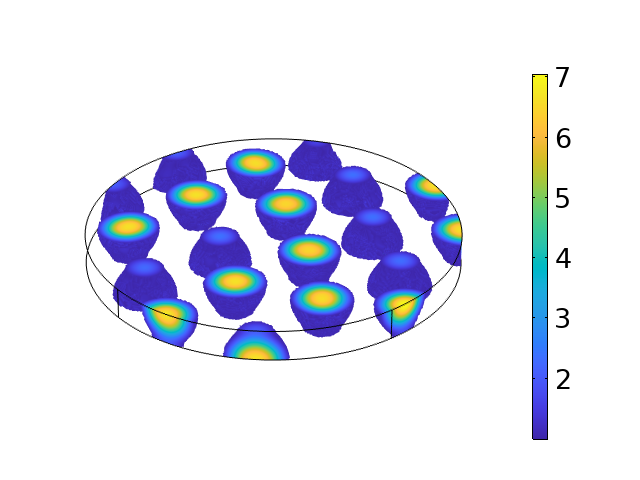}
\includegraphics[width=0.4\textwidth]{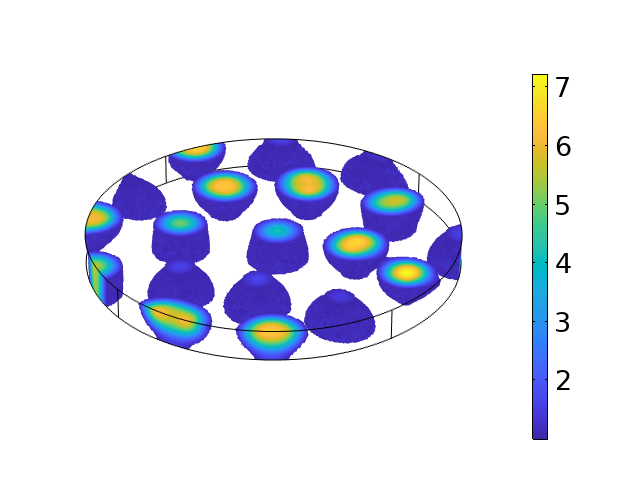}

(a)iii \hspace{6cm} (b)iii

\includegraphics[width=0.4\textwidth]{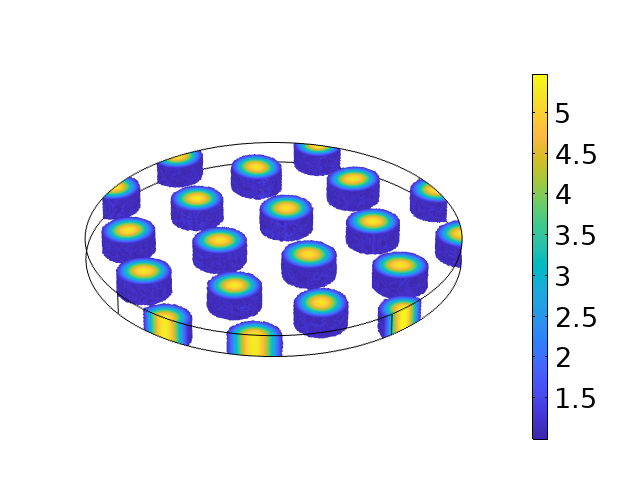}
\includegraphics[width=0.4\textwidth]{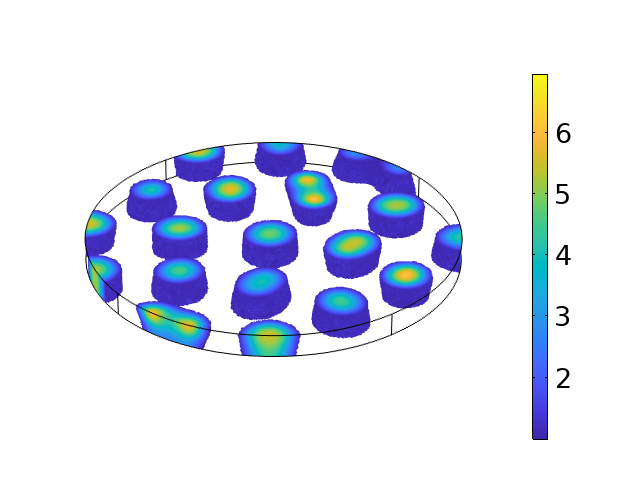}

(a)iv \hspace{6cm} (b)iv

\vspace{-0.1in}
\caption{Plots of $u_1$ corresponding to the kinetics \eqref{Schnack} with parameters $a=0$, $b=1.1$, $d_1=1$, and $d_2=50$ on a domain which is taken to evolve with radius as $r_1(t) = 7.5\sqrt{1+st}$ and height $r_2(t) = 60/(1+st)$ with (a) $s=0.001$, (b) $s=0.1$. The final time $t_f$ is selected so that $r_1(t_f)=30$, $r_2(t_f)=3.75$. Panels in (a,b) are shown at times $t=0.013335t_f, 0.2t_f, 0.8t_f, t_f$ in (i)-(iv), respectively. \label{cylinplots}}
\end{figure}

\begin{figure}
\centering

\includegraphics[width=0.45\textwidth]{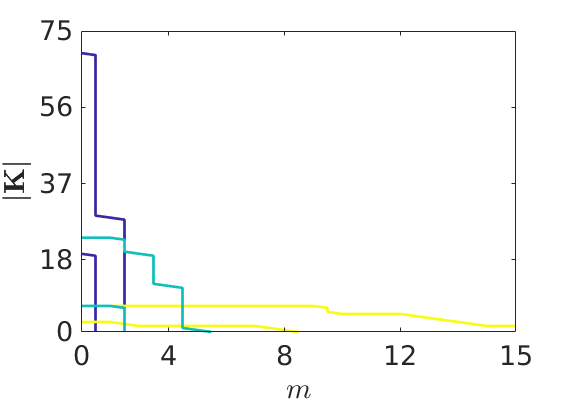}
\includegraphics[width=0.45\textwidth]{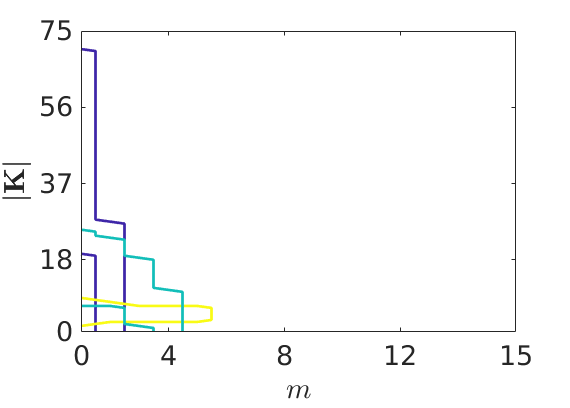}

\hspace{0.2cm} (a) \hspace{6.6cm} (b)

\vspace{-0.1in}
\caption{Plots of dispersion sets $\mathcal{K}_t$ at fixed times, with (c)i corresponding to the dynamics in (a), and (c)ii corresponding to the dynamics in (b). In particular, we plot $\mathcal{K}_0$ (corresponding to integer pairs $(|\mathbf{K}|,m)$ bounded by yellow curves), $\mathcal{K}_{0.2t_{f}}$ (bounded by teal curves), and $\mathcal{K}_{0.8t_f}$ (bounded by purple curves). We use the notation $|\mathbf{K}|$ to denote a sequential numbering of the eigenvalues $j_{\ell,k}$ ordered by magnitude. \label{cylinplotsStab}}
\end{figure}

We show simulations of such evolving cylinders in Fig.~\ref{cylinplots}, where we threshold the solutions to focus on the regions of high $u_1$ concentration. For this choice of parameters, the typical three--dimensional Turing structures are arrangements of spheres of high activator concentration, with partial spheres on the boundary. We see these spheres form quickly, and their number and placement change during slow growth. Eventually these structures flatten into quasi-two-dimensional cylinders near the end of the simulation shown in (a). In contrast, the structures in (b) do not have time to organize into spheres as the domain evolves more rapidly, and so we see multiple regions coalescing and mixing. The flattening and distortion of three-dimensional structures was discussed in \cite{hunding1985morphogen}, but the influence of rapid domain evolution was neglected, and it is clear that it plays a role in the emergent structures for such a domain. Dispersion sets are shown for each case in Fig.~\ref{cylinplotsStab}, and as in Fig.~\ref{rectangplotsStab}, we observe that modes of similar magnitudes are excited throughout the simulation times, but that modes corresponding to the vertical and radial directions are excited at different times. Note that rapid growth can especially influence transient mode selection (consider $\mathcal{K}_0$), and hence depending on the nonlinearities involved, the final patterned state.

\section{Applications to systems with higher-order spatial derivatives}\label{shsec}
As a specific example of a high-order spatial system, we focus on the scalar Swift-Hohenberg equation \cite{swift1977hydrodynamic}. In addition to being a canonical model for pattern formation on static domains, we remark that in the past few years this equation has been studied in evolving domains \cite{knobloch2014stability, krechetnikov2017stability}, in particular on one-dimensional intervals of the form $[0,r(t)]$. We choose to study the following form of the Swift-Hohenberg equation:
\begin{equation}\label{sh1}
\frac{\partial u}{\partial t} +\nabla_{\Omega(t)} \cdot (\mathbf{Q}u) = - d\left( q + \nabla_{\Omega(t)}^2 \right)^2 u + \left( -a + dq^2\right)u - u^3\,,
\end{equation}
where $d>0$ is the diffusion coefficient, while $q$ and $a$ are positive parameters. We choose this form as it admits a single stable steady state, $u=0$. We again mention that dynamics from the complex Swift-Hohenberg equation, with $u^3$ replaced by $|u|^2u$ and the sign of $a$ reversed, were considered on an evolving domain by \cite{knobloch2014stability} and \cite{krechetnikov2017stability}, with the domain being a time-dependent interval. Therefore, the Swift-Hohenberg equation is a natural canonical example of a fourth-order equation which permits pattern formation in the presence of an evolving domain. In order to match \eqref{sh1} with the theory of Sections \ref{hoformulation} and \ref{instabhigher}, we note that the choice of $\mathcal{P}(y)= - y^2 + 2qy$ corresponds to \eqref{sh1}.

As \eqref{sh1} is a good deal different from the reaction-diffusion equations \eqref{RD_system:eqn} we have studied thus far, we first consider the Turing instability condition for \eqref{sh1} on a static domain. In the static case, $\Omega(t)=\Omega(0)$. The zero steady state is always stable to homogeneous perturbations. A perturbation of the form \eqref{pert2} of this steady state is governed by
\begin{equation}
\frac{dV_k}{dt} = \left( -d\left( \rho_k^2 - 2q\rho_k\right)-a\right)V_k\,.
\end{equation}
The perturbation corresponding to the $k$th eigenfunction is unstable provided $-d\left( \rho_k^2 - 2q\rho_k\right)-a >0$, which is in turn possible provided that $\rho_k$ satisfies the inequality
\begin{equation}\label{sh2}
q - \sqrt{q^2 - \frac{a}{d}} < \rho_k < q + \sqrt{q^2 - \frac{a}{d}}\,.
\end{equation}
Clearly, we must have that $q > \sqrt{\frac{a}{d}}$, in order for there to be the possibility of a Turing instability. As $\rho_k \rightarrow \infty$ with $k\rightarrow \infty$, there are at most finitely many such $\rho_k$ satisfying \eqref{sh2}, as in the standard results on Turing instabilities in systems of two reaction-diffusion equations with second-order space derivatives. 

Returning to the dynamics of \eqref{sh1} on an evolving domain, an application of Theorem \ref{Thm5} gives that the condition for transient instability of the $k$th perturbation \eqref{pert2} is
\begin{equation}\label{sh3}
2dq \rho_k(t) - d\rho_k(t)^2 - a - \frac{\dot{\mu}(t)}{\mu(t)} >0 \,,
\end{equation}
where the base state $U(t)=0$ no longer depends on time. Therefore, there is a transient instability due to the $k$th spatial mode for $t\in \mathcal{I}_k$ provided that $\rho_k(t)$ satisfies
\begin{equation}
q - \sqrt{q^2 - \frac{1}{d}\left( a  + \frac{\dot{\mu}(t)}{\mu(t)}\right)} < \rho_k(t) < q + \sqrt{q^2 - \frac{1}{d}\left( a + \frac{\dot{\mu}(t)}{\mu(t)}\right)}
\end{equation}
for all $t\in \mathcal{I}_k$. 

We remark that this analysis is substantially simpler than that for the reaction-diffusion systems in the preceding sections for two reasons. Firstly it is simpler due to the scalar nature of the condition from Theorem \ref{Thm5}, and secondly because of the simple choice of kinetics which give a trivial base state (and hence no nonautonomous impacts from the base state's variation). Even in cases with kinetics admitting a nontrivial base state, we remark that such a function would evolve according to a scalar ODE, and hence essentially be slaved to the dynamics of $\frac{\dot{\mu}(t)}{\mu(t)}$.

We further remark that although we have used the spatial perturbation \eqref{pert2} corresponding to a Laplace-Beltrami eigenfunction, and hence invoked Theorem \ref{Thm5}, we shall later see the linear instability result is in good agreement with the patterns which emerge from numerical simulations. Therefore, the simplified perturbation and resulting instability criteria given in Theorem \ref{Thm5} did not discard any useful information about the onset of instability. As such, we did not have to consider the more complicated spectral problem for the fourth-order spatial operator in order to obtain a more general spectrum as used in the instability criteria of Theorem \ref{Thm6}. That said, for more complicated spatial operators $\mathcal{P}$ or more complicated spatial domains $\Omega(t)$, one should be aware that the instability criteria of Theorem \ref{Thm6} may be necessary.

\begin{figure}
\centering
\includegraphics[width=0.45\textwidth]{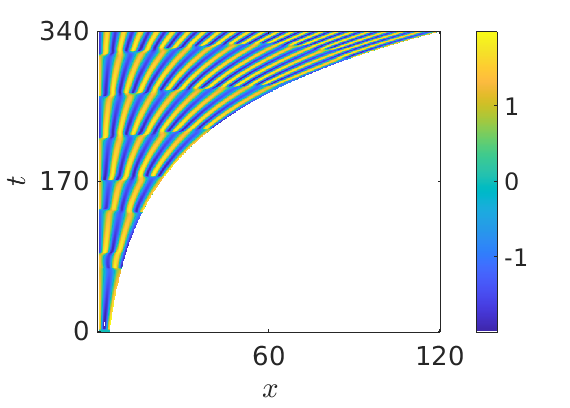}
\includegraphics[width=0.4\textwidth]{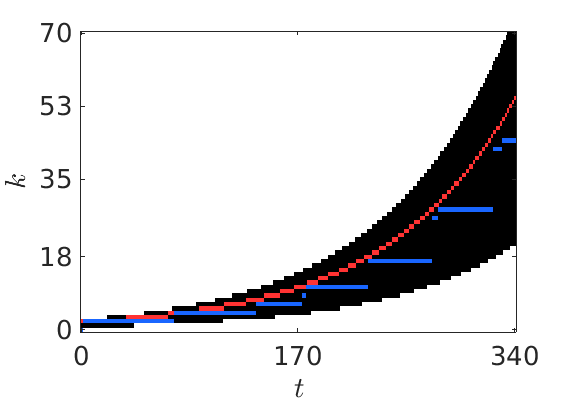}

a(i) \hspace{6cm} a(ii) 

\includegraphics[width=0.45\textwidth]{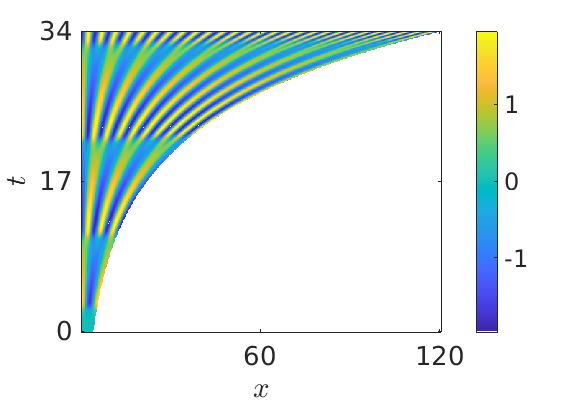}
\includegraphics[width=0.4\textwidth]{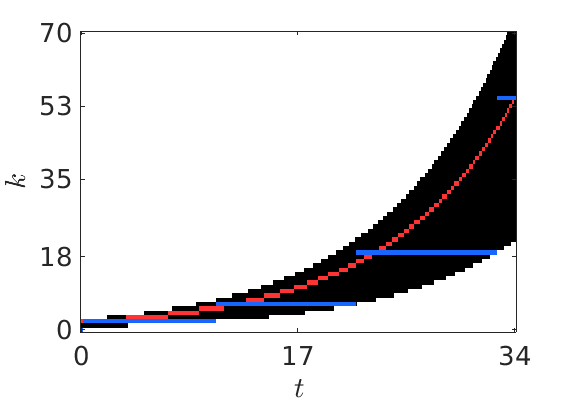}

b(i) \hspace{6cm} b(ii) 

\vspace{-0.1in}
\caption{Plots corresponding to the Swift-Hohenberg equation \eqref{sh1} on a growing interval with parameters $a=1$, $d=1$, and $q=2$. The domain is taken to grow as $r(t) = 4\exp(st)$ for growth rates (a) $s=0.01$ and (b) $s= 0.1$. In all simulations we take the final time such that the domain has grown to $30$ times its initial size. In column (i) we show plots of the PDE solution $u$ over space and time. In column (ii) we plot the dispersion set $\mathcal{K}_t$ in black, with the theoretically maximally growing mode in red and the largest frequency component of the FFT of $u(x,t)$ from the full numerical solution in blue. NB: The temporal and mode axes have different ranges for different growth rates.\label{SH1Dexp}}
\end{figure}

\begin{figure}
\centering
\includegraphics[width=0.45\textwidth]{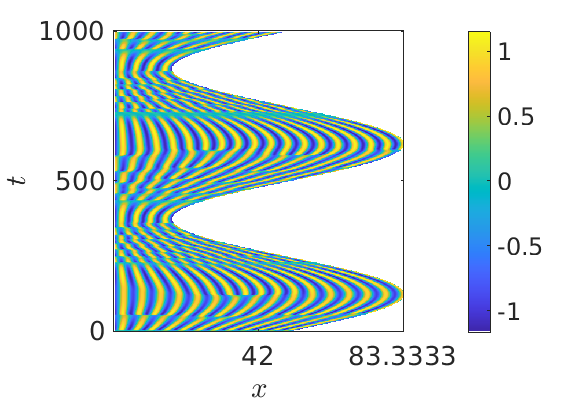}
\includegraphics[width=0.4\textwidth]{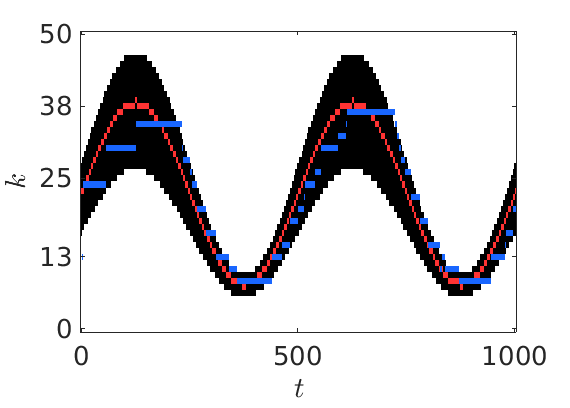}

a(i) \hspace{6cm} a(ii) 

\includegraphics[width=0.45\textwidth]{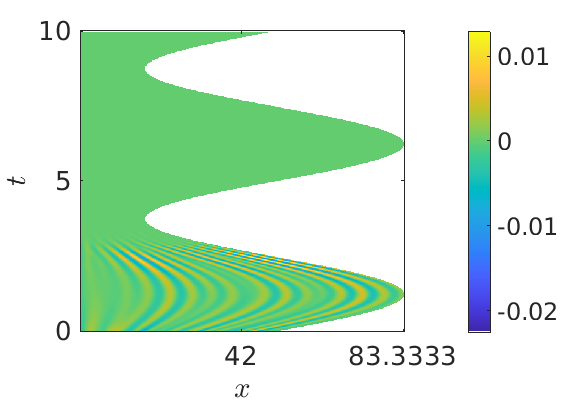}
\includegraphics[width=0.4\textwidth]{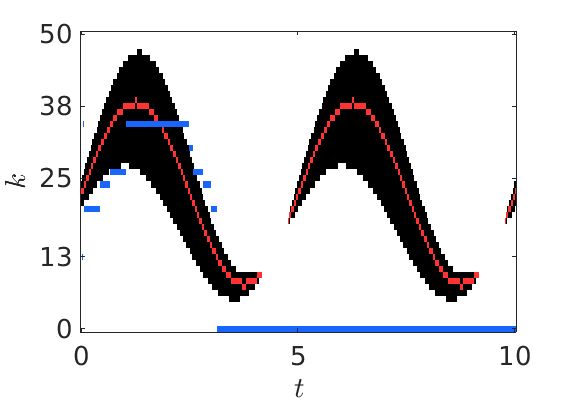}

b(i) \hspace{6cm} b(ii) 

\vspace{-0.1in}
\caption{Plots corresponding to the Swift-Hohenberg equation \eqref{sh1} on a growing interval with parameters $a=3$, $d=1$, and $q=2$. The domain is taken to grow as $r(t) = 50(1+(2/3)\sin(4\pi t/t_f))$ for time periods (a) $t_f=10^3$ and (b) $10$. In all simulations we take the final time such that the domain has grown to $30$ times its initial size. In column (i) we show plots of the PDE solution $u$ over space and time. In column (ii) we plot the dispersion set $\mathcal{K}_t$ in black, with the theoretically maximally growing mode in red and the largest frequency component of the FFT of $u(x,t)$ from the full numerical solution in blue. NB: The temporal and mode axes have different ranges for different growth rates.\label{SH1Dsin}}
\end{figure}

\subsection{Swift-Hohenberg on an evolving interval}

As in Section \ref{interval}, we first consider the dynamics of the Swift-Hohenberg equation \eqref{sh1} on an evolving line segment, $x \in [0, r(t)]$. As before, we consider the instability of modes with eigenvalues given by $\rho_k(t) = \frac{\pi^2 k^2}{r(t)^2}$, and volume expansion $\mu(t)=r(t)$. Broadly we find similar behaviours to that found in the reaction-diffusion setting, with the dispersion sets providing an approximation to the observed patterns. 

We first consider exponentially growing domains in Fig.~\ref{SH1Dexp}. As in Fig.~\ref{expplots}, we observe mode-doubling behaviour for some growth rates, and the observed modes are broadly characterized by the dispersion sets. However, we note that the dispersion sets and our instability analysis is much less sensitive to the rate of domain evolution than in the coupled reaction-diffusion system setting, which is likely due to the trivial base state. Nevertheless, full numerical simulations depend on the rate of domain evolution due to nonlinear pattern evolution, and hence we cannot capture the variations in unstable modes precisely.

Next we consider a periodically evolving interval in Fig.~\ref{SH1Dsin}. Qualitatively the dynamics are very similar to the Schnakenberg example on such a domain given in Fig.~\ref{sinplots}. We do remark that, again, the dispersion sets are more regular than in cases where the base state can exhibit complex dynamics. We also note that the linear analysis does not explain the loss of pattern in Fig.~\ref{sinplots}b(ii), where seemingly the rapid movement between unstable modes does not leave enough time for a net growth of any particular instability, and hence the transient pattern eventually falls back into the homogeneous equilibrium state. 

\subsection{Swift-Hohenberg on an evolving sphere}
Swift-Hohenberg equations and generalizations thereof have been studied on static spherical surfaces \cite{matthews2003pattern, sigrist2011symmetric}, with the model associated with the study of wrinkling and buckling of membranes \cite{stoop2015curvature}.

In Fig.~\ref{SH2Dspherelin} we give an example of solutions to \eqref{sh1} on the surface of a linearly growing 2-sphere. As anticipated from the one dimensional examples, the dispersion sets for increasing growth rate are very similar, having only small differences in the size of the unstable region at the beginning of the growth interval (this is due to the form of $\frac{\dot{\mu}(t)}{\mu(t)}$ for linear growth, which is much larger for small $t$). The simulations shown indicate that even this small variation can have a marked impact on the timescale at which patterns form; in the case of $s=0.01$, patterns are developed (e.g. $\max(u)\sim O(1)$) as early as $t= 7\approx 0.002t_f$, whereas for $s=0.15$, we see that the initial $O(10^{-1})$ perturbation has substantially decayed by $t=0.2t_f$, though a pattern does eventually form. As in the reaction-diffusion systems on manifolds shown before, these more rapidly grown manifolds lead to less regular patterns, though the wavelength is approximately the same.

\begin{figure}
\centering
\includegraphics[width=0.4\textwidth]{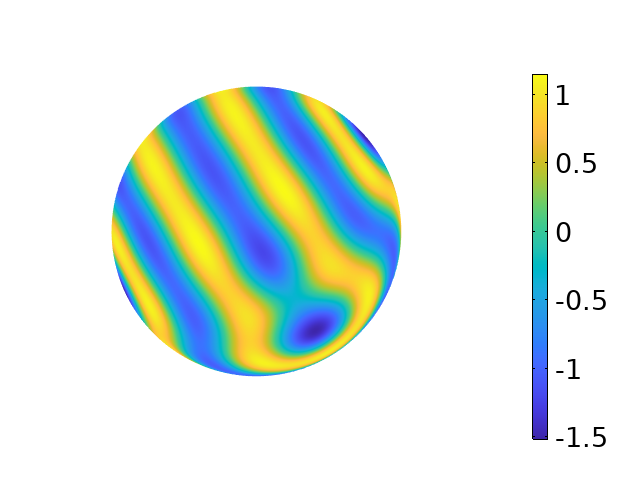}
\includegraphics[width=0.4\textwidth]{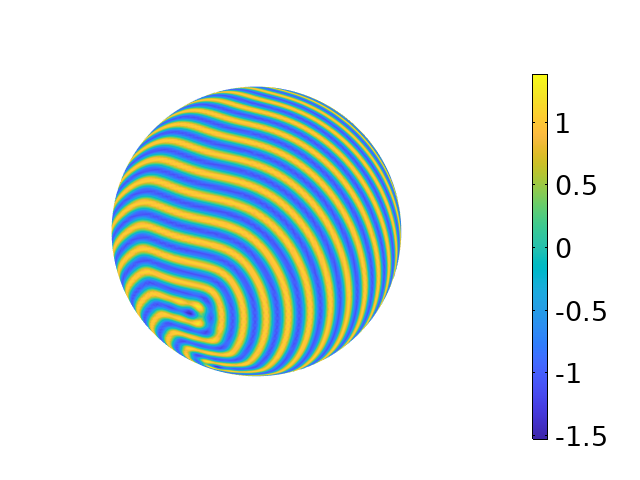}

a(i) \hspace{6cm} a(ii) 

\includegraphics[width=0.4\textwidth]{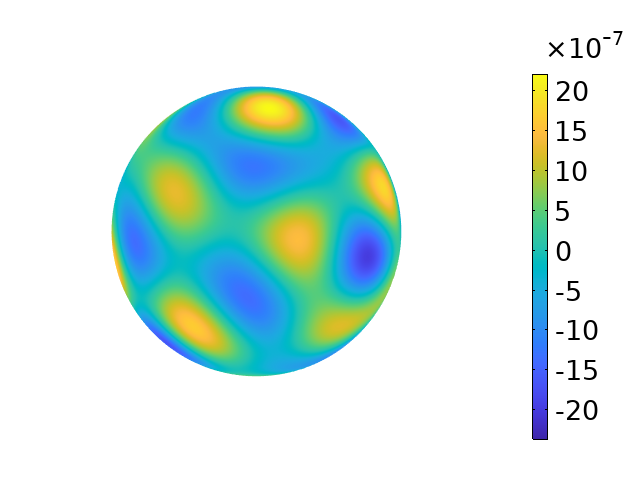}
\includegraphics[width=0.4\textwidth]{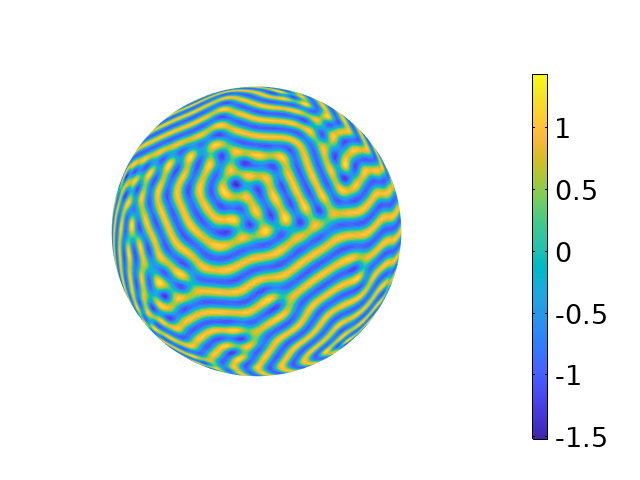}

b(i) \hspace{6cm} b(ii) 

\includegraphics[width=0.4\textwidth]{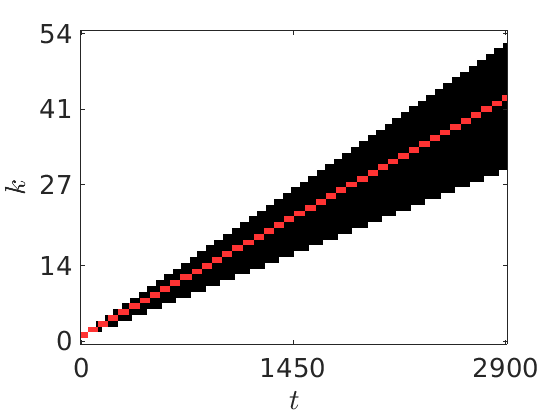}
\includegraphics[width=0.4\textwidth]{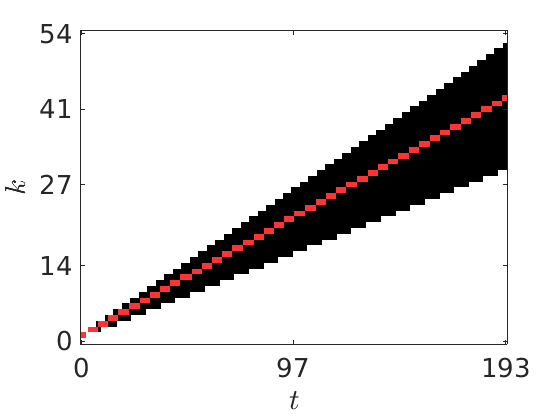}

c(i) \hspace{6cm} c(ii) 

\vspace{-0.1in}
\caption{Plots corresponding to the Swift-Hohenberg equation \eqref{sh1} on the surface of a growing sphere with parameters $a=3$, $d=1$, and $q=2$. The domain is taken to grow as $r(t) = (1+st)$ for growth rates (a) $s=0.01$ and (b) $s=0.15$, at times (i) $t=0.2t_f$ and (ii) $t=t_f$. For rows (a)-(b) we show plots of the PDE solution $u$ over space and time. Finally in row (c) we plot the dispersion set $\mathcal{K}_t$ in black for simulations corresponding to growth rates (i) $s=0.01$ and (ii) $s=0.15$. NB: The color scale in b(i) differs from the other plots, as the solution is still nearly zero.\label{SH2Dspherelin}}
\end{figure}

\section{Discussion}\label{sec5}
We have reviewed and generalized much of the contemporary work studying diffusion-driven instabilities in reaction-diffusion systems on growing domains. Our extensions include properly accounting for the non-autonomous nature of the base state of the system which is perturbed, allowing for general dilational evolution of the domain (of which a special case is the more commonly studied isotropic growth), and deriving a differential inequality (involving model parameters and the Laplace-Beltrami spectrum) to determine if a specific mode becomes unstable during a given time frame. Theorem \ref{Thm3} is a natural generalization of the Turing conditions on a static domain, yet explicitly accounts for the history-dependence due to the non-autonomous nature of the problem \cite{klika2017history}, and allows for arbitrary growth functions without the need to rely on slow growth or other simplifying assumptions. 

To summarize, the instability conditions given in Theorem \ref{Thm3} take the form
\begin{equation}\label{cond1}
\text{reaction kinetic terms} - \left( d_1 J_{22} + d_2 J_{11}\right) \rho_k + d_1 d_2 \rho_k^2 < \text{domain evolution terms}.
\end{equation}
Here ``reaction kinetic terms" include terms resulting from the dynamics of \eqref{uniform} which have been linearized (including terms involving the Jacobian matrix $J(\mathbf{U}(t))$), while ``domain evolution terms" involve terms which are specific to the manner of domain growth, such as terms involving the time derivatives of $\mu(t)$ and $\rho_k(t)$. If the instability is diffusion driven, then for the $k=0$ mode $\rho_k(t) \equiv 0$, we should have
\begin{equation}\label{cond2}
\text{reaction kinetic terms} > \text{domain evolution terms}.
\end{equation}
If \eqref{cond2} does not hold, then there is some homogeneous instability not due to diffusion. If \eqref{cond2} holds, while \eqref{cond1} holds for a particular index $k = k^* >0$ and a particular interval $\mathcal{I}_{k^*}$, then the spatial perturbation \eqref{pert2} corresponding to this particular $k^*$ inducing an instability for $t\in \mathcal{I}_{k^*}$. When the domain is static, the domain evolution terms vanish, and we are left with instability conditions of the form 
\begin{equation}\label{cond3}
\text{reaction kinetic terms} - \left( d_1 J_{22} + d_2 J_{11}\right) \rho_k + d_1 d_2 \rho_k^2 < 0,
\end{equation}
while the $k=0$ mode \eqref{cond2} is stable when reaction kinetic terms are positive, which is just the classical Turing condition on a static domain. Therefore, the condition given in Theorem \ref{Thm3} is a fairly natural generalization of the classical Turing conditions. Hence, although we have made few assumptions, and have avoided both asymptotic approximations of growth functions or assuming a constant base state, within our general framework we have captured the spirit of the original Turing conditions. In the limit of no growth, our results recover the Turing conditions for a static domain in a completely natural manner, without further effort or appeal to simplifications. Similarly, instability conditions for particular kinds of slow growth known in the literature (such as slow exponential growth), as well as quasi-static approximations for the slow growth regime, fall out of our results in the relevant limits.

Due to the time-dependent nature of growth terms and of the Laplace-Beltrami spectra, we have phrased our results in terms of instabilities present over a given time interval, rather than as $t\rightarrow \infty$ like in the classical Turing instability. This more general approach allows for the understanding of transient instabilities. This is a useful generalization, as in practice Turing patterns are selected in finite time, with patterns then remaining static once formed. As such, there is an interplay between the rate of growth of the domain and the rate at which this mode selection occurs. Hence, pattern formation will rely strongly on when certain modes result in instability on the timescale of the reaction kinetics, rather than simply if such modes ever induce instability at any arbitrary time. In contrast to the standard Turing theory for static domains, we conjecture that it may be the time duration for which a mode remains unstable that matters more than the degree or magnitude of instability at any instantaneous time (which is often considered by comparing the real part of eigenvalues in the static case in the limit $t\rightarrow \infty$). Of course, one would need nonlinear theory to address such issues (which is beyond the scope of the present paper), and even then, resolution is likely only on a case-by-case basis for given reaction kinetics, growth functions, and spatial patterns. 

Our numerical examples illustrate both the power of our analytical results (the dispersion sets obtained broadly match simulations, with some exception), and the large variety of phenomena one can expect when studying such problems. We remark that we intentionally only considered a small fraction of the parameter spaces for only two kinds of reaction kinetics, and already have observed behaviours which are qualitatively distinct from what has been commonly observed in related Turing systems in the past. In particular, we note that our main choice of the Schnackenberg kinetics \eqref{Schnack} gives a relatively simple Jacobian, and a base state with an uncomplicated evolution equation. More complicated reaction kinetics can likely lead to many new phenomena due to the complexity of non-autonomous phase spaces. Our results will immediately apply in such cases, providing insight into spatial instabilities around even a time-dependent base state, as shown for instance in Sec.~\ref{sec_FHN}.

Similarly, we have only considered a handful of growth functions and domain geometries, but unlike many results in the literature, the instability condition in Theorem \ref{Thm3} applies so long as one can compute derivatives of the growth functions and the time-dependent Laplace-Beltrami spectrum of the domain. We have used the cases of volume-preserving evolution and the Swift-Hohenberg example to show that even when the base state is static, domain evolution can substantially change the unstable modes observed, along with the qualitative behavior of spatial patterns altogether (cf.~Sec.~\ref{sec44}). Hence, changes in the structure of a domain are sufficient to modify the linear stability properties, as well as the patterns formed, even if the area or volume of the domain is fixed. Indeed, these examples show some of the largest discrepancies between quasi-static or asymptotically large-time approaches common in the literature (which tend to work best when evolution is slow and monotone) and pattern selection which actually occurs at transient timescales, as seen in the strong time dependence of not only the extent but also the shape of the dispersion sets $\mathcal{K}_t$ shown in Figs.~\ref{rectangplots}-\ref{cylinplots}.

We have also demonstrated the application of the same basic instability analysis on a canonical higher-order scalar spatial system, the Swift-Hohenberg equation. This demonstrates the generality of our approach, which is to be expected of a linear stability analysis. Additionally, this model is a much simpler test bed for exploring a variety of nonlinear kinetics and complicated evolution scenarios, as one can explore the influence of growth in the absence of a base state which can exhibit complex temporal dynamics. Extensions to planar domains could also be considered as generalizations of planar studies in the literature \cite{lloyd2008localized, m1995pattern}. These would be especially simple in the case of area-preserving evolution, as the nonautonomous dynamics would then be entirely embedded in the time-dependent spectral parameters.

The framework developed in this paper may also be applied to problems with time-dependent reaction kinetics, on either growing or static domains. In particular, temporal oscillations have been employed in photosensitive reactions to control Turing patterns, and in some cases eliminate them \cite{dolnik2001resonant, horvath1999control, wang2006resonant}. Spatiotemporal forcing has also been used to mimic domain growth in such systems \cite{konow2019turing, miguez2006effect, miguez2005turing, rudiger2003dynamics}. As our approach allows for non-autonomous Jacobian matrices, such applications will naturally benefit from our approach. Another natural extension is to consider the problem of domain evolution simultaneous to an imposed flow, to study problems for which the velocity due to dilution is augmented with a velocity field which transports the chemical species or morphogens in some manner. Recently, the flow properties and cross-section geometry for activator-inhibitor systems within a tube were shown to influence emergent Turing patterns \cite{van2019diffusive}. The extension of such results to scenarios where the tube dilates periodically in time would be one such extension incorporating evolving spatial domains which would have biochemical and physiological relevance. Additionally, spatial heterogeneity in reaction-diffusion systems is also of contemporary interest, and incorporating this alongside growth would lead to a much more biologically realistic extension of Turing's theory \cite{page2005complex, krause2019one}. We also mention that many other contemporary problems in a range of fields consist of systems on time-evolving domains \cite{knobloch2014stability, knobloch2015problems}. While a general theory of such systems does not exist, understanding the impact of growth for reaction-diffusion systems can provide an important example for these larger classes of systems. We anticipate that the results presented here can be readily generalized to other settings.


\begin{thebibliography}{100}

\bibitem{amar2013anisotropic}
{\sc M.~B. Amar and F.~Jia}, {\em Anisotropic growth shapes intestinal tissues
  during embryogenesis}, Proceedings of the National Academy of Sciences, 110
  (2013), pp.~10525--10530.

\bibitem{baker2008partial}
{\sc R.~E. Baker, E.~Gaffney, and P.~Maini}, {\em Partial differential
  equations for self-organization in cellular and developmental biology},
  Nonlinearity, 21 (2008), p.~R251.

\bibitem{bansagi2011tomography}
{\sc T.~B{\'a}ns{\'a}gi, V.~K. Vanag, and I.~R. Epstein}, {\em Tomography of
  reaction-diffusion microemulsions reveals three-dimensional turing patterns},
  Science, 331 (2011), pp.~1309--1312.

\bibitem{barrass2006mode}
{\sc I.~Barrass, E.~J. Crampin, and P.~K. Maini}, {\em Mode transitions in a
  model reaction--diffusion system driven by domain growth and noise}, Bulletin
  of Mathematical Biology, 68 (2006), pp.~981--995.

\bibitem{barreira2011surface}
{\sc R.~Barreira, C.~M. Elliott, and A.~Madzvamuse}, {\em The surface finite
  element method for pattern formation on evolving biological surfaces},
  Journal of Mathematical Biology, 63 (2011), pp.~1095--1119.

\bibitem{beloussov2003geometro}
{\sc L.~Beloussov and V.~Grabovsky}, {\em A geometro-mechanical model for
  pulsatile morphogenesis}, Computer Methods in Biomechanics \& Biomedical
  Engineering, 6 (2003), pp.~53--63.

\bibitem{binder2008modeling}
{\sc B.~J. Binder, K.~A. Landman, M.~J. Simpson, M.~Mariani, and D.~F.
  Newgreen}, {\em Modeling proliferative tissue growth: a general approach and
  an avian case study}, Physical Review E, 78 (2008), p.~031912.

\bibitem{bittig2008dynamics}
{\sc T.~Bittig, O.~Wartlick, A.~Kicheva, M.~Gonz{\'a}lez-Gait{\'a}n, and
  F.~J{\"u}licher}, {\em Dynamics of anisotropic tissue growth}, New Journal of
  Physics, 10 (2008), p.~063001.

\bibitem{cahn1958free}
{\sc J.~W. Cahn and J.~E. Hilliard}, {\em Free energy of a nonuniform system.
  i. interfacial free energy}, The Journal of chemical physics, 28 (1958),
  pp.~258--267.

\bibitem{callahan1999pattern}
{\sc T.~Callahan and E.~Knobloch}, {\em Pattern formation in three-dimensional
  reaction--diffusion systems}, Physica D: Nonlinear Phenomena, 132 (1999),
  pp.~339--362.

\bibitem{castillo2016turing}
{\sc J.~A. Castillo, F.~S{\'a}nchez-Gardu{\~n}o, and P.~Padilla}, {\em A
  turing--hopf bifurcation scenario for pattern formation on growing domains},
  Bulletin of Mathematical Biology, 78 (2016), pp.~1410--1449.

\bibitem{chaplain2001spatio}
{\sc M.~A. Chaplain, M.~Ganesh, and I.~G. Graham}, {\em Spatio-temporal pattern
  formation on spherical surfaces: numerical simulation and application to
  solid tumour growth}, Journal of Mathematical Biology, 42 (2001),
  pp.~387--423.

\bibitem{coen2004genetics}
{\sc E.~Coen, A.-G. Rolland-Lagan, M.~Matthews, J.~A. Bangham, and
  P.~Prusinkiewicz}, {\em The genetics of geometry}, Proceedings of the
  National Academy of Sciences, 101 (2004), pp.~4728--4735.

\bibitem{cohen1981generalized}
{\sc D.~S. Cohen and J.~D. Murray}, {\em A generalized diffusion model for
  growth and dispersal in a population}, Journal of Mathematical Biology, 12
  (1981), pp.~237--249.

\bibitem{colbois2019neumann}
{\sc B.~Colbois and L.~Provenzano}, {\em Neumann eigenvalues of the biharmonic
  operator on domains: geometric bounds and related results}, arXiv preprint
  arXiv:1907.02252,  (2019).

\bibitem{comanici2008patterns}
{\sc A.~Comanici and M.~Golubitsky}, {\em Patterns on growing square domains
  via mode interactions}, Dynamical Systems, 23 (2008), pp.~167--206.

\bibitem{corson2009turning}
{\sc F.~Corson, O.~Hamant, S.~Bohn, J.~Traas, A.~Boudaoud, and Y.~Couder}, {\em
  Turning a plant tissue into a living cell froth through isotropic growth},
  Proceedings of the National Academy of Sciences, 106 (2009), pp.~8453--8458.

\bibitem{crampin2002pattern}
{\sc E.~Crampin, W.~Hackborn, and P.~Maini}, {\em Pattern formation in
  reaction-diffusion models with nonuniform domain growth}, Bulletin of
  Mathematical Biology, 64 (2002), pp.~747--769.

\bibitem{crampin2001reaction}
{\sc E.~Crampin and P.~Maini}, {\em Reaction-diffusion models for biological
  pattern formation}, Methods and Applications of Analysis, 8 (2001),
  pp.~415--428.

\bibitem{Crampin}
{\sc E.~J. Crampin, E.~A. Gaffney, and P.~K. Maini}, {\em Reaction and
  diffusion on growing domains: scenarios for robust pattern formation},
  Bulletin of Mathematical Biology, 61 (1999), pp.~1093--1120.

\bibitem{cross1993pattern}
{\sc M.~C. Cross and P.~C. Hohenberg}, {\em Pattern formation outside of
  equilibrium}, Reviews of Modern Physics, 65 (1993), p.~851.

\bibitem{de1997twist}
{\sc A.~De~Wit, P.~Borckmans, and G.~Dewel}, {\em Twist grain boundaries in
  three-dimensional lamellar turing structures}, Proceedings of the National
  Academy of Sciences, 94 (1997), pp.~12765--12768.

\bibitem{dhillon2017bifurcation}
{\sc D.~S.~J. Dhillon, M.~C. Milinkovitch, and M.~Zwicker}, {\em Bifurcation
  analysis of reaction diffusion systems on arbitrary surfaces}, Bulletin of
  Mathematical Biology, 79 (2017), pp.~788--827.

\bibitem{dolnik2001resonant}
{\sc M.~Dolnik, A.~M. Zhabotinsky, and I.~R. Epstein}, {\em Resonant
  suppression of turing patterns by periodic illumination}, Physical Review E,
  63 (2001), p.~026101.

\bibitem{drasdo2001individual}
{\sc D.~Drasdo and M.~Loeffler}, {\em Individual-based models to growth and
  folding in one-layered tissues: intestinal crypts and early development},
  Nonlinear Analysis-Theory Methods and Applications, 47 (2001), pp.~245--256.

\bibitem{ermentrout1991stripes}
{\sc B.~Ermentrout}, {\em Stripes or spots? nonlinear effects in bifurcation of
  reaction—diffusion equations on the square}, Proceedings of the Royal
  Society of London. Series A: Mathematical and Physical Sciences, 434 (1991),
  pp.~413--417.

\bibitem{facsko2004dissipative}
{\sc S.~Facsko, T.~Bobek, A.~Stahl, H.~Kurz, and T.~Dekorsy}, {\em Dissipative
  continuum model for self-organized pattern formation during ion-beam
  erosion}, Physical Review B, 69 (2004), p.~153412.

\bibitem{feijo2001cellular}
{\sc J.~A. Feij{\'o}, J.~Sainhas, T.~Holdaway-Clarke, M.~S. Cordeiro, J.~G.
  Kunkel, and P.~K. Hepler}, {\em Cellular oscillations and the regulation of
  growth: the pollen tube paradigm}, Bioessays, 23 (2001), pp.~86--94.

\bibitem{fitzhugh1955mathematical}
{\sc R.~FitzHugh}, {\em Mathematical models of threshold phenomena in the nerve
  membrane}, Bulletin of Mathematical Biology, 17 (1955), pp.~257--278.

\bibitem{gerber2003riccati}
{\sc M.~Gerber, B.~Hasselblatt, and D.~Keesing}, {\em The riccati equation:
  pinching of forcing and solutions}, Experimental Mathematics, 12 (2003),
  pp.~129--134.

\bibitem{ghadiri_krechetnikov_2019}
{\sc M.~Ghadiri and R.~Krechetnikov}, {\em Pattern formation on time-dependent
  domains}, Journal of Fluid Mechanics, 880 (2019), pp.~136--179.

\bibitem{gierer1972theory}
{\sc A.~Gierer and H.~Meinhardt}, {\em A theory of biological pattern
  formation}, Kybernetik, 12 (1972), pp.~30--39.

\bibitem{gjorgjieva2007turing}
{\sc J.~Gjorgjieva and J.~T. Jacobsen}, {\em Turing patterns on growing
  spheres: the exponential case}, Discrete and Continuous Dynamical Systems
  Supplement,  (2007), pp.~436--445.

\bibitem{grigoryan2015stability}
{\sc G.~Grigoryan}, {\em On the stability of systems of two first-order linear
  ordinary differential equations}, Differential Equations, 51 (2015),
  pp.~283--292.

\bibitem{hetzer2012characterization}
{\sc G.~Hetzer, A.~Madzvamuse, and W.~Shen}, {\em Characterization of turing
  diffusion-driven instability on evolving domains}, Discrete and Continuous
  Dynamical Systems-Series A, 32 (2012), pp.~3975--4000.

\bibitem{m1995pattern}
{\sc M.~Hilali, S.~M{\'e}tens, P.~Borckmans, and G.~Dewel}, {\em Pattern
  selection in the generalized swift-hohenberg model}, Physical Review E, 51
  (1995), p.~2046.

\bibitem{horvath1999control}
{\sc A.~K. Horv{\'a}th, M.~Dolnik, A.~P. Munuzuri, A.~M. Zhabotinsky, and I.~R.
  Epstein}, {\em Control of turing structures by periodic illumination},
  Physical Review Letters, 83 (1999), p.~2950.

\bibitem{hunding1985morphogen}
{\sc A.~Hunding}, {\em Morphogen prepatterns during mitosis and cytokinesis in
  flattened cells: Three dimensional turing structures of reaction-diffusion
  systems in cylindrical coordinates}, Journal of Theoretical Biology, 114
  (1985), pp.~571--588.

\bibitem{hyman1986kuramoto}
{\sc J.~M. Hyman and B.~Nicolaenko}, {\em The kuramoto-sivashinsky equation: a
  bridge between pde's and dynamical systems}, Physica D: Nonlinear Phenomena,
  18 (1986), pp.~113--126.

\bibitem{ince}
{\sc E.~Ince}, {\em Ordinary Differential Equations}, Dover, 1956.

\bibitem{josic2008unstable}
{\sc K.~Josi{\'c} and R.~Rosenbaum}, {\em Unstable solutions of nonautonomous
  linear differential equations}, SIAM Review, 50 (2008), pp.~570--584.

\bibitem{keener1998mathematical}
{\sc J.~P. Keener and J.~Sneyd}, {\em Mathematical physiology}, vol.~1,
  Springer, 1998.

\bibitem{kielhofer1997pattern}
{\sc H.~Kielh{\"o}fer}, {\em Pattern formation of the stationary cahn-hilliard
  model}, Proceedings of the Royal Society of Edinburgh Section A: Mathematics,
  127 (1997), pp.~1219--1243.

\bibitem{klika2017significance}
{\sc V.~Klika}, {\em Significance of non-normality-induced patterns: Transient
  growth versus asymptotic stability}, Chaos: An Interdisciplinary Journal of
  Nonlinear Science, 27 (2017), p.~073120.

\bibitem{klika2017history}
{\sc V.~Klika and E.~A. Gaffney}, {\em History dependence and the continuum
  approximation breakdown: the impact of domain growth on turing’s
  instability}, Proceedings of the Royal Society A: Mathematical, Physical and
  Engineering Sciences, 473 (2017), p.~20160744.

\bibitem{knobloch2014stability}
{\sc E.~Knobloch and R.~Krechetnikov}, {\em Stability on time-dependent
  domains}, Journal of Nonlinear Science, 24 (2014), pp.~493--523.

\bibitem{knobloch2015problems}
{\sc E.~Knobloch and R.~Krechetnikov}, {\em Problems on time-varying domains:
  Formulation, dynamics, and challenges}, Acta Applicandae Mathematicae, 137
  (2015), pp.~123--157.

\bibitem{kondo2010reaction}
{\sc S.~Kondo and T.~Miura}, {\em Reaction-diffusion model as a framework for
  understanding biological pattern formation}, Science, 329 (2010),
  pp.~1616--1620.

\bibitem{konow2019turing}
{\sc C.~Konow, N.~H. Somberg, J.~Chavez, I.~R. Epstein, and M.~Dolnik}, {\em
  Turing patterns on radially growing domains: experiments and simulations},
  Physical Chemistry Chemical Physics, 21 (2019), pp.~6718--6724.

\bibitem{korzec2008stationary}
{\sc M.~D. Korzec, P.~L. Evans, A.~M{\"u}nch, and B.~Wagner}, {\em Stationary
  solutions of driven fourth-and sixth-order cahn--hilliard-type equations},
  SIAM Journal on Applied Mathematics, 69 (2008), pp.~348--374.

\bibitem{krause2018emergent}
{\sc A.~L. Krause, A.~M. Burton, N.~T. Fadai, and R.~A. Van~Gorder}, {\em
  Emergent structures in reaction-advection-diffusion systems on a sphere},
  Physical Review E, 97 (2018), p.~042215.

\bibitem{krause2018influence}
{\sc A.~L. Krause, M.~A. Ellis, and R.~A. Van~Gorder}, {\em Influence of
  curvature, growth, and anisotropy on the evolution of turing patterns on
  growing manifolds}, Bulletin of Mathematical Biology,  (2018), pp.~1--41.

\bibitem{krause2019one}
{\sc A.~L. Krause, V.~Klika, T.~E. Woolley, and E.~A. Gaffney}, {\em From one
  pattern into another: Analysis of turing patterns in heterogeneous domains
  via wkbj}, arXiv preprint arXiv:1908.07219,  (2019).

\bibitem{krechetnikov2017stability}
{\sc R.~Krechetnikov and E.~Knobloch}, {\em Stability on time-dependent
  domains: convective and dilution effects}, Physica D: Nonlinear Phenomena,
  342 (2017), pp.~16--23.

\bibitem{kuramoto1978diffusion}
{\sc Y.~Kuramoto}, {\em Diffusion-induced chaos in reaction systems}, Progress
  of Theoretical Physics Supplement, 64 (1978), pp.~346--367.

\bibitem{laptev1997dirichlet}
{\sc A.~Laptev}, {\em Dirichlet and neumann eigenvalue problems on domains in
  euclidean spaces}, Journal of Functional Analysis, 151 (1997), pp.~531--545.

\bibitem{leppanen2002new}
{\sc T.~Lepp{\"a}nen, M.~Karttunen, K.~Kaski, R.~A. Barrio, and L.~Zhang}, {\em
  A new dimension to turing patterns}, Physica D: Nonlinear Phenomena, 168
  (2002), pp.~35--44.

\bibitem{levine1985unrestricted}
{\sc H.~Levine, M.~Protter, and L.~Payne}, {\em Unrestricted lower bounds for
  eigenvalues for classes of elliptic equations and systems of equations with
  applications to problems in elasticity}, Mathematical methods in the applied
  sciences, 7 (1985), pp.~210--222.

\bibitem{liu2006two}
{\sc R.~Liu, S.~Liaw, and P.~Maini}, {\em Two-stage turing model for generating
  pigment patterns on the leopard and the jaguar}, Physical review E, 74
  (2006), p.~011914.

\bibitem{lloyd2008localized}
{\sc D.~J. Lloyd, B.~Sandstede, D.~Avitabile, and A.~R. Champneys}, {\em
  Localized hexagon patterns of the planar swift--hohenberg equation}, SIAM
  Journal on Applied Dynamical Systems, 7 (2008), pp.~1049--1100.

\bibitem{macdonald2013simple}
{\sc C.~B. Macdonald, B.~Merriman, and S.~J. Ruuth}, {\em Simple computation of
  reaction--diffusion processes on point clouds}, Proceedings of the National
  Academy of Sciences, 110 (2013), pp.~9209--9214.

\bibitem{madzvamuse2008stability}
{\sc A.~Madzvamuse}, {\em Stability analysis of reaction-diffusion systems with
  constant coefficients on growing domains}, International Journal of Dynamical
  Systems and Differential Equations, 1 (2008), pp.~250--262.

\bibitem{madzvamuse2010stability}
{\sc A.~Madzvamuse, E.~A. Gaffney, and P.~K. Maini}, {\em Stability analysis of
  non-autonomous reaction-diffusion systems: the effects of growing domains},
  Journal of Mathematical Biology, 61 (2010), pp.~133--164.

\bibitem{maini2002implications}
{\sc P.~K. Maini, E.~J. Crampin, A.~Madzvamuse, A.~J. Wathen, and R.~D.
  Thomas}, {\em Implications of domain growth in morphogenesis}, in
  Mathematical Modelling \& Computing in Biology and Medicine: 5th ESMTB Conf,
  vol.~1, 2002, pp.~67--73.

\bibitem{marcon2012turing}
{\sc L.~Marcon and J.~Sharpe}, {\em Turing patterns in development: what about
  the horse part?}, Current Opinion in Genetics \& Development, 22 (2012),
  pp.~578--584.

\bibitem{matthews2003pattern}
{\sc P.~Matthews}, {\em Pattern formation on a sphere}, Physical Review E, 67
  (2003), p.~036206.

\bibitem{meinhardt1998models}
{\sc H.~Meinhardt, A.-J. Koch, and G.~Bernasconi}, {\em Models of pattern
  formation applied to plant development}, in Symmetry in Plants, World
  Scientific, 1998, pp.~723--758.

\bibitem{mendez2010reaction}
{\sc V.~Mendez, S.~Fedotov, and W.~Horsthemke}, {\em Reaction-transport
  systems: mesoscopic foundations, fronts, and spatial instabilities}, Springer
  Science \& Business Media, 2010.

\bibitem{menzel2005modelling}
{\sc A.~Menzel}, {\em Modelling of anisotropic growth in biological tissues},
  Biomechanics and Modeling in Mechanobiology, 3 (2005), pp.~147--171.

\bibitem{mierczynski2017instability}
{\sc J.~Mierczynski}, {\em Instability in linear cooperative systems of
  ordinary differential equations}, SIAM Review, 59 (2017), pp.~649--670.

\bibitem{miguez2006effect}
{\sc D.~G. M{\'\i}guez, M.~Dolnik, A.~P. Munuzuri, and L.~Kramer}, {\em Effect
  of axial growth on turing pattern formation}, Physical Review Letters, 96
  (2006), p.~048304.

\bibitem{miguez2005turing}
{\sc D.~G. M{\'\i}guez, V.~P{\'e}rez-Villar, and A.~P. Mu{\~n}uzuri}, {\em
  Turing instability controlled by spatiotemporal imposed dynamics}, Physical
  Review E, 71 (2005), p.~066217.

\bibitem{miura2006mixed}
{\sc T.~Miura, K.~Shiota, G.~Morriss-Kay, and P.~K. Maini}, {\em Mixed-mode
  pattern in doublefoot mutant mouse limb—turing reaction--diffusion model on
  a growing domain during limb development}, Journal of Theoretical Biology,
  240 (2006), pp.~562--573.

\bibitem{MurrayI}
{\sc J.~D. Murray}, {\em Mathematical Biology. I An Introduction.
  \textit{Interdisciplinary Applied Mathematics V. 18}}, Springer-Verlag New
  York, 2003.

\bibitem{Murray}
{\sc J.~D. Murray}, {\em Mathematical Biology. II Spatial Models and Biomedical
  Applications. \textit{Interdisciplinary Applied Mathematics V. 18}},
  Springer-Verlag New York, 2003.

\bibitem{nagumo1962active}
{\sc J.~Nagumo, S.~Arimoto, and S.~Yoshizawa}, {\em An active pulse
  transmission line simulating nerve axon}, Proceedings of the IRE, 50 (1962),
  pp.~2061--2070.

\bibitem{nechaeva2002rhythmic}
{\sc M.~V. Nechaeva and T.~M. Turpaev}, {\em Rhythmic contractions in chick
  amnio-yolk sac and snake amnion during embryogenesis}, Comparative
  Biochemistry and Physiology Part A: Molecular \& Integrative Physiology, 131
  (2002), pp.~861--870.

\bibitem{novick1984nonlinear}
{\sc A.~Novick-Cohen and L.~A. Segel}, {\em Nonlinear aspects of the
  cahn-hilliard equation}, Physica D: Nonlinear Phenomena, 10 (1984),
  pp.~277--298.

\bibitem{ochoa1984generalized}
{\sc F.~L. Ochoa}, {\em A generalized reaction diffusion model for spatial
  structure formed by motile cells}, Biosystems, 17 (1984), pp.~35--50.

\bibitem{page2005complex}
{\sc K.~M. Page, P.~K. Maini, and N.~A. Monk}, {\em Complex pattern formation
  in reaction--diffusion systems with spatially varying parameters}, Physica D:
  Nonlinear Phenomena, 202 (2005), pp.~95--115.

\bibitem{pawlow2011sixth}
{\sc I.~Paw{\l}ow and W.~M. Zajaczkowski}, {\em A sixth order cahn-hilliard
  type equation arising in oil-water-surfactant mixtures}, Communications on
  Pure \& Applied Analysis, 10 (2011), pp.~1823--1847.

\bibitem{peaucelle2015control}
{\sc A.~Peaucelle, R.~Wightman, and H.~H{\"o}fte}, {\em The control of growth
  symmetry breaking in the arabidopsis hypocotyl}, Current Biology, 25 (2015),
  pp.~1746--1752.

\bibitem{Maini}
{\sc R.~G. Plaza, F.~Sanchez-Garduno, P.~Padilla, R.~A. Barrio, and P.~K.
  Maini}, {\em The effect of growth and curvature on pattern formation},
  Journal of Dynamics and Differential Equations, 16 (2004), pp.~1093--1121.

\bibitem{pleijel1950eigenvalues}
{\sc {\AA}.~Pleijel}, {\em On the eigenvalues and eigenfunctions of elastic
  plates}, Communications on Pure and Applied Mathematics, 3 (1950), pp.~1--10.

\bibitem{preska2014target}
{\sc A.~Preska~Steinberg, I.~R. Epstein, and M.~Dolnik}, {\em Target turing
  patterns and growth dynamics in the chlorine dioxide--iodine--malonic acid
  reaction}, The Journal of Physical Chemistry A, 118 (2014), pp.~2393--2400.

\bibitem{raspopovic2014digit}
{\sc J.~Raspopovic, L.~Marcon, L.~Russo, and J.~Sharpe}, {\em Digit patterning
  is controlled by a bmp-sox9-wnt turing network modulated by morphogen
  gradients}, Science, 345 (2014), pp.~566--570.

\bibitem{rossi2016control}
{\sc F.~Rossi, N.~P. Duteil, N.~Yakoby, and B.~Piccoli}, {\em Control of
  reaction-diffusion equations on time-evolving manifolds}, in Decision and
  Control (CDC), 2016 IEEE 55th Conference on, IEEE, 2016, pp.~1614--1619.

\bibitem{rost1995anisotropic}
{\sc M.~Rost and J.~Krug}, {\em Anisotropic kuramoto-sivashinsky equation for
  surface growth and erosion}, Physical review letters, 75 (1995), p.~3894.

\bibitem{rozada2014stability}
{\sc I.~Rozada, S.~J. Ruuth, and M.~Ward}, {\em The stability of localized spot
  patterns for the brusselator on the sphere}, SIAM Journal on Applied
  Dynamical Systems, 13 (2014), pp.~564--627.

\bibitem{rudiger2003dynamics}
{\sc S.~R{\"u}diger, D.~Miguez, A.~Munuzuri, F.~Sagu{\'e}s, and J.~Casademunt},
  {\em Dynamics of turing patterns under spatiotemporal forcing}, Physical
  Review Letters, 90 (2003), p.~128301.

\bibitem{saez2007rigidity}
{\sc A.~Saez, M.~Ghibaudo, A.~Buguin, P.~Silberzan, and B.~Ladoux}, {\em
  Rigidity-driven growth and migration of epithelial cells on microstructured
  anisotropic substrates}, Proceedings of the National Academy of Sciences, 104
  (2007), pp.~8281--8286.

\bibitem{sanchez2018turing}
{\sc F.~S{\'a}nchez-Gardu{\~n}o, A.~L. Krause, J.~A. Castillo, and P.~Padilla},
  {\em Turing--hopf patterns on growing domains: The torus and the sphere},
  Journal of Theoretical Biology,  (2018).

\bibitem{sarfaraz2018domain}
{\sc W.~Sarfaraz and A.~Madzvamuse}, {\em Domain-dependent stability analysis
  of a reaction--diffusion model on compact circular geometries}, International
  Journal of Bifurcation and Chaos, 28 (2018), p.~1830024.

\bibitem{satnoianu2000turing}
{\sc R.~A. Satnoianu, M.~Menzinger, and P.~K. Maini}, {\em Turing instabilities
  in general systems}, Journal of Mathematical Biology, 41 (2000),
  pp.~493--512.

\bibitem{schnakenberg1979simple}
{\sc J.~Schnakenberg}, {\em Simple chemical reaction systems with limit cycle
  behaviour}, Journal of Theoretical Biology, 81 (1979), pp.~389--400.

\bibitem{seifert1991shape}
{\sc U.~Seifert, K.~Berndl, and R.~Lipowsky}, {\em Shape transformations of
  vesicles: Phase diagram for spontaneous-curvature and bilayer-coupling
  models}, Physical Review A, 44 (1991), p.~1182.

\bibitem{seul1995domain}
{\sc M.~Seul and D.~Andelman}, {\em Domain shapes and patterns: the
  phenomenology of modulated phases}, Science, 267 (1995), pp.~476--483.

\bibitem{shimoda2016time}
{\sc Y.~Shimoda, J.~Kumagai, M.~Anzai, K.~Kabashima, K.~Togashi, Y.~Miura,
  H.~Shirasawa, W.~Sato, Y.~Kumazawa, and Y.~Terada}, {\em Time-lapse
  monitoring reveals that vitrification increases the frequency of contraction
  during the pre-hatching stage in mouse embryos}, Journal of Reproduction and
  Development,  (2016), pp.~2015--150.

\bibitem{shoji2003stripes}
{\sc H.~Shoji, Y.~Iwasa, and S.~Kondo}, {\em Stripes, spots, or reversed spots
  in two-dimensional turing systems}, Journal of Theoretical Biology, 224
  (2003), pp.~339--350.

\bibitem{shoji2007turing}
{\sc H.~Shoji, K.~Yamada, D.~Ueyama, and T.~Ohta}, {\em Turing patterns in
  three dimensions}, Physical Review E, 75 (2007), p.~046212.

\bibitem{sigrist2011symmetric}
{\sc R.~Sigrist and P.~Matthews}, {\em Symmetric spiral patterns on spheres},
  SIAM Journal on Applied Dynamical Systems, 10 (2011), pp.~1177--1211.

\bibitem{sivashinsky1977nonlinear}
{\sc G.~Sivashinsky}, {\em Nonlinear analysis of hydrodynamic instability in
  laminar flames---i. derivation of basic equations}, Acta astronautica, 4
  (1977), pp.~1177--1206.

\bibitem{stoop2015curvature}
{\sc N.~Stoop, R.~Lagrange, D.~Terwagne, P.~M. Reis, and J.~Dunkel}, {\em
  Curvature-induced symmetry breaking determines elastic surface patterns},
  Nature materials, 14 (2015), p.~337.

\bibitem{striegel2009chemically}
{\sc D.~A. Striegel and M.~K. Hurdal}, {\em Chemically based mathematical model
  for development of cerebral cortical folding patterns}, PLoS Computational
  Biology, 5 (2009), p.~e1000524.

\bibitem{swift1977hydrodynamic}
{\sc J.~Swift and P.~C. Hohenberg}, {\em Hydrodynamic fluctuations at the
  convective instability}, Physical Review A, 15 (1977), p.~319.

\bibitem{tan2018polyamide}
{\sc Z.~Tan, S.~Chen, X.~Peng, L.~Zhang, and C.~Gao}, {\em Polyamide membranes
  with nanoscale turing structures for water purification}, Science, 360
  (2018), pp.~518--521.

\bibitem{thery2006cell}
{\sc M.~Thery and M.~Bornens}, {\em Cell shape and cell division}, Current
  opinion in cell biology, 18 (2006), pp.~648--657.

\bibitem{toole2013turing}
{\sc G.~Toole and M.~K. Hurdal}, {\em Turing models of cortical folding on
  exponentially and logistically growing domains}, Computers \& Mathematics
  with Applications, 66 (2013), pp.~1627--1642.

\bibitem{toole2014pattern}
{\sc G.~Toole and M.~K. Hurdal}, {\em Pattern formation in turing systems on
  domains with exponentially growing structures}, Journal of Dynamics and
  Differential Equations, 26 (2014), pp.~315--332.

\bibitem{townsend2013extension}
{\sc A.~Townsend and L.~N. Trefethen}, {\em An extension of chebfun to two
  dimensions}, SIAM Journal on Scientific Computing, 35 (2013), pp.~C495--C518.

\bibitem{toyama2008apoptotic}
{\sc Y.~Toyama, X.~G. Peralta, A.~R. Wells, D.~P. Kiehart, and G.~S. Edwards},
  {\em Apoptotic force and tissue dynamics during drosophila embryogenesis},
  Science, 321 (2008), pp.~1683--1686.

\bibitem{Turing1952}
{\sc A.~M. Turing}, {\em The chemical basis of morphogenesis}, Philosophical
  Transactions of the Royal Society of London. Series B, Biological Sciences,
  237 (1952), pp.~37--72.

\bibitem{ubeda2008root}
{\sc S.~Ubeda-Tom{\'a}s, R.~Swarup, J.~Coates, K.~Swarup, L.~Laplaze, G.~T.
  Beemster, P.~Hedden, R.~Bhalerao, and M.~J. Bennett}, {\em Root growth in
  arabidopsis requires gibberellin/della signalling in the endodermis}, Nature
  Cell Biology, 10 (2008), p.~625.

\bibitem{ueda2012mathematical}
{\sc K.-I. Ueda and Y.~Nishiura}, {\em A mathematical mechanism for
  instabilities in stripe formation on growing domains}, Physica D: Nonlinear
  Phenomena, 241 (2012), pp.~37--59.

\bibitem{van2019diffusive}
{\sc R.~A. Van~Gorder, H.~Kim, and A.~L. Krause}, {\em Diffusive instabilities
  and spatial patterning from the coupling of reaction--diffusion processes
  with stokes flow in complex domains}, Journal of Fluid Mechanics, 877 (2019),
  pp.~759--823.

\bibitem{Varea1997}
{\sc C.~Varea, J.~Arag{\'o}n, and R.~Barrio}, {\em Confined {Turing} patterns
  in growing systems}, Physical Review E, 56 (1997), p.~1250.

\bibitem{varea1999turing}
{\sc C.~Varea, J.~Aragon, and R.~Barrio}, {\em Turing patterns on a sphere},
  Physical Review E, 60 (1999), p.~4588.

\bibitem{vinograd1952criterion}
{\sc R.~Vinograd}, {\em On a criterion of instability in the sense of lyapunov
  of the solutions of a linear system of ordinary differential equations},
  Doklady Akad. Nauk SSSR (N.S.), 84 (1952), pp.~201--204.

\bibitem{wang2006resonant}
{\sc H.~Wang, K.~Zhang, and Q.~Ouyang}, {\em Resonant-pattern formation induced
  by additive noise in periodically forced reaction-diffusion systems},
  Physical Review E, 74 (2006), p.~036210.

\bibitem{wang2015three}
{\sc Q.~Wang and X.~Zhao}, {\em A three-dimensional phase diagram of
  growth-induced surface instabilities}, Scientific Reports, 5 (2015), p.~8887.

\bibitem{wu1974note}
{\sc M.~Wu}, {\em A note on stability of linear time-varying systems}, IEEE
  transactions on Automatic Control, 19 (1974), pp.~162--162.

\end{thebibliography}
\end{document}